\documentclass[reqno,a4paper,11pt]{amsart}
\usepackage[english]{babel}
\usepackage{color}
\usepackage{soul}

\usepackage{amssymb,accents}
\usepackage{mathrsfs, mathtools}

\usepackage[all]{xy} 
\input xy 
\xyoption{all}
\xyoption{2cell} 
\xyoption{v2}
\SelectTips{cm}{11} 

\usepackage{subdepth}

\usepackage{hyperref}

\DeclareMathOperator{\Aut}{Aut}

\DeclareMathOperator{\Mod}{Mod}
\DeclareMathOperator{\Fred}{Fred}

\DeclareMathOperator{\id}{id}
\DeclareMathOperator{\hol}{hol}

\DeclareMathOperator{\Ad}{Ad}

\DeclareMathOperator{\op}{op}

\DeclareMathOperator{\cC}{\mathscr{C}} 
 
\DeclareMathOperator{\DD}{DD}

\DeclareMathOperator{\Rmod}{RMod}
\DeclareMathOperator{\RVect}{RVect}
\DeclareMathOperator{\RCyc}{RCyc}
\DeclareMathOperator{\Vect}{Vect}

\DeclareMathOperator{\pt}{pt}

\theoremstyle{plain}
\newtheorem{theorem}[equation]{Theorem}
\newtheorem{corollary}[equation]{Corollary}
\newtheorem{lemma}[equation]{Lemma}
\newtheorem{proposition}[equation]{Proposition}

\theoremstyle{definition}
\newtheorem{definition}[equation]{Definition}

\theoremstyle{remark}
\newtheorem{remark}[equation]{Remark}
\newtheorem{example}[equation]{Example}

\numberwithin{equation}{section}
\numberwithin{figure}{section}

\renewcommand{\cH}{{\mathcal H}}

\newcommand{\cS}{{\mathcal S}}
\newcommand{\cA}{{\mathcal A}}

\newcommand{\cK}{{\mathcal K}}
\newcommand{\cU}{{\mathcal U}}

\newcommand{\cP}{{\mathcal P}}

\newcommand{\cB}{{\mathcal B}}
\newcommand{\cZ}{{\mathcal Z}}
\renewcommand{\cR}{{\mathcal R}}

\newcommand{\CC}{{\mathbb C}}
\newcommand{\RR}{{\mathbb R}}
\newcommand{\ZZ}{{\mathbb Z}}

\renewcommand{\a}{\alpha}
\renewcommand{\b}{\beta}
\renewcommand{\c}{\gamma}
\renewcommand{\d}{\delta}

\newcommand{\bg}{{\rm bg}}

\newcommand{\KR}{K\hspace{-1pt}R}
\newcommand{\KO}{K\hspace{-1pt}O}
\newcommand{\KK}{K\hspace{-1pt}K}
\newcommand{\KKR}{K\hspace{-1pt}K\hspace{-1pt}R}
\newcommand{\KSp}{K\hspace{-1pt}Sp}

\newcommand{\WR}{W\hspace{-1pt}R}

\makeatletter
\newlength{\@thlabel@width}%
\newcommand{\thmenumhspace}{\settowidth{\@thlabel@width}{(1)}\sbox{\@labels}{\unhbox\@labels\hspace{\dimexpr-\leftmargin+\labelsep+\@thlabel@width-\itemindent}}}
\makeatother

\begin{document}

\begin{flushright}

\end{flushright}

\vskip 1cm

\title[Real bundle gerbes, orientifolds and twisted $\KR$-homology]{Real
  bundle gerbes, orientifolds \\[1mm] and twisted $\boldsymbol{\KR}$-homology}

  \author[P. Hekmati]{Pedram Hekmati}
  \address[Pedram Hekmati]
  {Department of Mathematics\\ 
  University of Auckland\\ 
  Auckland 1010 \\ 
  New Zealand}
  \email{p.hekmati@auckland.ac.nz}

  \author[M.K. Murray]{Michael~K.~Murray}
  \address[Michael K.~Murray]
  {School of Mathematical Sciences\\
  University of Adelaide\\
  Adelaide, SA 5005 \\
  Australia}
  \email{michael.murray@adelaide.edu.au}

  \author[R.J. Szabo]{Richard~J.~Szabo}
  \address[Richard J.~Szabo]
  {Department of Mathematics, Maxwell Institute for Mathematical Sciences and The Higgs Centre for Theoretical Physics\\
  Heriot-Watt University\\
  Edinburgh, EH14 4AS \\
  United Kingdom}
\address{
Centro de Matem\'atica, Computac\~ao e
Cognic\~ao\\
Universidade de Federal do ABC\\
Santo Andr\'e, SP\\
Brazil}
  \email{r.j.szabo@hw.ac.uk}
  
  \author[R.F. Vozzo]{Raymond~F.~Vozzo}
\address[Raymond F.~Vozzo]
{School of Mathematical Sciences\\
University of Adelaide\\
  Adelaide, SA 5005 \\
  Australia}
\email{raymond.vozzo@adelaide.edu.au} 

\thanks{The authors acknowledge support under the  Australian
Research Council's {\sl Discovery Projects} funding scheme (project numbers DP120100106 and DP130102578), the Consolidated Grant ST/L000334/1 
from the UK Science and Technology Facilities Council, the Action MP1405 QSPACE from the European Cooperation in Science and Technology
(COST), and the Visiting Researcher Program
Grant 2016/04341--5 from the Fundac\~ao de Amparo \'a Pesquisa do
Estado de S\~ao Paulo (FAPESP, Brazil).  
Report no.: \ EMPG--16--15\\
We thank Alan Carey, Jonathan Rosenberg and Bogdan Stefanski
for helpful discussions.}

\subjclass[2010]{}

\begin{abstract} 

We consider Real bundle gerbes on manifolds equipped with an involution and prove that they are classified by their Real Dixmier--Douady class in Grothendieck's equivariant sheaf cohomology. We show that the Grothendieck group of Real bundle gerbe modules is isomorphic to twisted $\KR$-theory for a torsion Real Dixmier--Douady class. Using these modules as building blocks, we introduce geometric cycles for twisted  $\KR$-homology and prove that they generate a real-oriented generalised homology theory dual to twisted $\KR$-theory for Real closed manifolds, and more generally for Real finite CW-complexes,  for any Real Dixmier--Douady class. This is achieved by defining an explicit natural transformation to analytic twisted $\KR$-homology and proving that it is an isomorphism. Our model both refines and extends previous results by Wang \cite{Wang} and Baum--Carey--Wang \cite{BCW} to the Real setting. Our constructions further provide a new framework for the classification of orientifolds in string theory, providing precise conditions for orientifold lifts of $H$-fluxes and for orientifold projections of open string states.
\end{abstract}
\maketitle

\tableofcontents

\bigskip

\section{Introduction and summary}

In \cite{Ati} Atiyah introduces the notion of $\KR$-theory for a space $M$ with an involution $\tau \colon M \to M$ as a common generalisation of real and complex $K$-theory. This is defined on the semi-group of complex vector bundles which are `Real' in the sense that the involution $\tau$ lifts to an anti-linear involution on the total space. In this paper we provide a definition of twisted $\KR$-theory, as well as its dual homology theory, and describe some new approaches to the construction of orientifolds of Type II string theory, using modules for a certain kind of bundle gerbe.

To motivate the mathematical ideas that we use, note that an involution $\tau$ acting on a space $M$ is equivalent to an action of $\ZZ_2$ where $\tau$ defines the action of the non-trivial element in $\ZZ_2$. There is an induced action of $\ZZ_2$ on the space of functions
$f \colon M \to \CC$ given by $\tau(f)(m) = f ( \tau(m))$. As a result this space has two distinguished subsets: the `Real' functions which satisfy $\tau(f) = \bar f$ and the `invariant' functions which satisfy $\tau(f) = f$. Notice that Real does not mean that the function is real-valued unless $\tau$ acts trivially on $M$.  

When we replace functions by more complicated geometric objects such as $U(1)$-bundles $L \to M$, then 
the definitions of Real and invariant also involve a choice of isomorphism $\tau^{-1}(L) \simeq L^*$ or $\tau^{-1}(L) \simeq L$
which, in an appropriate sense, squares to the identity. In the latter case we will call the line bundle $L$ `equivariant' rather than 
invariant because it corresponds exactly to a lift of the $\ZZ_2$-action on $M$ to $L$.

When we pass to bundle gerbes, we
have to also deal with the fact that there are two kinds of isomorphism for bundle gerbes, so it is possible to define the $\ZZ_2$-action 
to be either by isomorphisms or by stable isomorphisms. The former leads to the notion of Real bundle gerbes \cite{Mou} and the latter to the notion of 
Jandl bundle gerbe \cite{SchSch}. In this paper, we elucidate the relation between these two kinds of gerbes and show that both notions, equipped with the appropriate idea of stable isomorphism, are sufficient to capture Grothendieck's equivariant sheaf cohomology group $H^2(M; \ZZ_2, \overline{\cU(1)})$~\cite{Gro} through a Real version of the Dixmier--Douady class. 

Once these preliminaries are in place, it is relatively straightforward to extend the results of~\cite{BouCarMat, CareyWang2} to Real bundle gerbes. In particular, we define Real bundle gerbe modules
and prove that they model twisted $\KR$-theory for a torsion Real Dixmier--Douady class by
establishing a Real version of the Serre--Grothendieck theorem.
We then  introduce a geometric model for twisted $\KR$-homology using Real bundle gerbe modules.
This is where the latter come into their own, since the geometric cycles work for arbitrary
twisting classes. A merit of our model is that it uses actual Real bundle gerbe modules and not merely
twisted $\KR$-theory data in the definition of cycles. 
We define an assembly map to analytic twisted $\KR$-homology and prove that for Real closed manifolds, and more generally for Real finite CW-complexes, it is an isomorphism by constructing an explicit inverse. Consequently, the Real bundle gerbe cycles define
a real-oriented generalised homology theory dual to twisted $\KR$-theory.

Twisted $\KR$-homology is a primary theory in the sense that it subsumes complex, real and quaternionic
 $K$-homology as special cases. For a twisting class $[H]$ we recover the construction of complex twisted $K$-homology by Wang for compact manifolds from~\cite{Wang} and the Baum--Carey--Wang construction for finite
CW-complexes from~\cite{BCW} via
$\KR(M\amalg M, [H] \amalg - [H]) = K(M,[H])$,
where the involution acts by exchanging the two copies of $M$ and sends $([H] \amalg - [H])$  to  $(-[H]
\amalg [H])$. Moreover, our model subsumes the Deeley--Goffeng model from \cite{DeeleyGoffeng} which uses closed spin$^c$ $PU(n)$-manifolds.  In the complex
setting our Real bundle gerbe cycles are closely related to  their projective $K$-cycles, but unlike in~\cite{DeeleyGoffeng} where they use $PU(n)$-equivariant maps, our
proof that the assembly map is an isomorphism works for arbitrary twistings. When the involution $\tau$ is trivial, we obtain a geometric model for twisted $\KO$-homology, $\KR(M,[H]) \cong \KO(M,[H])$. We note that this isomorphism holds under the condition that the {\em sign choice} associated to the twisting class $[H]$ is positive. On the other hand, when   the sign choice of $[H]$ is negative and the corresponding  complex Dixmier--Douady class vanishes, we obtain a geometric model for untwisted quaternionic $K$-homology, $\KR(M,[H]) \cong \KSp(M)$ (see Example \ref{discretetorsion}). We address the question of sign choices, as well as connective structures on Real bundle gerbes, in detail in the sequel \cite{HMSV}.

Spaces with involutions  give an efficient way to construct new string backgrounds, which in the presence of fluxes are important for model building in string theory; in this setting the pair $(M,\tau)$ is called an `orientifold'.
Part of the motivation behind this work is to better sharpen the current understanding of orientifold constructions in string theory in the presence of background $H$-flux, as Ramond-Ramond charges and currents in these backgrounds are classified by twisted (differential) $\KR$-theory~\cite{Witten,BGS,BraunStef,DFM,GaoHori}. The mathematical formalism that we develop in this paper provides a new framework in which to investigate various features of orientifolds. In particular, there are four problems that can be tackled using our perspective.

Firstly, the Dirac quantization condition on the $B$-field must be implemented by locating its quantum flux in a suitable cohomology group, so that the usual class $[H]\in H^3(M,\ZZ)$ of the $H$-flux must be equivariant in an appropriate sense. In this paper we give for the first time a necessary and
sufficient condition (including torsion) for an $H$-flux to lift to an orientifold
$H$-flux via a long exact sequence in Grothendieck's equivariant sheaf cohomology. Secondly, the orientifold projection conditions on open string states are known only in some simple examples; in the following we give a general definition of Real bundle gerbe D-branes
appropriate to an orientifold background, and in particular our construction of twisted $\KR$-homology precisely defines the orientifold projections of open string states. Thirdly, in a given situation one may be interested in D-branes not only on top of an orientifold plane (O-plane); our homological classification naturally accounts for these open string states as well and provides new consistency conditions for D-branes in orientifolds. Finally, in Type~II orientifolds, D-branes on top of an O-plane can have either an $SO(n)$ or $Sp(n)$ gauge symmetry depending on the choice of orientifold action; this defines the `type' of an O-plane. Conditions for the allowed distributions of O-plane types for a given involution $\tau$ are discussed more systematically in~\cite{HMSV}. For some recent progress in this direction, see~\cite{DMR1,DMR2,DFM,GaoHori}.

In summary the paper proceeds as follows. In Section~\ref{sec:bg} we review the theory of bundle gerbes and bundle gerbe modules, and explain how it was used in \cite{BouCarMat, CareyWang2} to define twisted $K$-theory. Up to stable isomorphism, bundle gerbes over $M$ are classified precisely by their Dixmier--Douady class in $H^2( M, \cU(1)) = H^3(M, \ZZ)$. In the case of Real bundle gerbes there is a corresponding Real Dixmier--Douady class which lives in Grothendieck's equivariant sheaf cohomology group $H^2(M; \ZZ_2, \overline{\cU(1)})$~\cite{Gro}, and we develop the necessary parts of this theory in Section~\ref{sec:HR}. As an 
introduction to the notion of Real bundle gerbes, we first consider Real line bundles in Section~\ref{sec:RealLine}. In Section~\ref{sec:Real bundle gerbes} we 
introduce the definition of Real bundle gerbes, and briefly discuss their relationship with the apparently weaker notion of Jandl bundle gerbes. 
The corresponding notion of Real bundle gerbe module is introduced in Section~\ref{sec:KR} and related to twisted $\KR$-theory. In Section~\ref{sec:orientifolds} we describe some applications of our formalism to the orientifold construction in string theory, and introduce the notion of Real bundle gerbe D-brane which serves as an impetus for the definition of geometric twisted $\KR$-homology that we give in Section~\ref{sec:Geometric cycles}. The paper concludes with the construction of the Real assembly map to analytic twisted $\KR$-homology and a proof that it gives an isomorphism. 

\section{Bundle gerbes and their modules}\label{sec:bg}

In this section we will briefly review the various facts about bundle
gerbes and bundle gerbe $K$-theory that will be relevant for us in
later sections; more details can be found in \cite{BouCarMat, Mur}. The reader familiar with bundle gerbes and their modules can safely skip this section.

Let $M$ be a manifold and $Y \xrightarrow{\pi} M$ a surjective submersion. We denote by $Y^{[p]}$ the
$p$-fold fibre product of $Y$ with itself, that is $Y^{[p]} = Y \times_M Y
\times_M \cdots \times_M Y$. This is a simplicial space whose face
maps are given by the projections $\pi_i \colon Y^{[p]} \to Y^{[p-1]}$
which omit the $i$-th factor. A \emph{bundle gerbe} $(P, Y)$ (or
simply $P$ when $Y$ is understood) over $M$ is defined by a principal
$U(1)$-bundle (or a hermitian line bundle) $P \to Y^{[2]}$ together with
a bundle gerbe multiplication given by an isomorphism of bundles
$\pi_3^{-1}(P) \otimes \pi_1^{-1}(P) \to \pi_2^{-1}(P)$ over $Y^{[3]}$, which is
associative over $Y^{[4]}$. On fibres the multiplication looks like
$P_{(y_1, y_2)} \otimes P_{(y_2, y_3)} \to  P_{(y_1, y_3)}$ for $(y_1,
y_2, y_3) \in Y^{[3]}$. This implies that if $(y_1, y_2, y_3, y_4) \in Y^{[4]}$, then 
\begin{gather}
\label{eq:useful}
\begin{aligned}
P_{(y_1, y_2)}  \otimes P_{(y_3, y_4)} &\simeq  P_{(y_1, y_2)}  \otimes P_{(y_2, y_3)} \otimes P_{(y_3, y_4)} \otimes P_{(y_3, y_2)} \\[4pt]
&\simeq P_{(y_1, y_4)}   \otimes P_{(y_3, y_2)} \ .
\end{aligned}
\end{gather}

We can multiply two bundle gerbes over $M$ together. Namely, we define $(P, Y)
\otimes (Q, X):= (P \otimes Q, Y \times_M X)$, where here $P$ and $Q$ are pulled back to $(Y \times_M X)^{[2]}$ by the obvious maps to 
$Y^{[2]}$ and $X^{[2]}$.

The \emph{dual} of $(P, Y)$ is the bundle gerbe $(P^*, Y)$, where by
$P^*$ we mean the $U(1)$-bundle which is $P$ with the action of $U(1)$
given by $p \cdot z = p \, \bar z = p \, z^{-1}$, that is, as a space $P^*
= P$, but with the conjugate $U(1)$-action.

Given a map $f \colon N \to M$ and a bundle gerbe $(P, Y)$ over $M$, we can
pull back the surjective submersion $Y\to M$ to a surjective submersion
$f^{-1}(Y) \to N$ and the bundle gerbe $(P,Y)$ to a bundle gerbe $f^{-1}(P,
Y) := \big( \big(f^{[2]} \big)^{-1}(P), f^{-1}(Y) \big)$ over $N$, where
$f^{[2]} \colon f^{-1}\big(Y^{[2]} \big) \to Y^{[2]}$ is the map induced by $f \colon f^{-1}(Y) \to Y$.

A bundle gerbe $(P, Y)$ over $M$ defines a class in $H^3(M, \ZZ)$,
called the \emph{Dixmier--Douady class} of $P$, as follows. Let $\cU =
\{U_\a \}_{\a\in I}$ be a good cover of $M$ with sections $s_\a \colon
U_\a \to Y$, where as usual we write
$U_{\a_0\cdots\a_p}:=U_{\a_0}\cap\cdots\cap U_{\a_p}$. 
  On double overlaps $U_{\a\b}$ they define sections $(s_\a, s_\b)$  of $Y^{[2]}$ by $m \mapsto (s_\a(m), s_\b(m))$. Choose sections $\sigma_{\a\b}$ of the pullback bundle $(s_\a, s_\b)^{-1}(P) \to U_{\a\b}$. Using the bundle gerbe multiplication we have
$$
\sigma_{\a\b} \, \sigma_{\b\gamma} = \sigma_{\a\c} \, g_{\a\b\c} \ ,
$$
for some maps $g_{\a\b\c}\colon U_{\a\b\c} \to U(1)$ on triple
overlaps which satisfy the cocycle condition and hence define a class
in $  H^2( M, \cU(1)) = H^3(M, \ZZ)$. We call this element the
Dixmier--Douady class of $P$ and denote it by $\DD(P)$. Conversely, any class $[H]\in H^3(M,\ZZ)$ defines a bundle gerbe $(P,Y)$ over $M$ with $\DD(P)=[H]$.

An \emph{isomorphism} between two bundle gerbes $(P, Y)$ and $(Q, X)$
over $M$ is a pair of maps $(\hat f,  f)$ where $f \colon Y \to X$ is
an isomorphism that covers the identity on $M$, and $\hat f \colon P
\to Q$ is a map of $U(1)$-bundles that covers the induced map $f^{[2]} \colon Y^{[2]} \to X^{[2]}$ and commutes with the bundle gerbe product. Isomorphism is too strong to be the right notion of equivalence for bundle gerbes, since there are many non-isomorphic bundle gerbes with the same Dixmier--Douady class. The correct notion of equivalence is stable isomorphism \cite{MurSte2} defined below, which has the property that two bundle gerbes are stably isomorphic if and only if they have the same Dixmier--Douady class.

We say that a bundle gerbe $(P, Y)$ is \emph{trivial} if there exists a $U(1)$-bundle $L \to Y$ such that $P$ is isomorphic to $ \d L := \pi_1^{-1}(L) \otimes \pi_2^{-1}(L)^*$ with the canonical multiplication $(\d L)_{(y_1, y_2)} \otimes (\d L)_{(y_2, y_3)} = L_{y_1}^* \otimes L_{y_2} \otimes L_{y_2}^* \otimes L_{y_3} = L_{y_1}^* \otimes L_{y_3} = (\d L)_{(y_1, y_3)}$. A choice of $L$ and an isomorphism $P\simeq \d L$ is called a \emph{trivialisation}; any two trivialisations differ by the pullback of a line bundle on $M$. 
 The Dixmier--Douady class is precisely the obstruction to the bundle gerbe being trivial.
 Two bundle gerbes are \emph{stably isomorphic} if $ Q\otimes P^*$ is
 trivial, and a \emph{stable isomorphism} $P \to Q$ is a choice of trivialisation of $Q \otimes P^*$. Explicitly, if $(P, Y)$ and $(Q, X)$ are stably isomorphic bundle gerbes over $M$, then a stable isomorphism $P \to Q$ is a bundle $R \to Y \times_M X$ such that
\begin{equation}\label{E:stable iso}
P_{(y_1, y_2)}\otimes R_{(y_2, x_2)} \simeq  R_{(y_1, x_1)} \otimes
Q_{(x_1, x_2)} \ .
\end{equation}
If $f \colon (P, Y) \to (Q, X)$ is an isomorphism of bundle gerbes then using \eqref{eq:useful} we have 
\begin{align*}
(Q \otimes P^*)_{( x_1, x_2, y_1, y_2)} &= Q_{(x_1, x_2)} \otimes P_{(y_1, y_2)}^*  \\[4pt]
& = Q_{(x_1, x_2)} \otimes Q_{(f(y_1), f(y_2))}^* \\[4pt]
&= Q_{(x_1, f(y_1))} \otimes Q_{(x_2, f(y_2))}^* \ .
\end{align*}
Hence there is an induced stable isomorphism given by $Q \otimes P^* = \d L$, where $L \to X \times_M Y$ is given by $L_{(x, y)} = Q_{(x, f(y))}^*$.

In the case that two bundle gerbes are defined over the same surjective submersion, the situation is slightly simpler. If $(P, Y)$ and $(Q, Y)$ are bundle gerbes, a stable isomorphism is a bundle $R \to Y^{[2]}$ 
and the isomorphism (\ref{E:stable iso}) becomes 
$$
P_{(y_1, y_2)}\otimes R_{(y_2, y_2')} \simeq  R_{(y_1, y_1')} \otimes
Q_{(y_1', y_2')} \ ,
$$
for $y_1, y_2, y_1', y_2'$ all in the same fibre of $Y$. Since we can include $Y$ into $Y^{[2]}$ as the diagonal, we can restrict $Q \otimes P^*$ to $Y$ and this  induces a stable
isomorphism $(Q \otimes P^*, Y) \to (Q \otimes P^*, Y^{[2]})$. Hence $(Q \otimes P^*, Y) $ is trivial
if and only if $(Q \otimes P^*, Y^{[2]})$ is trivial. From the theory
of bundle gerbe modules and the fact that a trivialisation is a bundle
gerbe module of rank one (see below), it follows that there is a bijective correspondence between trivialisations 
of $(Q \otimes P^*, Y) $ and trivialisations of $(Q \otimes P^*, Y^{[2]})$. Thus we can regard a stable isomorphism
$R \colon (P, Y) \to (Q, Y)$ as a bundle $R \to Y$ together with isomorphisms
\begin{eqnarray}\label{eq:stableisosame}
P_{(y_1, y_2)} \otimes R_{y_2} \simeq  R_{y_1} \otimes Q_{(y_1, y_2)}
\ .
\end{eqnarray}

Given stable isomorphisms $R\colon(P, Y) \to (Q, X)$ and $S \colon (Q, X) \to (T, Z)$ there is a general theory of how to compose them. In the case $Y = X = Z$ it reduces to the following.  Assume we have \eqref{eq:stableisosame} and 
$$
Q_{(y_1, y_2)}\otimes S_{y_2} \simeq  S_{y_1} \otimes T_{(y_1, y_2)} \
.
$$
Then we induce maps 
\begin{align*}
P_{(y_1, y_2)}\otimes(R_{y_2}\otimes S_{y_2})  &\simeq  R_{y_1} \otimes Q_{(y_1, y_2)}\otimes  S_{y_2} \\[4pt]
                               & \simeq (R_{y_1}\otimes  S_{y_1}) \otimes T_{(y_1, y_2)}
                               \end{align*}
which define the product.

Any stable isomorphism \eqref{eq:stableisosame}
induces an inverse $Q \to P$,
$$
Q_{(y_1, y_2)}\otimes R_{y_2}^* \simeq  R_{y_1}^* \otimes P_{(y_1,
  y_2)} \ ,
$$
and a dual $P^* \to Q^*$,
$$
P^*_{(y_1, y_2)} \otimes R^*_{y_2} \simeq  R^*_{y_1} \otimes
Q^*_{(y_1, y_2)} \ .
$$

Given a map $\tau \colon M \to M$ and a stable isomorphism $R \colon
(P, Y) \to (Q, Y)$ there is a stable 
isomorphism
$\tau^{-1}(R) \colon \tau^{-1}(P, Y) \to \tau^{-1}(Q, Y)$. 

If $(P, Y)$ is a bundle gerbe, then  a \emph{bundle gerbe module} is a vector bundle $E \to Y$ with a family of bundle
maps 
$$
P_{(y_1, y_2)} \otimes E_{y_2} \simeq  E_{y_1} 
$$
satisfying the natural associativity condition that on any triple $(y_1, y_2, y_3) \in Y^{[3]}$ 
the two maps 
$$
 P_{(y_1, y_2)} \otimes P_{(y_2, y_3)}  \otimes E_{y_3}  \
 \longrightarrow \ P_{(y_1, y_3)} \otimes E_{y_3} \ \longrightarrow \ E_{y_1}
$$
and 
$$
 P_{(y_1, y_2)} \otimes  P_{(y_2, y_3)} \otimes E_{y_3}  \
 \longrightarrow \ P_{(y_1, y_2)} \otimes E_{y_2} \ \longrightarrow \ E_{y_1}
$$
are equal. We denote by $\Mod(P,Y)$ the semi-group of bundle gerbe modules under direct sum and by $K_{\bg}(M,P)$ the corresponding 
Grothendieck group which we call the {\em bundle gerbe $K$-theory group} of $(P, Y)$.

It is shown in \cite[Proposition 4.3]{BouCarMat}  that if $(P, Y)$ and $(Q, X)$ are bundle
gerbes over $M$ then any stable isomorphism  $R \colon(P, Y) \to (Q, X)$  induces a semi-group isomorphism $\Mod(P,Y) \to \Mod(Q, X)$ 
and thus an isomorphism $K_{\bg}(M,P) \simeq K_{\bg}(M,Q)$.  There is an important subtlety that needs noting.  Different stable isomorphisms between bundle gerbes can give rise to different isomorphisms on twisted $K$-theory. So while $K_{\bg}(M,P)$ and  $ K_{\bg}(M,Q)$ are isomorphic if $\DD(P) = \DD(Q)$
the actual isomorphism is not determined until a stable isomorphism is chosen.  It is a common abuse of notation 
however to write $K_{\bg}(M, [H])$ to mean a group in the isomorphism class of $K_{\bg}(M,P)$ for some bundle
gerbe $(P, Y)$ with $\DD(P) = [H] \in H^3(M, \ZZ)$. 
 In \cite{CareyWang2} it was shown that for any torsion class $[H] \in H^3(M, \ZZ)$, the group $K_{\bg}(M, [H])$  is isomorphic to the twisted $K$-theory $K(M, [H])$.

\begin{example}
If $(P,Y)$ is trivial so that $P = \delta K$, then the bundle gerbe module action $P_{(y_1,y_2)} \otimes E_{y_2} \simeq E_{y_1}$ implies $K_{y_1}^* \otimes K_{y_2} \otimes E_{y_2} \simeq E_{y_1}$ and so
$$
K_{y_2}\otimes E_{y_2} \simeq K_{y_1}\otimes E_{y_1} \ ,
$$ 
which are descent data for the bundle $K \otimes E \to Y$. Conversely,
if $F$ is a bundle on $M$ then $\delta K$ acts on $K^* \otimes \pi^{-1}(F)$ and
so it defines a module. This gives an isomorphism from the semi-group
of bundle gerbe
modules $\Mod(\delta K,Y)$ to the semi-group of vector bundles $\Vect(M)$, which implies that the bundle gerbe $K$-theory of a trivial bundle gerbe on $M$ is isomorphic to the $K$-theory of $M$.
\end{example}

\section{Grothendieck's equivariant sheaf cohomology}\label{sec:HR}

In his famous Tohoku paper \cite{Gro}, Grothendieck introduced a
cohomology theory for sheaves with group actions.  We will be
concerned  with the case that the group is the cyclic group $\ZZ_2$.

Let $M$ be a manifold with an involution $\tau \colon M \to M$; this
is of course the same thing as an action of $\ZZ_2$ on $M$.  The pair
$(M,\tau)$ is called a \emph{Real manifold} and we will simply write
$M$ when there is no risk of confusion. Real manifolds are objects in
a category whose morphisms $f:(M,\tau)\to (M',\tau'\,)$ are
equivariant smooth maps, that is $f\circ \tau=\tau'\circ f$.

Let $\cS$
be a sheaf of abelian groups with an action of $\ZZ_2$ covering that
on $M$ \cite{Gro}. Again we only need to describe the action of the
non-trivial element of $\ZZ_2$ which must be involutive and is also denoted $\tau$. For any such $\ZZ_2$-sheaf denote by $\Gamma^{\ZZ_2}_M(\cS)$ the space of $\ZZ_2$-invariant sections of $\cS$. Grothendieck denotes the right derived
functors of $\Gamma^{\ZZ_2}_M$ applied to $\cS$ by $H^p(M; \ZZ_2, \cS)$. 
 We are interested primarily in the case when $\cS$ is the sheaf  of smooth
functions taking values in the group $U(1)$ which we denote by $\cU(1)$. We will adopt this same notation when we give this sheaf the 
trivial $\ZZ_2$ action and denote it $\overline{\cU(1)}$ when we give it  the conjugation action 
 $\tau(f) = \bar f \circ \tau $.

We want to calculate this cohomology via a \v{C}ech construction  using \cite[Section 5.5]{Gro}.
Following \cite{Mou} we say that an open cover $\cU = \{ U_\a \}_{\a
  \in I}$ of $M$ is \emph{Real} if $U_\a \in \cU$ 
implies that $\tau(U_\a) \in \cU$ and the indexing set $I$ has an involution denoted $\a \mapsto \bar \a$ such that $\tau(U_\a) = U_{\bar\a}$. It is always possible to choose a good cover with the property that the involution on $I$ has no 
fixed points. For this, pick a metric on $M$ and make it $\tau$-invariant by averaging. Then the image of any geodesically convex set under $\tau$ is again a geodesically convex set, so a  family of geodesically convex subsets and their $\tau$-translates provide a good cover of $M$. We can further extend the indexing set $I$ so that $\a $ and $\bar \a$ are never the same index. This can be done  by replacing $I$ with $I \times \ZZ_2$ so that 
$\overline{(\a, \pm\, 1)} = (\a, \mp\, 1)$ and letting $U_{(\a, 1)} =
U_\a$ and $U_{(\a, -1)} = U_{\bar \a}$.   We will not make this
replacement explicit but simply assume that $I$ has the required
property.  For later use we note the trivial fact that if $I$ is a finite set with an involution  without fixed points, then $| I |$ is even, and $I$ is the disjoint union of two 
subsets $I_+$ and $I_-$ that are interchanged by the involution.

Given a $\ZZ_2$-sheaf $\cS$, we can introduce the space $C^p(\cU;
\ZZ_2, \cS)$ of all cochains $\sigma$ which are invariant under
$\tau$, that is 
$$
 \sigma_{\a_0 \cdots \a_p} = \tau(\sigma_{\bar\a_0 \cdots \bar\a_p} \circ \tau).
$$
The associated \v{C}ech cohomology groups are defined in the usual way as the inductive limit over refinements of Real open covers. For 
the particular cases of the sheaves $\cU(1)$ and $\overline{\cU(1)}$ it follows 
from \cite[Corollary 1, p.~209]{Gro} that the limit is in fact achieved for a Real good
cover with free action on its indexing set.

Explicitly the two cases of interest are as follows. 
Given a map
 $$
g_{\bar\a_0 \cdots \bar\a_p}   \colon U_{\bar\a_0 \cdots \bar\a_p} \longrightarrow U(1)
 $$
 then
 $$
g_{\bar\a_0 \cdots \bar\a_p} \circ \tau \colon U_{\a_0 \cdots \a_p}
\longrightarrow U(1) \ ,
$$
and we can define an involution $\tau^*$ on $C^p(\cU, \cU(1))$ by $\tau^*(g)_{\a_0 \cdots \a_p} = g_{\bar\a_0 \cdots \bar\a_p} \circ \tau$ 
for $g \in C^p(\cU, \cU(1))$. We are interested in two natural subcomplexes of the ordinary \v{C}ech complex $C^p(\cU, \cU(1))$ defined by how cochains behave under $\tau^*$. 
Firstly there is $C^p(\cU; \ZZ_2,  \overline{\cU(1)})$, the subgroup of {\em Real cochains} which satisfy $\tau^*(g) = \bar g$ or 
$$
\bar g_{\a_0 \cdots \a_p} = g_{\bar\a_0 \cdots \bar\a_p} \circ \tau \ .
$$
Secondly there is $C^p(\cU; \ZZ_2, \cU(1))$, the subgroup of  {\em invariant cochains} which satisfy $\tau^*(g) = g$ or
$$
g_{\a_0 \cdots \a_p} = g_{\bar\a_0 \cdots \bar\a_p} \circ \tau \ .
$$

The groups $H^p(M; \ZZ_2, \overline{\cU(1)})$, $H^p( M, \cU(1))$ and $H^p(M; \ZZ_2, \cU(1))$ are related by a long exact sequence which we now describe. 
If $\cS $ is any sheaf of abelian groups on $M$ then $\cS \oplus
\tau^{-1}(\cS)$ is a $\ZZ_2$-sheaf and $H^p(M; \ZZ_2, \cS \oplus
\tau^{-1}(\cS))    = H^p(M, \cS)$.  Then there is a short exact sequence of $\ZZ_2$-sheaves
$$
\label{eq:short-exact-sheaves}
 0 \ \longrightarrow \ \overline{\cU(1)} \ \longrightarrow \ \cU(1) \oplus
 \tau^{-1}\big(\cU(1) \big) \ \longrightarrow \ \cU(1) \ \longrightarrow \ 0
$$
where the maps are $f \mapsto (f , \bar f \circ \tau)$ and $(g, h)
 \mapsto g \, ( h \circ \tau)$. 
Exactness is straighforward.

It follows that there is a long exact sequence in cohomology and we are particularly interested in the lowest degree groups
\begin{equation}
\label{eq:beginlongsequence}
\xymatrix@R=2ex{0 \ 
\ar[r] & \   H^0(M; \ZZ_2, \overline{\cU(1)}) \ \ar[r] & \   H^0(M, \cU(1)) \  \ar[r]^{1 \times \tau^*} & \   H^0(M; \ZZ_2, \cU(1)) \ \ar`r[d]`[lll]`[ddlll] `[ddll][ddll]&\\
&&&&\\
 & \   H^1(M; \ZZ_2, \overline{\cU(1)}) \ \ar[r] & \   H^1(M, \cU(1)) \ \ar[r]^{1 \times \tau^*} & \   H^1(M; \ZZ_2, \cU(1)) \ \ar`r[d]`[lll]`[ddlll] `[ddll][ddll]&\\
&&&&\\
 & \   H^2(M; \ZZ_2, \overline{\cU(1)}) \ \ar[r] & \   H^2(M, \cU(1)) \ \ar[r]^{1 \times \tau^*} & \   H^2(M; \ZZ_2, \cU(1)) \ \ar[r]  & \ \cdots \ .
}
\end{equation}
In Sections~\ref{sec:RealLine} and~\ref{sec:Real bundle gerbes} we will provide geometric interpretations of the groups in this sequence and the homomorphisms between them. In particular, $H^1(M; \ZZ_2, \overline{\cU(1)})$ will be shown to correspond to Real isomorphism classes of Real line bundles and $H^2(M; \ZZ_2, \overline{\cU(1)})$ to Real stable isomorphism classes of Real bundle gerbes. The homomorphism 
$$
 H^p(M; \ZZ_2, \overline{\cU(1)})  \longrightarrow   H^p(M, \cU(1))
 $$
corresponds to forgetting the Real structures involved. The long exact sequence is a tool for addressing important questions surrounding this forgetful map such as when a 
line bundle or bundle gerbe admits a Real structure, or can be ``lifted'' to an equivalent Real object, how many such lifts there are and which Real objects are trivial after we forget their Real structure.

We conclude by elucidating the relation to ordinary equivariant cohomology. Denote by $\cZ$ and $\cR$ the sheaf of functions with values in $\ZZ$ and $\RR$ respectively, by the same notation the correspondings $\ZZ_2$-sheaves with trivial action of $\tau$  and by $\overline \cZ$ and $\overline \cR$ the corresponding $\ZZ_2$-sheaves with $\tau$ acting as multiplication by $-1$. Consider the exponential sequence for the $\ZZ_2$-sheaves
$$
1 \ \longrightarrow \ \overline{\cZ} \ \longrightarrow \ \overline{\cR} \ \longrightarrow \ \overline{\cU(1)} \ \longrightarrow \ 1.
$$
As explained in~\cite[p.~10]{Fok2015}, this gives rise to an isomorphism  $H^p(M; \ZZ_2, \overline{\cU(1)}) \simeq H^{p+1}(M; \ZZ_2, \overline{\cZ})$ since $\overline{\cR}$ is a fine $\ZZ_2$-sheaf. 
As $\ZZ_2$ is finite, it follows from \cite[Section~6]{St} that the group $ H^p(M; \ZZ_2, \overline{\cZ})$ is naturally isomorphic to the Borel equivariant cohomology with local coefficients defined by
$$
H^{p}_{\ZZ_2}(M,\ZZ(1)) = H^{p}(E\ZZ_2\times_{\ZZ_2} M,\ZZ(1)) \ ,
$$
for $p\geq 1$, where the local system $\ZZ(1)$ on $E\ZZ_2\times_{\ZZ_2} M$ is defined
by the $\ZZ_2$-action of the fundamental group
$\pi_1(E\ZZ_2\times_{\ZZ_2} M)$ by $-1$ on $\ZZ$ through the natural homomorphism in the homotopy exact sequence 
$$
\pi_1(M) \ \longrightarrow \ \pi_1(E\ZZ_2\times_{\ZZ_2} M) \ \longrightarrow \ \ZZ_2 
$$
for the fibration $M\to E\ZZ_2\times_{\ZZ_2} M \to B\ZZ_2$. There is further a Leray--Serre spectral sequence associated to this fibration,
$$
E^{p,q}_2 = H^p_{\rm gp}(\ZZ_2,H^q(M,\ZZ)\otimes \ZZ(1)) \Longrightarrow H^{p+q}_{\ZZ_2}(M,\ZZ(1)) \ ,
$$
where $H^p_{\rm gp}(\ZZ_2,H^q(M,\ZZ)\otimes \ZZ(1))$ denotes the group cohomology of $\ZZ_2$ with values in the $\ZZ_2$-module $H^q(M,\ZZ)\otimes \ZZ(1)$. Since these cohomology groups are torsion in all non-zero degrees, it follows that rationally $E^{p,q}_2=0$ for $p\neq 0$. Thus the spectral sequence collapses at the second page and the only contribution comes from the degree zero group cohomology given by the invariants of the module (cf. also~\cite[Proposition~3.26]{Fok2015} for an alternative proof)
$$
H^q_{\ZZ_2}(M,\RR(1)) \simeq_\RR E_2^{0,q} = \big\{x\in H^q(M,\RR) \ \big| \ \tau^*(x) = -x \big\} \ .
$$

\section{Real and equivariant line bundles}
\label{sec:RealLine}

Let $M$ be a Real manifold. To understand the sequence \eqref{eq:beginlongsequence} it is useful to explore the 
geometric interpretations of the various terms. First we consider the degree zero terms. 

\begin{proposition}
\label{cor-Hzero}
If $M$ is one-connected, then the sequence 
$$
0 \ \longrightarrow \   H^0(M; \ZZ_2,  \overline{\cU(1)}) \ \longrightarrow \   H^0(M, \cU(1)) \ \xrightarrow{ g \mapsto  g \, g \circ \tau} \   H^0(M; \ZZ_2, \cU(1)) \ \longrightarrow \ 0
$$
is exact. 
\end{proposition}
\begin{proof}
Let $f:M\to U(1)$ be invariant. Since we can regard $f \colon M \to \RR/ \ZZ$ and $H^1(M, \ZZ)
= 0$, we can lift $f$ to a map $\hat f \colon M \to \RR$. As $f$ satisfies $f \circ \tau = f$ we have  $\hat f \circ \tau = \hat f + k$ for $k \in \ZZ$ a constant because $M$ is connected. But $\tau^2 = 1$ so that $\hat f = \hat f \circ \tau + k$ and thus $k = 0$.   If we let 
$\hat g = \frac{\hat f}{2}$ and project $\hat g$ to $g \colon M \to U(1)$ then $(g \circ \tau)\, g = f$.
\end{proof}

Consider the very similar case that $f \colon M \to U(1)$ is Real,
that is $f \circ \tau = \bar f$.
Then it is tempting to conclude that there
is a map $g \colon M \to U(1)$ such that $f = (g \circ \tau)\, \bar g $. This is however not true in general. Consider a lift
$\hat f$ of $f$, then $\hat f \circ \tau + \hat f = k$ for some $k\in\ZZ$, and the 
image of $k$ in $\ZZ_2$ is well-defined independently of the lift of $f$; call it $\epsilon(f)$. 
If $\epsilon(f) = 0$, then we can define $\hat g = -\frac{\hat f}2$ and 
$$
\hat g \circ \tau   -  \hat g  =  \frac{\hat f }2 - \frac{\hat f \circ \tau}2 = \hat f
$$ 
so that $(g \circ \tau)\, \bar g = f $.  Hence we have
\begin{proposition}
\label{prop:Hzero}
If $M$ is one-connected, then the sequence 
$$
0 \ \longrightarrow \   H^0(M; \ZZ_2, \cU(1)) \ \longrightarrow \   H^0( M, \cU(1)) \ \xrightarrow{ g \mapsto \bar g \, g \circ \tau} \   H^0(M; \ZZ_2,  \overline{\cU(1)}) \ \xrightarrow{ \ \epsilon \ } \ \ZZ_2 \ \longrightarrow \ 0
$$
is exact.
\end{proposition}
Notice that $\epsilon(f)$ can also be defined as follows. As $M$ is one-connected we can 
choose a square root of $f$ and consider $\sqrt{f} \,(\sqrt{f} \circ \tau)$ which
is independent of the choice of square root. Then
$(\sqrt{f}\, (\sqrt{f} \circ \tau))^2 = f \,(f \circ \tau) = 1$ so that $\sqrt{f}\,( \sqrt{f} \circ \tau)
= (-1)^{\epsilon(f)} $ defines a constant element of $\ZZ_2$.  In particular if $f = -1$, then $\epsilon(f) = 1$.

Now we consider the degree~one terms. For this, we say that a line bundle $L \to M$ is \emph{Real} if there is a complex anti-linear map $\tau_L \colon L \to L$ covering $\tau \colon M\to M$ whose square is the identity. We will usually suppress the subscript on $\tau_L$.

 \begin{proposition} 
 \label{prop:classify-Real-line}
 The group $  H^1(M; \ZZ_2,  \overline{\cU(1)})$ classifies isomorphism classes of Real line 
 bundles on $M$. 
 \end{proposition}
  \begin{proof} 
 Let $L\to M$ be a Real line bundle with Real structure $\tau \colon L \to L$. Let $\cU$ be a good cover as in Section \ref{sec:HR}. Split the indexing set $I$ for $\cU$ into $I^+$ and $I^-$ interchanged
 by $\tau$. Then choose sections $s_\a \colon U_\a \to L$ for $\a \in I^+$ and
 for $\bar \a \in I^-$ define $s_{\bar\a} = \tau s_\a \circ \tau$.  Because $\tau^2 = 1$
 it follows that $s_{\bar\a} = \tau s_\a \circ \tau$ for all $\a \in I$, and if $g_{\a\b}$
 satisfies $s_\a = s_\b\, g_{\a\b}$ then $g_{\bar\a \bar\b} = \bar{g}_{\a\b}\circ \tau$ is a Real 
 cocycle. 
 
 Let $g_{\a\b}$ be a Real cocycle representing a class in $  H^1(M; \ZZ_2, \overline{\cU(1)})$.
 We can find a line bundle $L \to M$ with local sections $s_\a$ such that $s_\a = s_\b\, g_{\a\b}$.
 If $v = v_\a \,s_\a \in L$, then define $\tau( v) = \bar{v}_\a \, (s_{\bar \a} \circ \tau)$.  If we change
 to $v = v_\b \,s_\b$, then $v_\a\, g_{\a\b} = v_\b$ so that 
\begin{align*}
 \bar{v}_\b\, (s_{\bar\b} \circ \tau) &= \bar{v}_\a \,
                                        \bar{g}_{\a\b}\, (s_{\bar \a} \, g_{\bar\a\bar\b}^{-1}  \circ \tau) \\[4pt]
   &=  \bar{v}_\a\, (s_{\bar \a}  \circ \tau) \, \bar{g}_{\a\b}\, (g_{\bar\a\bar\b}^{-1} \circ \tau)\\[4pt]
    &=  \bar{v}_\a \, (s_{\bar \a}  \circ \tau) \ ,
  \end{align*}                       
giving a well-defined Real structure because $g_{\a\b}$ is Real. It is easy to see that $\tau^2 = 1$  as required. 
\end{proof}

\begin{remark}
We may refer to $  H^1(M; \ZZ_2,  \overline{\cU(1)})$ as the \emph{Real Picard
  group} of $M$, and  the class in $  H^1(M; \ZZ_2,  \overline{\cU(1)})$
corresponding to a Real line bundle $L\to M$ as the \emph{Real Chern
  class} of $L$.
\end{remark}
 
If $M$ is one-connected, then the sequence
 $$
0 \ \longrightarrow \   H^1(M; \ZZ_2, \overline{\cU(1)}) \ \longrightarrow \   H^1( M, \cU(1)) \ \xrightarrow{1 \times \tau^*} \   H^1(M; \ZZ_2, \cU(1))  \ \longrightarrow \  \cdots
 $$
 is exact.  In particular if a line bundle $L\to M$ admits a Real structure, then the Real structure
 is unique up to isomorphism. We can prove this directly as follows.  Assume that $\tau \colon L \to L$
 is a Real structure.  Then any other Real structure takes the form $f \, \tau $ for a map $f \colon M \to U(1)$. 
 Because $(f \, \tau)^2 = 1$ and $\tau^2 = 1$, we deduce that $(f \circ \tau)\, \bar f  = 1$.
 So $f \colon M \to U(1)$ is invariant and thus $f = (g \circ
 \tau)\, g$ for some $g \colon M\to U(1)$. It follows
 that $(L, \tau)$ and $(L, f \, \tau)$ are isomorphic by the isomorphism $L \to L$ induced
 by multiplication with $g$.

We similarly say that a line bundle $L \to M$ is \emph{equivariant} if we lift $\tau \colon M\to M$ to a complex linear isomorphism $\tau \colon L \to L$ with $\tau^2 = 1$; we call the lift of $\tau$ a \emph{$\tau$-action} on $L$.  We  have

  \begin{proposition} 
 The group $  H^1(M; \ZZ_2, \cU(1))$ classifies isomorphism classes of equivariant line 
 bundles on $M$. 
 \end{proposition}
 \begin{proof} 
 We omit the  proof as it is  very similar to the case of Real line bundles in Proposition \ref{prop:classify-Real-line}.  It also follows by combining \cite[Theorem 5.2 and Lemma 4.4]{Gom} and \cite[Section 6]{St}.
  \end{proof}

If $\tau \colon L \to L$ is a lift of $\tau \colon M\to M $ 
making the line bundle $L\to M$ equivariant, then so is $-\tau$.  We can show that, up to isomorphism, these are
the only possible $\tau$-actions when $M$ is one-connected.  Indeed, as in the Real case any new $\tau$-action
takes the form $f\, \tau$ and thus $(f \circ \tau)\, f = 1$ so that
$f \colon M\to U(1) $ is Real. Now
whether or not we can make $(L, \tau)$ and $(L, f \, \tau)$ the same
up to isomorphism depends on the sign of $\epsilon(f)$, so there are only the two possibilities.

We can now interpret the terms in the second row of the exact sequence \eqref{eq:beginlongsequence} geometrically as follows. If $f \colon M \to U(1)$ is invariant,
then the image of the coboundary homomorphism is the trivial line bundle with the Real structure induced
by $f$, that is the Real structure induced by multiplying the trivial Real structure with $f$.
The map $  H^1(M; \ZZ_2, \overline{\cU(1)}) \to   H^1( M, \cU(1))$ forgets the Real structure, while the map 
 $  H^1( M, \cU(1)) \to   H^1(M; \ZZ_2, \cU(1))$ sends a line bundle $J \to M$ representing a class in $  H^1( M, \cU(1))$ to 
 the equivariant line bundle $J \otimes \tau^{-1}(J)$ with the $\tau$-action induced by the obvious isomorphism
 $$
 J \otimes \tau^{-1}(J) \longrightarrow \tau^{-1}(J \otimes
 \tau^{-1}(J) ) \simeq \tau^{-1}(J) \otimes  J \ .
 $$
We postpone the description of the maps in the third row of the
sequence \eqref{eq:beginlongsequence} until Section~\ref{sec:Real bundle
  gerbes}, where we give a way of geometrically realising classes in 
$  H^2(M; \ZZ_2, \overline{\cU(1)})$ as Real stable isomorphism classes of Real bundle gerbes.

\begin{remark}\label{remark:real-bundle} If $L \to M$ is a $U(1)$-bundle then $L \otimes \tau^{-1}(L) \to M$ is naturally equivariant using the 
obvious identification $ \tau^{-1}(L \otimes \tau^{-1}(L)) = \tau^{-1}(L) \otimes L = L \otimes \tau^{-1}(L)$.  A Real structure
 on $L$ is precisely an invariant section of $L \otimes \tau^{-1}(L)$.
 If $\tau$ is a Real structure, then $s(m) = \ell \otimes \tau(\ell)$
 is an invariant section where $\ell \in L_m$, and vice-versa.
 \end{remark}

\begin{example}
\label{ex:pt}
Let $M=\pt$ be a point. A line bundle over a point is a one-dimensional vector space. 
Up to isomorphism there is a unique Real structure on $\CC$ given by conjugation so 
$  H^1(\pt; \ZZ_2, \overline{\cU(1)})=0$. On the other hand, the equivariant line bundles over a 
point are just the collection of possible involutions on $\CC$ which are $\pm\,1$, so $  H^1(\pt; \ZZ_2, \cU(1))=\ZZ_2$.
\end{example}

\begin{example}\label{ex:trivialtau}
Let $\tau=\id_M$ be the trivial involution on $M$. Then $\tau^{-1}(L)=L$
for any line bundle $L$ on $M$, and any Real line bundle can be
naturally regarded as an ordinary real line bundle on $M$~\cite{Ati}, so $ 
H^1(M;\ZZ_2,\overline{\cU(1)})\simeq H^1(M,\ZZ_2)$. Any line bundle
$L\to M$ is trivially equivariant and there are two non-isomorphic
lifts $\pm \id_L$ of $\tau=\id_M$, so $ 
H^1(M;\ZZ_2, \cU(1))\simeq\ZZ_2\oplus H^2(M,\ZZ)$. 
\end{example}

\begin{example}
Let $N$ be any manifold and let $M = N \times \ZZ_2$ with the free
action $\tau \colon (n,x)\mapsto (n,-x)$. The space $M$ is two copies of $N$
labelled by $\pm\,1$, and $\tau$ exchanges the two copies. Any line
bundle $L\to M$ is a pair of line bundles $(L_+,L_-)$ on
$N\times\{+1\}$ and $N\times\{-1\}$, respectively, with $\tau^{-1}(L)=
(L_-,L_+)$. Thus any Real line bundle over $M$ is of the form $(J,J^*)$ and so is completely determined by the complex line bundle $J\to N$~\cite{Ati}, hence
$  H^1(M;\ZZ_2,\overline{\cU(1)})\simeq  
H^1(N, \cU(1))=H^2(N,\ZZ)$.  Similarly any equivariant line bundle over
$M$ is of the form $(J,J)$ and there are two non-isomorphic
$\tau$-actions, hence $  H^1(M; \ZZ_2 , \cU(1))\simeq \ZZ_2\oplus H^2(N,\ZZ)$. The map which sends $  H^1( M, \cU(1))\simeq  
H^1(N, \cU(1))\oplus   H^1(N, \cU(1))$ into
$  H^1(M;\ZZ_2, \cU(1))$ is $(L_+,L_-)\mapsto(L_+\otimes
L_-,L_+\otimes L_-)$.
\end{example}

\begin{example}\label{ex:S2}
As an example of the theory we have developed, we  classify the Real and equivariant line bundles $L$ on  $S^2$
for any Real structure $\tau \colon S^2 \to S^2$.    First note that
we have shown generally that if $M$ is $1$-connected and $L \to M$ then it has zero or one Real structures
and zero or two equivariant structures. Second note from~\cite[Theorem~4.1]{ConstKol} that up to conjugation 
by a diffeomorphism (what they call equivalence) any involution is of three types: (a) it is homotopic to the 
identity, (b) it is equivalent to the antipodal map; or (c) it is equivalent to  conjugation 
on $\CC P^1$  or reflection $(x, y, z) \mapsto (x, y, -z)$ in the equator.   It is a straightforward exercise to show that if $\tau$ is an 
involution and $\tilde \tau = \chi^{-1} \tau \chi$ for a diffeomorphism $\chi$ then $L$ has a Real or equivariant
structure for $\tau$ if and only if $\chi^{-1}L$ has a Real or equivariant structure for $\tilde \tau$. 

First notice that if  $L= \CC \times S^2$ any involution $\tau$ lifts to a Real structure $\tau(u, z) = (\tau(u), \bar z)$ and  to two equivariant structures $\tau(u, z) = (u, \pm z)$ for any Real structure.  

Assuming now that $L$ is not trivial we use various topological facts. First we have $\deg(\tau) = \pm 1$ depending
if it is homotopic to the identity map or the antipodal map.  Moreover  $L$ has  Real structure $\tau^{-1}( L) \simeq L^*$ so that 
$\deg(\tau)  = -1$ and if $L$ has an equivariant structure $\tau^{-1}( L) \simeq L$ so that $\deg(\tau) = 1$.    Bearing
this in mind we consider the three possibilities for $\tau$. 

\noindent
(a) $\tau$ is homotopic to the identity map so $\tau^{-1}(L) \simeq L$
for any line bundle $L \to S^2$ (this includes the equivalence classes
of the identity and the rotation by $\pi$). In that case a class in $  H^1(S^2; \ZZ_2, \overline{\cU(1)})$ represents a line
bundle for which $L \simeq \tau^{-1}(L)^* \simeq L^*$ which is only
possible if $L = S^2\times \CC $ is trivial and 
there is a unique Real structure on it so $  H^1(S^2; \ZZ_2, \overline{\cU(1)}) = 0$. Let $L \to S^2$ be a line 
bundle, and let $\phi$ be the isomorphism $\tau^{-1}(L) \simeq L$.  Then
$\phi^2 = g$ for a map $g \colon S^2\to U(1)$ and it can be checked
that $g = g\circ \tau$, so we can solve $ f \, (f \circ \tau)\, g = 1$ which enables us to show that if 
$\tau = f \, \phi$, then $\tau^2 = 1$ so $L$ is equivariant.  There are two solutions of course so 
$  H^1(S^2; \ZZ_2, \cU(1)) \simeq \ZZ_2 \oplus \ZZ$.   The
inclusion $  H^1(S^2, \cU(1)) \to   H^1(S^2; \ZZ_2, \cU(1))$
sends $L\mapsto L \otimes \tau^{-1}(L) = L^2$ and hence maps $k \in \ZZ$ to $2k$. 

\noindent
(b) $\tau $ is equivalent to the antipodal map so $\tau^{-1}(L) \simeq
L^*$ for any line bundle $L \to S^2$.  Consider first the case that $\tau$ is the antipodal map
and the Hopf bundle $H \to S^2 = \CC P^1$. We can lift the antipodal map $\tau([z_0, z_1]) = 
[-\bar z_1, \bar z_0]$ to an anti-linear map on fibres of $H$ by $\tau(w_0, w_1) = (-\bar w_1, \bar w_0)$ but then $\tau^2 = -1$. As we have seen above this choice cannot be modified to 
give a Real structure.  So the Hopf bundle does not admit a Real
structure in this case. However any 
even power of the Hopf bundle does.  So $  H^1(S^2; \ZZ_2,
\overline{\cU(1)}) \simeq \ZZ$ and it contains the isomorphism
classes of $H^{2k}$ with the Real structure above, which map to the even Chern classes 
in $H^2(S^2, \ZZ)$.  Consider now an equivariant bundle $L \to M$. It then admits
an isomorphism $L \simeq \tau^{-1}(L) \simeq L^*$ which is only possible if $L $ is trivial and
hence has the identity and $-1$ as non-isomorphic
$\tau$-actions. So $  H^1(S^2; \ZZ_2, \cU(1)) = \ZZ_2$. Every line
bundle $L$ in $  H^1(S^2, \cU(1))$ maps to $L \otimes \tau^{-1}(L)
\simeq L \otimes L^* \simeq S^2 \times\CC$. A simple calculation shows that if $L = H$ we obtain the 
trivial line bundle with $-1$ as $\tau$-action. So if $L$ has odd Chern class it maps to the trivial
bundle with $\tau$-action $-1$ while if $L$ has even Chern class it maps to the trivial line 
bundle with the identity as $\tau$-action. 

In the case that $\tau$ is only equivalent to the antipodal map by a diffeomorphism $\chi$ then the arguments above
apply to  $\chi^{-1}( H)$ which is either $H$ or $H^*$ so we deduce the same results. 

\noindent
(c) $\tau$ is equivalent to the reflection about the equator, or
equivalently to the conjugation map $\tau([z_0,z_1])=[\bar z_1,\bar
z_0]$, so again $\tau^{-1}(L) \simeq
L^*$ for any line bundle $L \to S^2$. Again consider first the case that $\tau$ is this involution. This time, however, the
conjugation lifts to an anti-linear map on the fibres of $H$ as
$\tau(w_0,w_1)=(\bar w_1,\bar w_0)$ with $\tau^2=1$, which is the
standard Real structure on the Hopf bundle. Hence again $  H^1(S^2; \ZZ_2,
\overline{\cU(1)}) \simeq \ZZ$, but now the map $  H^1(S^2; \ZZ_2, \overline{\cU(1)}) \to   H^1(S^2,
U(1))$ is the identity. Similarly to the previous case, we have $
H^1(S^2; \ZZ_2, \cU(1)) = \ZZ_2$, where now every line bundle $L$ in $
H^1(S^2, \cU(1))$ maps to the trivial line bundle with identity $\tau$-action.

Again if $\tau$ is only equivalent to the conjugation we can make the same argument.
\end{example}

\begin{example}
\label{ex:2connected-line-bundle}
Let $M$ be two-connected, for example a connected and simply-connected
Lie group with the Cartan involution. Then all line bundles on $M$ are trivial, and so carry
$\tau$-actions. There is a unique Real structure by
Proposition~\ref{cor-Hzero}, so
$H^1(M; \ZZ_2, \overline{\cU(1)}) = 0$, and there are two non-isomorphic $\tau$-actions by
Proposition \ref{prop:Hzero}, hence $H^1(M; \ZZ_2, \cU(1)) = \ZZ_2$.
\end{example}

\section{Real bundle gerbes}\label{sec:Real bundle gerbes}

In this section we will describe a particular modification of the
definition of bundle gerbes, which realises the cohomology group
$  H^2(M; \ZZ_2, \overline{\cU(1)})$ in the same way that bundle gerbes realise $
H^2( M, \cU(1))$. Our notion of {\em Real bundle gerbes} coincides with that by Moutuou in the setting of groupoids \cite[Definition~2.8.1]{Mou}, but we omit the gradings which are not necessary for our purposes.

\subsection{Definitions and examples}

Let $M$ be a manifold with an involution $\tau \colon M \to M$.

\begin{definition}
\label{def:Rbg}
A {\em Real structure} on a bundle gerbe $(P, Y)$ over $M$ is a pair of maps $(\tau_P, \tau_Y)$ where $\tau_Y \colon Y \to Y$ is an involution 
covering $\tau \colon M \to M$, and $\tau_P \colon P \to P$ is a
 conjugate involution covering $\tau_Y^{[2]} \colon Y^{[2]}
\to Y^{[2]}$ and commuting with the bundle gerbe multiplication. A \emph{Real bundle gerbe} over $M$ is a bundle gerbe $(P, Y)$ over $M$ with a Real structure.
\end{definition}

By a  conjugate involution we mean that $\tau_P(p\, z) = \tau_P(p)\, \bar z$ and $\tau_P^2 = \id_P$. Often we will suppress the subscripts on $\tau_P$ and $\tau_Y$.

\begin{remark}
At first this definition appears to be far too strict, as it involves \emph{isomorphism} of bundle gerbes rather than stable isomorphism. There is indeed a weaker notion---known as a \emph{Jandl bundle gerbe}---which we will discuss in Section~\ref{sec:Jandl}. However, we shall see that every Jandl bundle gerbe is in fact equivalent to a Real bundle gerbe and this stronger notion is sufficient to represent the cohomology classes in question. 
\end{remark} 

\begin{remark}
Occasionally it will be important to emphasise the difference between $(P, Y)$ thought of as a Real bundle gerbe and $(P, Y)$ thought of as just a bundle gerbe obtained by forgetting the Real structure.  
In this case we will refer to the latter as a $U(1)$-bundle gerbe. 
\end{remark}

\begin{example}\label{ex:trivial}
If $R \to Y$ is a Real hermitian line bundle with Real structure $\tau_R\colon R \to R^*$, then $(\d R, Y)$ is a Real bundle gerbe with Real structure given by  $\d \tau_R\colon \d R = \pi_1^{-1}(R) \otimes \pi_2^{-1}(R)^* \to \pi_1^{-1}(R)^* \otimes \pi_2^{-1}(R) =\d R^*$. We say that a Real bundle gerbe $(P, Y)$ is \emph{Real trivial} if there is a Real line bundle $R \to Y$ such that
$P = \delta R$ as Real bundle gerbes; this means that $P = \d R$ as
bundle gerbes and  that the isomorphism commutes with the Real
structures. A choice of Real bundle $R$ and an isomorphism $P\simeq \d R$ is called a \emph{Real trivialisation}.
\end{example}

\begin{example} If $(P, Y)$ and $(Q, X)$ are Real bundle gerbes over $M$ with Real structures $\tau_P$ and $\tau_Q$, respectively, then $(P \otimes Q, Y \times_M X)$ is a Real bundle gerbe with the obvious Real structure $\tau_P \otimes \tau_Q\colon P \otimes Q \to P \otimes Q$.
\end{example}

\begin{example}
If $f\colon N\to M$ is an equivariant map of Real spaces and $(P,Y)$ is a Real bundle gerbe on $M$, then $f^{-1}(P,Y)$ is a Real bundle gerbe on $N$. Equivariance determines involutions $\tau_{f^{-1}(Y)}\colon f^{-1}(Y)\to f^{-1}(Y)$ covering $\tau_Y \colon Y\to Y$ and $\tau_{(f^{[2]})^{-1}(P)}\colon (f^{[2]})^{-1}(P)\to (f^{[2]})^{-1}(P)$ covering $\tau_P\colon P\to P$.
\end{example}

\begin{example}
\label{ex:Z2} In the case that $\tau = \id_M$ and $Q \to Y^{[2]}$ is a $\ZZ_2$-bundle gerbe (as in \cite{MatMurSte}) define  $P = Q \times_{\ZZ_2} U(1)$, $\tau_Y = \id_Y$ and $ \tau_P([q, z]) = [q, \bar z]$.
Then $(P, Y)$ is a Real bundle gerbe. Conversely, if $\tau = \id_M$ and $\tau_Y = \id_Y$, then the fixed point set of $\tau_P$ is a reduction of the $U(1)$-bundle $P \to Y^{[2]}$ to a $\ZZ_2$-bundle making it a $\ZZ_2$-bundle gerbe.
\end{example}

\begin{example}\label{ex:double}
Let $N$ be any manifold and let $(P,Y)$ be a bundle gerbe on $N$. Let
$M= N \times \ZZ_2$ with the involution $\tau \colon (n,x)\mapsto
(n,-x)$, and set $Z=Y\times\ZZ_2$ with projection
$p =\pi\times1$ and involution
$\tau_{Z} \colon (y,x)\mapsto(y,-x)$. The fibre product
$Z^{[2]}$ can be naturally identified as
$Y^{[2]}\times\ZZ_2$ with the involution
$\tau^{[2]}_{Z} \colon (y_1,y_2,x)\mapsto(y_1,y_2 ,-x)$, and we set
$Q=(P,P^*)\to Z^{[2]}$ with the involution $\tau_{Q}$ which exchanges the two slots. Then $(Q,Z)$ is a Real bundle gerbe on $M$. Any Real bundle gerbe on $M$ arises in this way.
\end{example}

\begin{example}[The basic bundle gerbe]
\label{ex:basic}
Let $G$ be a compact, connected, simply-connected, simple Lie group and
$\Omega G$ its based loop group. The universal $\Omega G$-bundle is
the path fibration $PG \to G$, where $PG$ is the space of based maps
$[0,1]\to G$ and the projection is evaluation at the endpoint. The
lifting bundle gerbe for this bundle associated to the universal
central extension $\pi \colon \widehat{\Omega G} \to \Omega G$ of the
loop group is a model for the \emph{basic bundle gerbe}. This is given
by the fibre product $PG^{[2]}\to \Omega G$ via $ (p_1, p_2) \mapsto
\gamma$, where $p_2 = p_1\, \gamma$. The basic bundle gerbe $Q \to
PG^{[2]}$ is then given by pulling back the central extension
$\widehat{\Omega G} \to \Omega G$ and the bundle gerbe multiplication
is induced by the group multiplication in $\widehat{\Omega G}$, that is
$$
Q_{(p_1, p_2)} = \widehat{\Omega G}_{p_1^{-1}\, p_2}
$$
where if $h \in \Omega G$ then $\widehat{\Omega G}_h = \pi^{-1}(h)$. 

Let $G$ be equipped with the involution $\tau \colon g\mapsto
g^{-1}$. This lifts to an involution $\tau \colon p \mapsto p^{-1}$ on
$ PG$ and if $p_2 = p_1 \, \gamma$ then 
$p_2^{-1} = p_1^{-1}\, (\Ad_{p_1}(\gamma^{-1}))$. By \cite{BaeCraSchSte} the adjoint action $\Ad \colon PG \to \Aut(\Omega G)$ lifts to an action on $\widehat{\Omega G}$ and hence we define a Real structure 
$
\tau \colon Q_{(p_1, p_2)} \to Q_{(p_1^{-1}, p_2^{-1})} 
$
 given by $\tau(q) = \Ad_{p_1}(q^{-1})$. If $q_{12} \in Q_{(p_1, p_2)}$ and $q_{23} \in Q_{(p_2, p_3)}$ then 
\begin{align*}
 \tau( q_{12}) \, \tau(q_{23} ) &= \Ad_{p_1} ( q^{-1}_{12}) \Ad_{p_2} (q^{-1}_{23}) \\[4pt]
               &= \Ad_{p_1}\big( q^{-1}_{12} \Ad_{p_1^{-1}\, p_2} (q^{-1}_{23})\big) \\[4pt]
                &= \Ad_{p_1}\big( q^{-1}_{12} \,q_{12} \,(q^{-1}_{23})\, q_{12}^{-1}\big)\\[4pt]
           & = \tau(q_{12}\, q_{23})
           \end{align*}
where here we use the fact that $\pi(q_{12}) = p_1^{-1} \, p_2$ so that the adjoint action
$\Ad_{p_1^{-1}\, p_2}$ on $\widehat{\Omega G}$ is conjugation by $q_{12}$.  
We also have
$$
\tau^2(q_{12}) = \tau\big(\Ad_{p_1}(q_{12}^{-1})\big) = \Ad_{p_1^{-1}}\big( (\Ad_{p_1}(q_{12}^{-1}) )^{-1}\big)
= \Ad_{p_1^{-1}}\big( \Ad_{p_1}(q_{12}) \big) = q_{12}
$$
and hence this is a Real structure.
\end{example}

\begin{example}[The tautological bundle gerbe]
\label{ex:tautological}
Let $M$ be two-connected. Assume that $\tau \colon M \to M$ has at
least one fixed point $m$ and $M$ admits an integral three-form $H$ satisfying
$\tau^*(H) = - H$; for example, these conditions are satisfied by the Lie group $M
= SU(n)$ with $\tau(g) = g^{-1}$.

Recall the construction of the \emph{tautological bundle gerbe} from
\cite{Mur}.  Let $Y = PM$ be the space
of paths based at $m$ with endpoint evaluation as  projection to $M$.
If $p_1, p_2\in Y$ have the same endpoint choose a surface $\Sigma\subset M$
spanning them, that is the boundary of $\Sigma$ is $p_1$ followed by
$p_2$ with the opposite orientation. Then the fibre of $P \to
Y^{[2]}$ consists of all triples $(p_1, p_2, \Sigma, z)$ modulo the equivalence relation 
$(p_1, p_2, \Sigma, z) \sim (p_1, p_2, \Sigma', z'\, )$  if
$\hol(\Sigma \cup \Sigma', H) \, z = z'$.
Here $\hol(S, H)$, for any closed
surface $S\subset M$, is the usual Wess--Zumino--Witten term defined by
\begin{equation}\label{eq:WZW}
\hol(S, H)  = \exp\Big( 2 \pi \,{\rm i}\, \int_{B(S)}\, H \Big)
\end{equation}
for a choice of three-manifold $B(S)$ whose boundary is $S$, which is well-defined because $H$ is an integral form. The bundle gerbe product is
$$
(p_1, p_2, \Sigma, z) \otimes  (p_2, p_3, \Sigma', z'\, )\longmapsto
(p_1, p_3, \Sigma\cup\Sigma', z\,z'\,) \ .
$$

We define a Real structure  $\tau$ by the fact that 
$$
 (p_1, p_2, \Sigma, z)   \longmapsto   ( \tau(p_1), \tau(p_2), \tau(\Sigma), \bar z)
 $$
descends through the equivalence relation to give a conjugate bundle
gerbe isomorphism $P \to \tau^{-1}(P)$.  We leave this easy check as an exercise
for the reader.  
\end{example}

\begin{example}[The lifting bundle gerbe]
A Lie group $G$ is {Real} if it possesses an involutive automorphism $\sigma\colon G\to G$. If $M$ is a Real space and $G$ is a Real Lie group then
a \emph{Real $G$-bundle} over $M$ is a principal $G$-bundle
$P$ with a Real structure $\tau_P$ that commutes with the involution
on $M$ and is compatible with the right $G$-action, that is $\tau_P(p\, g)=\tau_P(p)\, \sigma(g)$. A central extension
$$
1\ \longrightarrow \ U(1)\ \longrightarrow \ \widehat{G} \ \xrightarrow{ \ \pi \ }
\ G \ \longrightarrow \ 1 
$$
of a Real Lie group $G$ is called Real if $\widehat{G}$ is a Real Lie group whose Real structure descends to that on $G$ with respect to the conjugation involution on $U(1)$. We apply the \emph{lifting bundle gerbe} construction of \cite{Mur} to Real $G$-bundles. If $P \to M$ is a $G$-bundle then there is a map $\rho \colon P^{[2]} \to G$
defined by $p_2 = p_1 \, \rho(p_1, p_2)$; then $\rho(p_1, p_2)\, \rho(p_2, p_3) = \rho(p_1, p_3)$. The fibre $Q_{(p_1, p_2)}$ of the lifting bundle gerbe over $(p_1, p_2)$ is $\pi^{-1}(\rho(p_1, p_2)) 
\subset \widehat{G}$. Thus $Q = \rho^{-1}(\widehat{G})$ where we regard $\widehat{G} \to G$ 
as a $U(1)$-bundle; the group action on $\widehat{G}$ defines the bundle gerbe multiplication. If $P \to M$ is a Real $G$-bundle then 
$\rho(\tau_P(p_1), \tau_P(p_2)) = \sigma(\rho(p_1, p_2)) $ and the action of $\sigma$ on $G$ induces a Real structure on $Q_{(p_1, p_2)} \to Q_{(\tau_P(p_1), \tau_P(p_2))}$.
\label{ex:lifting}\end{example}

\subsection{The Real Dixmier--Douady class of a Real bundle gerbe}
\label{sec:RDD}
Let $M$ be a Real manifold and $(P, Y)$ a Real bundle gerbe over $M$. Just like ordinary bundle gerbes in Section \ref{sec:bg}, we will now show that a Real bundle gerbe gives rise to a cohomology class in $H^2(M; \ZZ_2, \overline{\cU(1)} )$.

Choose a  good Real open cover $\cU = \{U_\a\}_{\a \in I}$ as in Section \ref{sec:HR} and split $I$ as a disjoint union of $I_+$ and $I_-$ which are interchanged under $\tau$. For $\a \in I_+$ choose sections $s_\a \colon U_\a \to Y$ and define 
$s_{\bar \a} \colon U_{\bar \a} \to Y$ by $s_{\bar \a}= \tau s_\a \circ \tau$. Because $\tau$ is an involution we  have $s_\a = \tau s_{\bar \a} \circ\tau$ for all $\a \in I$.  Similarly split $I^2$ and for  $(\a, \b) \in I^2_+$ choose  $\sigma_{\a\b}(m) \in P_{(s_\a(m), s_\b(m))}$, and 
 define $\sigma_{\bar\a\bar\b} = \tau  \sigma_{\a\b} \circ\tau$.  Again it follows  that 
$\sigma_{\a\b}(m) \in P_{(s_\a(m), s_\b(m))}$ and
$\sigma_{\a\b} = \tau  \sigma_{\bar\a\bar\b}\circ \tau$ for all
$(\a,\b) \in I^2$, where we used $\tau_P^2 = \id_P$. 

Define $g_{\a\b\c}\colon U_{\a\b\c} \to U(1)$ by 
$$
\sigma_{\a\b} \, \sigma_{\b\c} = \sigma_{\a\c} \, g_{\a\b\c} \ .
$$
Then $g_{\bar\a\bar\b\bar\c}$ is given by
$$
\sigma_{\bar\a\bar\b} \, \sigma_{\bar\b\bar\c} = \sigma_{\bar\a\bar\c}
\, g_{\bar\a\bar\b\bar\c} \ .
$$
Applying $\tau$ to the first equation (and evaluating at $\tau(m)$) we get 
$$
(\tau \sigma_{\a\b}\circ\tau)\, (\tau \sigma_{\b\c} \circ\tau)= (\tau
\sigma_{\a\c} \circ\tau)\, (\bar{g}_{\a\b\c}\circ\tau ) \ .
$$
Hence $g_{\bar\a\bar\b\bar\c} = \bar{g}_{\a\b\c}\circ\tau$ so the
cocycle defined by $g_{\a\b\c}$ is Real. If we chose
different sections $\sigma'_{\a\b}$ satisfying $\sigma'_{\bar\a\bar\b}
= \tau  \sigma'_{\a\b} \circ\tau$, then $\sigma'_{\a\b} = \sigma_{\a\b}
\, h_{\a\b}$ for some $h_{\a\b} \colon U_{\a\b} \to U(1)$ satisfying $h_{\bar\a\bar\b} = \bar{h}_{\a\b} \circ\tau$, and thus $g_{\a\b\c}$ changes by a Real coboundary.

We call the class defined by $g_{\a\b\c}$ the \emph{Real Dixmier--Douady class} and denote it by 
$$
\DD_R(P) \ \in \   H^2(M; \ZZ_2, \overline{\cU(1)} )  \ . 
$$

This shows how a Real bundle gerbe yields a cohomology class in $H^2(M; \ZZ_2, \overline{\cU(1)} )$, which is natural with respect to pullbacks in the category of Real spaces. We
also immediately have $\DD_R(P^*)=-\DD_R(P)$ and

\begin{proposition}\label{P:additive}
The Real Dixmier--Douady class satisfies $\DD_R(P \otimes Q) = \DD_R(P) + \DD_R(Q)$.
\end{proposition}

We would like to define an equivalence relation on Real bundle gerbes that means two Real bundle gerbes are equivalent precisely when they have the same Real class. Following the approach of \cite{MurSte2} for $U(1)$-bundle gerbes we first prove

\begin{proposition}\label{P:obstruction}
The Real Dixmier--Douady class of a Real bundle gerbe $P$ vanishes precisely when $P$ is Real trivial.
\end{proposition}
\begin{proof}
First suppose that  $\DD_R(P)$ is trivial, so that if $g_{\a\b\c}$ is
a representative for the Real Dixmier--Douady class, chosen relative
to sections $\sigma_{\a\b}$ as before, then $g_{\a\b\c} = h_{\a\b} \,
\bar{h}_{\a\c} \, h_{\b\c}$ where $h_{\a\b}$ satisfies $h_{\bar \a\bar\b} = \bar{h}_{\a\b} \circ\tau$. We have
$$
\sigma_{\a\b} \, \sigma_{\b\c} = \sigma_{\a\c} \, h_{\a\b} \,
\bar{h}_{\a\c} \, h_{\b\c} \ ,
$$
and hence
$$
\sigma_{\a\b}\, \bar{h}_{\a\b}\, \sigma_{\b\c}\, \bar{h}_{\b\c} =
\sigma_{\a\c} \, \bar{h}_{\a\c} \ .
$$
Therefore we may define sections $\hat{\sigma}_{\a\b} =
\sigma_{\a\b}\, \bar{h}_{\a\b}$ which satisfy the cocycle condition. They further satisfy the condition $\hat{\sigma}_{\bar\a\bar\b} = \tau \hat{\sigma}_{\a\b}\circ\tau$ since
$$
\hat{\sigma}_{\bar\a\bar\b}	= \sigma_{\bar\a\bar\b}\, \bar{h}_{\bar\a\bar\b}
					= (\tau
                                        {\sigma}_{\a\b}\circ\tau )\, ({h}_{\a\b}\circ \tau)
					= \tau(\sigma_{\a\b}\, \bar{h}_{\a\b} \circ\tau)
					= \tau
                                        \hat{\sigma}_{\a\b}\circ\tau \
                                        .
$$
Define $R^\a \to \pi^{-1}(U_\a)$ by $R^\a_y = P_{(y, s_\a
  \pi(y))}$. Then $\coprod_{\a\in I}\, R^\a$ defines a bundle over
$\coprod_{\a\in I}\, \pi^{-1}(U_\a)$ and $\hat{\sigma}_{\a\b}(\pi(y))
\in P_{(s_\a \pi(y), s_\b \pi(y))} = P_{(y, s_\a \pi(y))}^* \otimes
P_{(y, s_\b \pi (y))}$ give descent data for $\coprod_{\a\in I}\,
R^\a$. This determines a bundle $R \to Y$ such that $P = \d R$. Note
that $R$ is Real since $\tau (p) \in P^*_{(\tau (y) , \tau s_\a \pi
  (y) )} = P^*_{(\tau (y) ,  s_{\bar{\a}} \pi (\tau( y)) )}$ for $p \in P_{(y, s_\a \pi (y))}$. Thus $P$ is Real trivial.

Suppose instead that $P = \d R$, where $R \to Y$ is a Real bundle with Real structure $\tau_R \colon R \to R^*$; then $(s_\a, s_\b)^{-1}(P) = s_\a^{-1}(R)^* \otimes s_\b^{-1}(R)$. Choose sections $h_\a \colon U_\a \to s_\a^{-1}(R)$ and define $h_{\bar\a} = \tau_R h_\a\circ \tau$ and sections $\sigma_{\a\b}$ of $(s_\a, s_\b)^{-1}(P)$ by $\sigma_{\a\b} = h^*_\a\, h_\b$. Since $P  = \d R$ as Real bundles these sections satisfy the Reality condition $\sigma_{\bar\a\bar\b} = \tau \sigma_{\a\b}\circ \tau$. It follows that $g_{\a\b\c} = 1$, and hence the Real Dixmier--Douady class of $P$ is trivial.
\end{proof}

We say that two Real bundle gerbes $(P, Y)$ and $(Q, X)$ are
\emph{Real stably isomorphic} if $Q \otimes P^*$ is Real trivial. A
\emph{Real stable isomorphism} $P\to Q$ is a Real trivialisation of $Q \otimes P^*$. Propositions~\ref{P:additive} and \ref{P:obstruction} imply that $P$ and $Q$ are Real stably isomorphic if and only if $\DD_R(P) = \DD_R(Q)$, and we have 

\begin{proposition} \label{P:Real classification}
The Real Dixmier--Douady class induces a bijection between Real bundle gerbes modulo Real stable isomorphism and $  H^2(M; \ZZ_2, \overline{\cU(1)})$.
\end{proposition}

\begin{proof}
This follows from the discussion above and all that remains is to show that every class in $  H^2(M; \ZZ_2, \overline{\cU(1)})$ gives rise to a Real bundle gerbe. We use the same approach as \cite{Mur} for $U(1)$-bundle gerbes.
Suppose $[g_{\a\b\c}] \in   H^2(M; \ZZ_2, \overline{\cU(1)} ) $. Let $Y := \coprod_{\a\in
  I}\,  U_{\a}$ be the nerve of the open cover $\cU$, and let $P \to Y^{[2]}$ be given by $\coprod_{\a,\b\in I}\, U_{\a\b} \times U(1)$. The bundle gerbe multiplication on $P$ is given by $(m, \a, \b, z) \otimes (m, \b, \c, w) = (m, \a , \c, z\, w \, g_{\a\b\c}(m))$ and the Real structure is $(m, \a, \b, z) \mapsto (\tau(m), \bar\a, \bar\b, \bar z)$. It is straightforward to show that the condition of being a Real cocycle implies that the Real structure
  commutes with the bundle gerbe product.
\end{proof}

Let $\cH$ be a complex separable Hilbert space  with a conjugation $v \mapsto \bar v$. This induces a complex anti-linear
involution $\sigma$  on the unitary operators $U(\cH)$ by $\sigma(g)(v) = \overline{ g (\bar v) }$ which also descends to the projective unitary group $PU(\cH)$. Then $\sigma$ is a group homomorphism. 
For our discussion of twisted $\KR$-theory later on
we need
\begin{proposition}
\label{prop:PU(H)}
There is a bijection between isomorphism classes of Real $PU(\cH)$-bundles and Real stable isomorphism classes of Real bundle gerbes on $M$.
\end{proposition} 
\begin{proof}
We apply the Real lifting bundle gerbe construction of Example~\ref{ex:lifting} to $PU(\cH)$-bundles. 
In \cite{Fok2015} it is shown that Real $PU(\cH)$-bundles are classified up to isomorphism by their Real Dixmier--Douady classes in $  H^1(M; \ZZ_2,  PU(\cH)) \simeq   H^2(M; \ZZ_2, \overline{\cU(1)} ) $.  It is straightforward to check that the Real Dixmier--Douady class of a Real $PU(\cH)$-bundle
 is the same as that of its lifting bundle gerbe and the result follows from Proposition \ref{P:Real classification}.
\end{proof}

\subsection{Equivariant line bundles and Real structures}
Consider the following part of the long exact sequence \eqref{eq:beginlongsequence} from Section \ref{sec:HR} 
\begin{equation}
\label{eq:shortlong}
\xymatrix@R=2ex{ 
  &   \qquad\qquad\qquad \cdots \ar[r] & \   H^1(M, \cU(1)) \ \ar[r]^{1 \times \tau^*} & \   H^1(M; \ZZ_2, \cU(1)) \ \ar`r[d]`[lll]`[ddlll] `[ddll][ddll]&\\
&&&&\\
 & \   H^2(M; \ZZ_2, \overline{\cU(1)}) \ \ar[r] & \   H^2(M, \cU(1)) \ \ar[r]^{1 \times \tau^*} & \   H^2(M; \ZZ_2, \cU(1)) \ \ar[r]  & \ \cdots \ .
}
\end{equation}
We showed in the case of line bundles that $H^1(M; \ZZ_2, \cU(1))$ classified isomorphism classes of equivariant line bundles. In a similar fashion 
it is possible to show that $H^2(M; \ZZ_2, \cU(1))$ classifies equivariant bundle gerbes up to equivariant stable isomorphism.  As we do not need 
this notion for our applications we will not spell out the details other than to note that it amounts to simply removing the conjugate
requirement in Definition \ref{def:Rbg} and appropriately modifying the definitions and proofs in Section \ref{sec:RDD}. 

In fact it is 
possible as in \cite{GSW} to cover both the Real and equivariant cases at once.   Just as in the line 
bundle case $P \otimes \tau^*(P)$ has a natural equivariant structure and the map ${1 \times \tau^*}$ is induced by the map 
$P \mapsto P \otimes \tau^*(P)$. 
If $P$ is a bundle gerbe with Dixmier-Douady  class in $H^2(M, \cU(1))$ in the kernel of $1 \times \tau^*$ then it is stably isomorphic, 
as a bundle gerbe, to a bundle gerbe with  a Real structure. The latter we have seen are classified  up to Real stable isomorphism 
by their Real Dixmier-Douady class in $H^2(M; \ZZ_2, \overline{\cU(1)})$.  Of course this will generally not be unique. In fact for a given class in the kernel of $1 \times \tau^*$ the set of inequivalent Real bundle gerbes up to Real stable isomorphism is a torsor over $  H^1(M;\ZZ_2, \cU(1))/  H^1( M, \cU(1))$. We have already described the map $  H^1( M, \cU(1)) \to   
H^1(M; \ZZ_2,  \cU(1))$. Similarly, if $P$ is a bundle gerbe with
Dixmier--Douady class in $  H^2( M, \cU(1))$ the map $  H^2(M,
U(1)) \to   H^2(M; \ZZ_2,  \cU(1))$ is induced by the map that sends $P$ to the equivariant
bundle gerbe $P \otimes \tau^{-1}(P)$.   The map $  H^2(M; \ZZ_2,
\overline{\cU(1)}) \to   H^2( M, \cU(1))$ is induced by the forgetful map
sending a Real bundle gerbe to the underlying $U(1)$-bundle
gerbe. Of course at the level of cohomology it sends the Real stable isomorphism class of a Real bundle gerbe to the stable isomorphism class of the underlying $U(1)$-bundle gerbe.  Therefore, to understand this sequence geometrically it remains
to show how an equivariant line bundle gives rise to a Real bundle gerbe.

Suppose that $P \to M$ is an equivariant bundle so that $\tau^{-1}(P) = P$. Let $Y = M \times \ZZ_2$ with the involution $\tau \colon (m , x)  \mapsto (\tau(m), x + 1)$ covering $\tau$ on $M$. Let $\pi_M \colon Y \to M$ be the projection. Then $Y$ is two copies of $M$ labelled by $0$ and $1$ so that any bundle $Q \to Y$ is a pair of bundles $(Q_0, Q_1)$
on $M \times \{0\}$ and $M \times \{1\}$, respectively. For such a
bundle $Q$ we have $\tau^{-1}(Q) = (\tau^{-1}(Q_1), \tau^{-1}(Q_0))$, and if $L \to
M$ is a line bundle then 
$\pi_M^{-1}(L) = (L, L) \to Y$.   Consider the bundle $  (U(1)_M, P) \to Y$ where $U(1)_M = M \times U(1)$ is the trivial $U(1)$-bundle on $M$.
Then $ (U(1)_M, P)$ is not Real, since 
\begin{align*}
\tau^{-1}(( U(1)_M, P)^*) &= \tau^{-1}( U(1)_M, P^*) \\[4pt] &=(\tau^{-1}(P^*), U(1)_M)
  \\[4pt] &= (P^*, U(1)_M) \\[4pt] &= (P^*, P^*) \otimes ( U(1)_M, P) 
\end{align*}
so that $\tau^{-1} ((U(1)_M, P))^* \otimes \pi_M^{-1}(P) = (U(1)_M, P)$. Hence $\d \tau^{-1} (U(1)_M, P)^* = \d  (U(1)_M, P)$ because $\delta(\pi^{-1}_M(P))$ is canonically trivial. It is straightforward to check that the Real structure this defines satisfies $\tau^2 =1 $.  Hence $\delta(U(1)_M, P)$
is a Real bundle gerbe. The coboundary map $  H^1(M; \ZZ_2, \cU(1)) \to   H^2(M; \ZZ_2, \overline{\cU(1)})$ is then induced by  $ P \mapsto \d (U(1)_M, P)$. Then $\d (U(1)_M, P)$ is trivial as a bundle gerbe, so it is in the kernel of the forgetful map $  H^2(M; \ZZ_2, \overline{\cU(1)}) \to   H^2( M, \cU(1))$, but it is not in general trivial as a Real bundle gerbe. In fact $\d (U(1)_M, P)$ is Real trivial if and only if $P = K \otimes \tau^{-1}(K)$ for some bundle $K$ on $M$. To see this note first that 
any Real bundle on $Y$ has the form $(K^*, \tau^{-1}(K))$, so if  $\d (U(1)_M, P)$ is Real trivial then $\d (U(1)_M, P) = \d (K^*, \tau^{-1}(K))$. Since
 $ (U(1)_M, P) $ and $ (K^*, \tau^{-1}(K))$ are two trivialisations of the same bundle gerbe, they differ by a line bundle $L$ on $M$ and thus
 $$
(U(1)_M, P) =  (K^*, \tau^{-1}(K)) \otimes   (L, L) = (K^* \otimes L,   \tau^{-1}(K) \otimes L )
$$
so that $P = K \otimes \tau^{-1}(K)$. Conversely if $P = K\otimes \tau^{-1}(K)$ for some bundle $K \to M$ then
$$
( U(1)_M, P)  = (U(1)_M, K \otimes \tau^{-1}(K)) =   (K^*, \tau^{-1}(K)) \otimes
(K, K) = (K^*, \tau^{-1}(K))\otimes \pi_M^{-1}(K) \ ,
$$
hence $\d  ( U(1)_M, P) = \d (K^*, \tau^{-1}(K))$ and thus $\d ( U(1)_M, P)$ is Real trivial.

\begin{example} \label{ex:gerbept}
Let $M=\pt$, so that $  H^1(\pt, \cU(1))=  H^2(\pt, \cU(1))=0$. The long exact
sequence \eqref{eq:beginlongsequence} gives 
$$
0 \ \longrightarrow \   H^1(\pt; \ZZ_2, \cU(1)) \ \longrightarrow \
  H^2(\pt; \ZZ_2, \overline{\cU(1)}) \ \longrightarrow \ 0
$$
and hence $  H^2(\pt; \ZZ_2, \overline{\cU(1)})=H^1(\pt; \ZZ_2, \cU(1)) = \ZZ_2$ by
Example~\ref{ex:pt}. Recall that the trivial line bundle has two possible
lifts $\pm\, 1$ of the trivial involution on a point. Then the construction above gives rise to two Real bundle gerbes over $\pt$ which are not Real stably isomorphic to each other. 

Consider the Real open cover $U_0 =\{{\rm   pt}\}=U_{\bar 0}$, and
define a Real two-cochain $g$ by taking $g_{0\bar 0 0}=-1=g_{\bar 0 0\bar 0}$ and
$g_{\a\b\c}=1$ otherwise. Then $\delta(g)=1$ so $g$ is a Real cocycle. Suppose
$\sigma_{\a\b}$ is Real so that $\sigma_{00}=\bar\sigma_{\bar0\bar0}$ and
$\sigma_{0\bar0}=\bar\sigma_{\bar0 0}$, and set
$g_{\a\b\c}=\sigma_{\a\b}\,\sigma_{\a\c}^{-1}\, \sigma_{\b\c}$. Then
we find $\sigma_{00}=\sigma_{\bar0\bar0}=1$ and $|\sigma_{0\bar0}|^2=-1$, and so
$g$ is a non-trivial cocycle in $  H^2({\rm pt};\ZZ_2, \overline{\cU(1)})=\ZZ_2$
which gives rise to a non-trivial Real bundle gerbe over $\pt$ by the construction in the proof of Proposition \ref{P:Real classification} (cf.\ also~\cite[Example~3.27]{Fok2015}).
\end{example}

\begin{example}
Let $M=S^2$. 
Since $  H^2(S^2, \cU(1)) = H^3(S^2, \ZZ) = 0$, it follows from
Example~\ref{ex:S2} and the long exact sequence \eqref{eq:beginlongsequence} that if $\tau $ is homotopic
to the identity then $  H^2(S^2; \ZZ_2, \overline{\cU(1)}) \simeq \ZZ_2 \oplus \ZZ_2$, if $\tau $ is 
equivalent to the antipodal map then $  H^2(S^2; \ZZ_2, \overline{\cU(1)}) = 0$,
while if $\tau$ is equivalent to the reflection about the equator then $  H^2(S^2; \ZZ_2, \overline{\cU(1)}) \simeq \ZZ_2$.
\label{ex:RealH2S2}\end{example}

\begin{example}
Let $G$ be a compact, connected, simply-connected, simple Lie group,
for example $G=SU(n)$. Then $G$ is two-connected so $H^1(G, \cU(1)) = 0$. We have seen in Example \ref{ex:2connected-line-bundle}  that $H^1(G; \ZZ_2, \cU(1)) = \ZZ_2$ so that the long exact sequence \eqref{eq:beginlongsequence} reduces in part to
$$
0 \ \longrightarrow \ \ZZ_2 \ \longrightarrow \ H^2(G; \ZZ_2 , \overline{\cU(1)})
\ \longrightarrow \ H^2(G, \cU(1)) = \ZZ \ .
$$
The final map is surjective since, as shown in Example
\ref{ex:basic}, the basic bundle gerbe on $G$ admits a Real structure. 
Hence  $H^2(G; \ZZ_2 , \overline{\cU(1)}) = \ZZ_2 \oplus \ZZ$, and hence each stable isomorphism class of bundle 
gerbe on $G$ arises from exactly two  Real stable isomorphism classes of Real bundle gerbes over $G$. This result also follows as a special case of~\cite[Proposition~4.2]{Fok2015}.
\label{ex:RealH2G}\end{example}

\subsection{Jandl gerbes}\label{sec:Jandl}
A prior and alternative approach to the notion of Real bundle gerbes is that of {\em Jandl gerbes}  introduced  in \cite{SchSch}.  These can be regarded
as replacing the isomorphism between $\tau^{-1}(P)^*$ and $P$ with a stable isomorphism.  As we will not need this notion for our results, we restrict
ourselves here to some general remarks indicating the connection with Real bundle gerbes.  Note that the discussion below refers to Jandl gerbes without connective structure.

Recall from Remark \ref{remark:real-bundle} that a Real structure on a line bundle $L$ can be understood as an invariant section of $L \otimes \tau^{-1}(L)$. Applying this idea to bundle gerbes, we could have defined a Real bundle gerbe to be a bundle gerbe $P$ with an equivariant trivialisation 
of $P \otimes \tau^*(P)$.   It is not difficult to see 
that  $P \otimes \tau^{-1}(P)$ being equivariantly trivial is equivalent to $P$ being a Jandl gerbe and that a choice of an equivariant trivialisation of $P \otimes \tau^{-1}(P)$ is a choice  of a Jandl structure for $P$.   

Notice  that given this,  the Dixmier-Douady class of a  bundle gerbe $P$  is zero under the second map ${1 \times \tau^*}$ below
$$
 H^2(M; \ZZ_2, \overline{\cU(1)}) \longrightarrow  H^2(M, \cU(1))
 \xrightarrow{ \ 1 \times \tau^* \ }  H^2(M; \ZZ_2, \cU(1)) 
$$
if and only if it admits a Jandl structure. This is in contrast to the Real case where the vanishing 
of $({1 \times \tau^*})( DD(P))$ only implies that $P$ is stably isomorphic (as a $U(1)$-bundle gerbe) to a bundle gerbe that has a Real structure.

\section{Real bundle gerbe modules and twisted $\KR$-theory}
\label{sec:KR}

In this section we introduce the Real version of the notion of bundle gerbe module
which was defined in Section~\ref{sec:bg}, and use it to provide a
geometric picture of twisted $\KR$-theory; an analogous description
also appears in \cite{Mou} in terms of Real twisted vector bundles.

\subsection{Bundle gerbe $\KR$-theory}

We begin with some preliminary remarks that will be tacitly used below. Let $V$ be a hermitian vector space and $R$ a $U(1)$-torsor. Define an equivalence relation on $R \times V$ by $(r\,z, v) \sim (r, v\,z)$ for any $z \in U(1)$. 
Denote the set of equivalence classes $[r, v]$ by $R \otimes V$ and make it into a vector space by defining $[r, v] + [r, w] = [r, v+w]$ and $\lambda\, [r, v ] = [r,  \lambda\, v] $ for $\lambda \in \CC$. Finally define an inner product by $\langle [r, v] , [r, w] \rangle = \langle v, w \rangle $. For any $r \in R$ the map 
$V \to R \otimes V$ defined by $v \mapsto [r, v]$ is a hermitian linear isomorphism. There is a natural isomorphism 
$$
\overline{ R \otimes V } \simeq  {R}^* \otimes \overline{V} 
$$
induced by the obvious identity on sets. If $L$ is a one-dimensional hermitian vector space and $R$ is the set of vectors
in $L$ of length one, then $R \otimes V \simeq L \otimes V$.

\begin{definition}\label{realgerbemodule}
Let $M$ be a Real manifold and $(P, Y)$ a Real bundle gerbe on $M$. Let $E$ be a vector bundle on $Y$ and $\tau_E \colon E \to E$ a conjugate linear involution of fibres commuting with 
the Real structure on $Y$. 
 We say that $E$ is a {\em Real bundle gerbe module} if it is a bundle gerbe module and the Real structure
 commutes with the bundle gerbe action on $E$ in the sense that for every pair $(y_1, y_2) \in Y^{[2]}$ there is a commutative
 diagram
$$ 
\xymatrix{ 
P_{(y_1, y_2)} \otimes E_{y_2} \ \ar[d]_{\tau \otimes \tau_E} \ar[r] & \ E_{y_1} \ar[d]^-{\tau_E} \\ 
 P_{(\tau(y_1), \tau(y_2))}^* \otimes \overline E_{\tau(y_2)} \ \ar[r] & \ \overline E_{\tau(y_1)} } 
$$ 
\end{definition}

We say that two Real bundle gerbe modules are \emph{isomorphic} if
they are isomorphic as Real vector bundles and the isomorphism
preserves the action of the Real bundle gerbe $(P, Y)$. Denote by
$\Rmod(P, Y)$ the set of all isomorphism classes of Real bundle gerbe
modules. It is straightforward to check that it is a commutative
semi-group under direct sum.

\begin{remark}
\label{remark:torsion}  A Real trivialisation of $(P, Y)$ is precisely a rank one Real bundle
gerbe module. As $P^{\otimes r}$ acts on the top
exterior power $\bigwedge^r E$, where $r={\rm rank}(E)$, it follows
that if the Real bundle gerbe $(P, Y)$ admits a finite-dimensional
bundle gerbe module of rank $r$, then the Real  Dixmier--Douady
class $\DD_R(P)$ is a torsion element in $  H^2(M; \ZZ_2,
\overline{\cU(1)} ) $ of order dividing $r$. 
\end{remark}

Denoting by $\RVect(M)$
the semi-group of Real vector bundles on $M$ in the sense of Atiyah~\cite{Ati},  we have

\begin{proposition} \label{modules}
\begin{enumerate}\thmenumhspace
\item If $(P, Y)$ and $(Q, X)$ are Real bundle gerbes, then a Real
  stable isomorphism $(P, Y)\to (Q, X)$ induces an isomorphism
of semi-groups $\Rmod(P, Y) \to \Rmod(Q, X)$. 
\item If $(P, Y)$ is a trivial Real bundle gerbe, then a choice of
  Real trivialisation defines an isomorphism of semi-groups $\Rmod(P, Y) \to \RVect(M)$.
\end{enumerate}
\end{proposition}
\begin{proof}
For (1) it is enough to follow the proof of \cite[Proposition 4.3]{BouCarMat} and notice that the bundle gerbe module in $\Rmod(Q, X)$ 
carries a Real structure.  Similarly the proof of (2) follows that of~\cite[Proposition 4.2]{BouCarMat} and it is straightforward to verify that the 
Real structure descends.
\end{proof}

If $(P, Y)$ is a Real bundle gerbe with torsion Dixmier--Douady class over a Real compact manifold $M$,
we denote by $\KR_{\bg}(M,P)$ the Grothendieck group of the
semi-group $\Rmod(P, Y)$ and call it the \emph{$\KR$-theory group} of the Real bundle gerbe. As an immediate corollary of Proposition~\ref{modules}~(1), we note that any choice of Real stable isomorphism $(P, Y)\to (Q, X)$ induces an isomorphism  
$\KR_{\bg}(M,P) \simeq  \KR_{\bg}(M,Q)$. In particular, the isomorphism class of the bundle
gerbe $\KR$-theory group depends only on the cohomology class of the Real Dixmier--Douady invariant. 
 By Proposition~\ref{modules}~(2) it follows that $\KR_{\bg}(M,P)
\simeq \KR(M)$ for any trivial Real bundle gerbe $(P,
Y)$. Furthermore, $\KR_{\bg}(M,P)$ is naturally a module over
$\KR(M)$ under tensor product with pullback of $\KR$-theory classes to
$Y$. More generally, if $(P,Y)$ and $(Q,X)$ are Real bundle gerbes on $M$, then there is a homomorphism $\KR_{\bg}(M,P)\otimes \KR_{\bg}(M,Q)\to\KR_{\bg}(M,P\otimes Q)$. One easily checks that $\KR_{\bg}(\cdot)$ is contravariant under pullback and is thus a well-defined functor from the category of Real spaces equipped with Real bundle gerbes to the category of abelian groups.

\begin{remark}\label{rem:K notation}
We note that just as in the case of complex twisted $K$-theory, the isomorphism $\KR_{\bg}(M,P) \simeq  \KR_{\bg}(M,Q)$ depends on a choice of stable isomorphism $P \simeq Q$.  The latter can be changed by  pullback and tensor product with a Real line bundle on $M$. Hence the isomorphism on $\KR$-theory is only defined up to the action on $\KR_{\bg}(M,Q)$ by the  Picard group of Real line bundles on $M$. When we have a specific Real bundle gerbe $(P, Y)$ representing a class $[H] \in H^2(M; \ZZ_2, \overline{\cU(1)})$,  we will often abuse notation and write $\KR_{\bg}(M, [H])$ for $\KR_{\bg}(M, P)$.\end{remark}

\begin{example}
If $M$ is a Real compact manifold with a trivial involution and $(P, Y)$ is a Real bundle
gerbe with $\tau$ acting trivially on $Y$, then the bundle gerbe
$\KR$-theory is related to the twisted $\KO$-theory defined in~\cite{MatMurSte}
via 
$$\KR_{\bg}(M,P) \simeq K\hspace{-1pt}O (M, [H])$$ where
$[H]=\DD_R(P) \in Tor(H^2(M,\ZZ_2))\subseteq H^2(M; \ZZ_2, \overline{\cU(1)})$. For this, we recall from Example \ref{ex:Z2} that in this case a  Real bundle gerbe reduces to a real gerbe when the involution acts trivially on the space, and similarly a Real vector bundle becomes an ordinary real vector bundle \cite{Ati}. It is straightforward to check that the gerbe action gives rise to a real bundle gerbe module and the result then follows by \cite[Proposition~7.3]{MatMurSte}.
\end{example}

\begin{example}
Let $N$ be any compact manifold and let $M=N\times\ZZ_2$ with the involution $\tau$ defined in Example~\ref{ex:double}. Recall that any Real bundle gerbe $(Q,Z)$ on $M$ is of the form $Q=(P,P^*)$, $Z=Y\times\ZZ_2$ for a bundle gerbe $(P,Y)$ on $N$, with $\tau_Q$ acting as $(P,P^*)\mapsto (P^*,P)$. The Real Dixmier--Douady class $\DD_R(Q)=(\DD(P),-\DD(P))$ is an element of the anti-diagonal subgroup $\ker(1\times\tau^*)$ of $H^2(M, \cU(1))= H^2(N, \cU(1))\oplus H^2(N, \cU(1))$. It follows that
$$
\KR_{\bg}(M,Q)\simeq K_{\bg}(N,P)\simeq K(N,[H])
$$
where $[H]=\DD(P)\in Tor(H^2(N, \cU(1)))\subseteq H^2(M;\ZZ_2,\overline{\cU(1)})$.
\end{example}

We recall from Proposition \ref{prop:PU(H)} that there is a bijective
correspondence between stable isomorphism classes of Real bundle
gerbes and isomorphism classes of Real principal $PU(\cH)$-bundles.
Every Real $PU(\cH)$-bundle determines a Real lifting bundle gerbe
with the same Real Dixmier--Douady class. Conversely, given any Real
bundle gerbe $(P, Y)$ and a Real bundle gerbe module $E\to Y$, the
projectivisation of $E$ descends to a Real projective bundle
$\mathcal{P}_E \to M$ due to the bundle gerbe action, and it is
straightforward to check that the class of the Real $PU(\cH)$-bundle
associated to $\mathcal{P}_E$ is $\DD_R(P)$. In the case of a torsion
class in $  H^2(M; \ZZ_2, \overline{\cU(1)})$, we have the Real
analogue of the Serre--Grothendieck Theorem (cf. \cite{AtiSeg,
  DonKar}).

\begin{theorem}[Real Serre--Grothendieck Theorem] \label{RSG} Any torsion class in $  H^2(M; \ZZ_2, \overline{\cU(1)} ) $ can be represented by a Real principal $PU(n)$-bundle.
\label{thm:RealSG}\end{theorem}
\begin{proof}
The torsion class can be represented by a Real bundle gerbe $(P, Y)$. The Dixmier--Douady class of this bundle gerbe is torsion so it is represented by a $PU(n)$-bundle.  This is equivalent to the bundle gerbe $(P, Y)$ admitting a rank $n$ bundle gerbe module $F \to Y$.  Define $E = F \oplus \tau^{-1}(\,\overline F\,)$. We show that $E$ is a Real bundle gerbe module for $(P, Y)$.  First note that there is the bundle gerbe module action
$$
P_{(y_1, y_2)} \otimes F_{y_2} \longrightarrow F_{y_1}
$$
and thus an induced action 
$$
P^*_{(y_1, y_2)} \otimes \overline F_{y_2} \longrightarrow \overline
F_{y_1} \ .
$$
Defining $E_y = F_y \oplus \overline F_{\tau(y)}$ we have 
\begin{align*}
P_{(y_1, y_2)} \otimes E_{y_2} &= P_{(y_1, y_2)} \otimes (F_{y_2} \oplus \overline F_{\tau(y_2)})\\[4pt]
&= (P_{(y_1, y_2)} \otimes F_{y_2}) \oplus  (P_{(y_1, y_2)} \otimes \overline F_{\tau(y_2)})\\[4pt]
&= (P_{(y_1, y_2)} \otimes F_{y_2}) \oplus  (P^*_{(\tau(y_1), \tau(y_2))} \otimes \overline F_{\tau(y_2)})\\[4pt] & \simeq F_{y_1} \oplus \overline F_{\tau(y_1)} \\[4pt]
 &= E_{y_1} \ .
\end{align*}
Clearly this is a bundle gerbe module action.  

Moreover we have 
$$
\tau^{-1}(E_y) = E_{\tau(y)} = F_{\tau(y)} \oplus \overline F_{y}
$$
and flipping elements maps this complex linearly to
$$
 \overline F_{y} \oplus F_{\tau(y)}  = \overline E_y
 $$
 so that $E$ is a Real bundle gerbe module.
 
 The existence of the Real bundle gerbe module implies that the Real Dixmier--Douady class of  $(P, Y)$
 is associated to a Real principal $PU(n)$-bundle. 
 \end{proof}
 
\begin{remark}
If $(P, Y)$ is a Real bundle gerbe, then we have constructed a map from the twisted $K$-theory with respect to the underlying $U(1)$-bundle gerbe to Real twisted $K$-theory of $(P, Y)$. This is a generalisation of the corresponding construction from \cite[p.~371]{Ati} in the untwisted case.

Notice also that this proof does not actually use the fact that the
class in $H^2(M; \ZZ_2, \overline{\cU(1)})$ is torsion, but rather that its image in
$H^3(M, \ZZ)$ is torsion. So we have proved that every class in $H^2(M; \ZZ_2, \overline{\cU(1)})$ which is torsion in $H^3(M, \ZZ)$ is actually torsion in
$H^2(M; \ZZ_2, \overline{\cU(1)})$ by Remark~\ref{remark:torsion}.
\end{remark}

\subsection{Twisted $\KR$-theory}\label{sec:twistedKRtheory}
 
We now turn to the problem of showing that the bundle gerbe
$\KR$-theory of a Real bundle gerbe on a compact Real manifold $M$ is in fact the same as the twisted $\KR$-theory of $M$.

Let $\cH$ be a complex separable Hilbert space with a conjugation $v \mapsto \bar v$ and take it to be $\ZZ_2$-stable, that is the  one-dimensional irreducible representation of $\ZZ_2$ occurs in $\cH$ with infinite multiplicity. The space of Fredholm operators $\Fred(\cH)$ acquires a natural involution
defined by $\sigma(T)( v) = \overline{T(\bar{v})}$ for all $T\in
\Fred(\cH)$ and $v\in \cH$. In order to get a representing space for $\KR$-theory, with a continuous action by the Real projective unitary group in the compact-open topology, we proceed as in \cite{AtiSeg} and replace $\Fred(\cH)$ by another space of Fredholm operators $\Fred^{(0)}(\hat\cH)$ where $\hat\cH = \cH \otimes \CC^2$. We refer to Section 3 in \cite{AtiSeg} for a detailed description of this space and its topology. The involution on $\Fred(\cH)$ extends naturally to $\Fred^{(0)}(\hat\cH)$ and the Real projective unitary group
$PU(\hat \cH)$ acts continuously on $\Fred^{(0)}(\hat\cH)$ by conjugation. For a
Real $PU(\hat\cH)$-bundle $\cP\to M$ we can form the associated Real
bundle $\Fred^{(0)}_{\mathcal{P}}= \mathcal{P}\times_{PU(\hat\cH)}
\Fred^{(0)}(\hat\cH)$ classified by its invariant $\DD_R(\mathcal{P})
\in   H^2(M; \ZZ_2, \overline{\cU(1)} ) $.  The twisted $\KR$-theory
group of $M$ is defined as the group of  Real homotopy classes of
continuous sections of
$\Fred^{(0)}_\mathcal{P}$,\footnote{\label{fn:admissible} For locally compact Real manifolds, this definition still works by restricting to Real \emph{admissible} sections, see Definition 2.1 in \cite{CareyWang2}.}
\begin{equation}
\KR(M, \mathcal{P}) =
\pi_0\big(\Gamma_M^{\ZZ_2}(\Fred_\mathcal{P}^{(0)}) \big) \ ,
\label{eq:tachyondef}\end{equation}
or equivalently as the space of all homotopy classes of $ PU(\hat\cH)\rtimes\ZZ_2$-equivariant maps
$$
\KR(M, \mathcal{P}) =  [\mathcal{P},\Fred^{(0)}(\hat\cH)]_{ PU(\hat\cH)\rtimes\ZZ_2} \ ,
$$
where the homotopies are through equivariant maps and the $\ZZ_2$-action is via the Real structures on the spaces.
At this point we abuse notation again as in Remark \ref{rem:K notation} and write 
$$
\KR(M, [H]) = \KR(M, \mathcal{P}) = [\mathcal{P},\Fred^{(0)}(\hat\cH)]_{ PU(\hat\cH)\rtimes\ZZ_2} \ ,
$$
where $\cP$ is chosen such that $\DD_R(\cP) = [H]$. 

\begin{theorem}\label{theorem}\label{K-isom}
Let $M$ be a Real compact manifold, $\cP$ be a Real $PU(n)$-bundle over $M$ with torsion Real Dixmier-Douady class $[H]$  and $(L_\cP,\mathcal{P})$ denote the corresponding lifting bundle gerbe. 
Then there is an isomorphism of abelian groups
$$
\KR_{\bg}(M,L_\cP) \simeq \KR(M, [H]) \ .
$$
\end{theorem}
\begin{proof} First we note that by Theorem~\ref{RSG}, there always
  exists a Real $PU(n)$-bundle $\mathcal P \to M$ associated to any torsion Real 
  Dixmier--Douady class $[H]\in   H^2(M; \ZZ_2,
  \overline{\cU(1)})$. The proof  proceeds along the same lines as the
  proof of Proposition 3.2 in \cite{CareyWang2}, although we only
  consider the compact-open topology which by Appendix 3 in
  \cite{AtiSeg} is equivalent to using the norm topology. The basic idea is to identify the additive category of Real bundle gerbe modules on $M$ with the additive category of Real $U(n)$-equivariant vector bundles of Real central character $1$ on the Real compact manifold $\cP$. 

Namely, any element of the center $g\in U(1)\subset U(n)$ gives rise to a $U(n)$-equivariant vector bundle automorphism $g_E\colon E\to E$, which is a central character if $g_E = \chi(g)\id_E$. As explained at the end of Section 6.2 in \cite{BouCarMat},  the bundle gerbe multiplication implies that $\chi(g)=1$, i.e. center  must act by scalar multiplication. It is further straightforward to check that the compatibility between the bundle gerbe multiplication and the Real structure in Definition \ref{realgerbemodule} corresponds precisely to the compatibility between the Real structure and the $U(n)$-action on  vector bundles on $\cP$. Therefore, if follows that $\KR_{\bg}(M,L_\cP)$ is isomorphic to the submodule $\KR_{U(n),(1)}(\cP) \subset \KR_{U(n)}(\cP)$ of weight $1$, where $\KR_{U(n)}(\cP)$ is regarded as an $RR(U(1))$-module and $RR(U(1))\subset \KR_{U(1)}(\pt)$ is the Real representation ring of $U(1)$.
 
Next we can view $\cP$ as a reduction of
  a Real $PU(\hat\cH)$-bundle $\tilde{\mathcal{P}}$ with the same
  Dixmier--Douady invariant, via the embedding $PU(n) \to PU(\CC^n
  \otimes \hat \cH)$, $g \mapsto g\otimes 1$ and by choosing an isomorphism $\CC^n \otimes \hat \cH \simeq \hat \cH$ of Real Hilbert spaces, where the involution on $\CC^n$ is by complex conjugation. We further require that $\hat\cH$ is a Real stable $U(n)$-Hilbert space and denote by $\hat\cH_{(1)}$ the Real $U(n)$-Hilbert subspace of weight $1$ under the Real U(1)-action. By the results in Section 6 and Appendix 3 of \cite{AtiSeg}, it follows that 
$$\KR_{U(n),(1)}(\cP) = [\cP, \Fred^{(0)}(\hat\cH_{(1)})]_{PU(n)\rtimes \ZZ_2 }$$
and we have
\begin{align*}
\KR_{\bg}(M,L_\cP) &\simeq \KR_{U(n),(1)}(\cP) \\
&=  [\cP, \Fred^{(0)}(\hat\cH_{(1)})]_{PU(n)\rtimes \ZZ_2 }\\
&\simeq[\cP\times_{PU(n)} PU(\hat\cH_{(1)}), \Fred^{(0)}(\hat\cH_{(1)})]_{PU(\hat\cH_{(1)})\rtimes \ZZ_2} \\
&\simeq \KR(M, [H]). 
\end{align*}
 
\end{proof}

\begin{remark}\label{RealSS}
We note that  a similar argument as in \cite{Seg}  shows that every Real bundle gerbe module is a direct summand of a trivial Real bundle gerbe module, so it follows that every element in
$\KR_{\bg}(M,L_\cP) $ can be represented in the form $[E]-[\CC^N_\cP]$ where
$\CC_\cP^N$ is the trivial Real vector bundle on $\cP$. 
\end{remark}

We now sketch the generalisation of this construction to bigraded
$\KR$-theory groups. For this, let $e_{p,q} \colon \RR^n\to\RR^n$ be the involution acting on $(x,y)\in
\RR^{p}\times\RR^q$ as $(x,y)\mapsto
(x,-y)$, where $p+q=n$; we denote the Real space $\RR^n$ with this involution as
$\RR^{p,q}$. Let $\CC\ell(n)$ be
the complex $\ZZ_2$-graded Clifford $C^*$-algebra on $n$
generators $e_1,\dots,e_{n}$ of degree one with the relations
$$
e_i\,e_j+e_j\, e_i= -2\, \delta_{ij} \ ,
$$
together with the linear embedding of $\RR^{n}$ into $\CC\ell(n)$
which sends the
standard basis of $\RR^{n}$ to $e_1,\dots,e_{n}$. The involution
$e_{p,q} \colon  \RR^{n}\to \RR^{n}$ induces an involutive automorphism of
$\CC\ell(n)$, also denoted $e_{p,q}$, and the corresponding Real algebra
$\CC\ell(n)$ is denoted $C\ell(\RR^{p,q})$.

Let $\hat\cH$ be a $\ZZ_2$-graded Real separable Hilbert space which is a
$*$-module over the Real Clifford algebra $C\ell(\RR^{p,q})$; we assume that each simple subalgebra of $C\ell(\RR^{p,q})$ is represented with infinite multiplicity on $\hat\cH$. Let
$\Fred^{(0)}_{p,q}(\hat\cH)$ be the Real space of Fredholm operators of odd
degree on
$\hat\cH$ which commute with the $C\ell(\RR^{p,q})$-action and topologised as in \cite{AtiSeg}; it is a
classifying space for the bigraded $\KR$-theory $\KR^{p,q}$. Let
$PU_{p,q}(\hat\cH)\subseteq PU(\hat\cH)$ be the subgroup of projective
unitaries commuting with the $C\ell(\RR^{p,q})$-action. Then $PU_{p,q}(\hat\cH)$ preserves $\Fred^{(0)}_{p,q}(\hat\cH)$. The
bigraded $(p,q)$ twisted $\KR$-theory group of $M$ is defined for a Real principal
$PU_{p,q}(\hat\cH)$-bundle $\cP\to M$ by
$$
\KR^{p,q}(M,\DD_R(\mathcal{P})) =
[\mathcal{P},\Fred^{(0)}_{p,q}(\hat\cH)]_{PU_{p,q}(\hat\cH)\rtimes
 \ZZ_2 }
$$
as above.

Let $\pi_{p,q} \colon M\times\RR^{p,q}\to M$ be the projection and define
$$
\KR_{\bg}^{p,q}(M,P):= \KR_{\bg}(M\times\RR^{p,q}, \pi_{p,q}^{-1}(P))
$$
where $\pi_{p,q}^{-1}(P)$ is the pullback Real bundle gerbe over the Real space
$M\times\RR^{p,q}$ and we are implicitly using $\KR$-theory with compact support, see Footnote~\ref{fn:admissible}. Then we immediately deduce from
Theorem \ref{K-isom} that
$$
\KR_{\bg}^{p,q}(M,P)\simeq \KR^{p,q}(M,\DD_R(P))
$$
for all $(p,q)$. This identifies $\KR_{\bg}^{p,q}(M,P)$ as the group
of virtual Real bundle gerbe modules with an
action of the Real Clifford algebra $C\ell(\RR^{p,q})$.

\begin{remark}
In the case of a non-torsion Dixmier--Douady class, it is possible to
introduce a Real analogue of infinite-rank $U_{\cK}$-bundle gerbe
modules as in \cite{BouCarMat}. We leave the formulation to the
reader. We will in fact see in Section~\ref{sec:Geometric
  cycles} that for the construction of geometric cycles for
$\KR$-homology twisted by an arbitrary Dixmier--Douady class, only finite-rank Real bundle gerbe modules are required.
\end{remark}
 
\section{Orientifolds and Real bundle gerbe D-branes}
\label{sec:orientifolds}

In this section we describe how our Real bundle gerbe
constructions find applications in the orientifold construction of
Type~II string theory, which includes Type~I string theory. In particular,
our bundle gerbe $\KR$-theory provides an appropriate receptacle for
the quantization of Ramond-Ramond charges and fluxes on these
backgrounds in a manner that we explain below. Our considerations here motivate the definition of twisted $\KR$-homology that we give in Section~\ref{sec:Geometric cycles}. The reader uninterested in the physics background behind our constructions may safely skip this section.

In the following, by a ``$B$-field'' we mean a gerbe with connection or a
class in a suitable differential cohomology theory as specified for
example in \cite{DFM, DFM1}. By ``quantum flux'' we mean the Dixmier--Douady class
of this gerbe: in string theory the $H$-flux usually refers to the $3$-form
curvature of the gerbe with connection, but the key feature is that it
represents the Dixmier--Douady class so has integer periods, and that is
what we shall mean by ``quantum''.

\subsection{D-branes and anomalies}

Let us begin by reviewing the well-known case
without involution, see e.g.~\cite{Szabo}, recast
into the context of this paper. We interpret our
manifold $M$ as spacetime of Type~II string theory which comes with
various geometric fields $F$, such as a Riemannian metric $g$ and a $B$-field whose three-form
flux $H$ defines a class $[H]\in H^3(M,\ZZ)$ by (generalised)
Dirac charge quantization~\cite{DFM1}; we can take $[H]=\DD(P)$ to be the Dixmier--Douady
class of a bundle gerbe $(P,Y)$ over $M$. In the worldsheet theory,
these fields are given by background functions $F(\phi(x))$ of closed string field configurations which are specified by a closed oriented Riemann surface $\Sigma$
and a smooth map $\phi \colon \Sigma\to M$. The string sigma-model associates to this data an exponentiated Euclidean action functional, one of whose factors is the amplitude
\begin{equation}\label{eq:Bfieldampl}
A_{g,H}(\phi,\Sigma) = \exp\big(-S_{\rm kin}(\phi)\big) \ \hol(\Sigma,\phi^*H) \ ,
\end{equation}
where $S_{\rm kin}(\phi)=\frac12\,\int_\Sigma\, \|{\rm d}\phi\|^2$ is the kinetic term which involves the orientation and a conformal structure on $\Sigma$ as well as the metric on $M$; in this generality the Wess--Zumino--Witten term $\hol(\Sigma,\phi^*H)$ from \eqref{eq:WZW} is usually called the \emph{$B$-field amplitude}.

If $\Sigma$ has a boundary, then
one needs to specify suitable boundary conditions for the maps $\phi \colon \Sigma\to M$ which are represented by a choice of the additional geometric data of a submanifold $f \colon Z\hookrightarrow M$ such that 
$\phi(\partial\Sigma)\subseteq Z$; this submanifold specifies the
\emph{worldvolume} of a wrapped D-brane. The open string field
configurations on the D-brane include a ``bundle'' $E$ on $Z$, which is its \emph{Chan--Paton
bundle}; we shall clarify its precise geometric meaning presently.

General
considerations from string theory imply that $E$ is not always a complex
vector bundle on $Z$ but should be more precisely described as defining a class $[E]$ in the
$K$-theory of $Z$ twisted by the class $f^*[H]+W_3(\nu)\in
H^3(Z,\ZZ)$. Here  the $2$-torsion class
$W_3(\nu)\in H^3(Z,\ZZ)$ is the third integral Stiefel--Whitney
class of the normal bundle $\nu\to Z$, which is the obstruction to a spin$^c$ structure on $\nu$
and will be regarded as the Dixmier--Douady class of the corresponding lifting
bundle gerbe $L_\nu$~\cite{Mur} associated
to the central extension 
$$
1 \ \longrightarrow \ U(1) \ \longrightarrow \ 
Spin^c(r) \ \longrightarrow \ SO(r) \ \longrightarrow \ 1
$$
where $r$ is the codimension of $Z$ in $M$. This is due to the Freed--Witten anomaly~\cite{FreedWitten} in the string
sigma-model associated to the space of smooth maps $\phi \colon \Sigma\to M$, that is, a factor of the exponentiated action which takes values in a (non-canonically trivialised) line bundle rather than $\CC$. Then the induced \emph{Ramond-Ramond
  charge} is computed by pushforward $f_! \colon
K_{\bg}(Z,f^*[H]+W_3(\nu))\to K(M,[H])$ under the inclusion $f
\colon Z\hookrightarrow M$~\cite{Szabo}, where $f^*[H]+W_3(\nu)$ is
the Dixmier--Douady class of the bundle gerbe $f^{-1}(P) \otimes L_\nu$. For vanishing $H$-flux and when the D-brane
is a stack of identical D-branes wrapping $Z$, the
Chan--Paton bundle $E$
can be regarded as a bundle gerbe module of rank $n$ for this lifting bundle gerbe
with $[E]\in K_{\bg}(Z,W_3(\nu))$; in particular, for a single D-brane
$n=1$ the complex line bundle $E\to\nu$ provides a trivialization for the
lifting bundle gerbe $L_\nu$ and describes a spin$^c$ structure on the 
normal bundle $\nu\to Z$, as expected in this case~\cite{FreedWitten}.

Anomaly free D-branes wrapping $Z$ satisfy the constraint~\cite{Kapustin,CarJohnMur}
$$
f^*[H] +W_3(\nu)=\beta\big(y(E)\big)
$$
in $H^3(Z,\ZZ)$,
where the \emph{'t~Hooft flux} $y(E)\in H^2(Z,\ZZ_n)$ is the obstruction to
an $SU(n)$-structure on the principal bundle associated to the corresponding projective vector bundle $\cP_E\to Z$, which may be regarded as the Dixmier--Douady class of the corresponding lifting
bundle gerbe associated to the central extension 
\begin{equation}
1\ \longrightarrow \ \ZZ_n\ \longrightarrow \ 
SU(n)\ \longrightarrow \ PU(n)\ \longrightarrow \ 1 \ , 
\label{eq:Zncentral}\end{equation}
and $\beta \colon H^2(Z,\ZZ_n)\to H^3(Z,\ZZ)$ is the Bockstein homomorphism associated to the exponential sequence $0\to\ZZ\xrightarrow{\times n} \ZZ\to\ZZ_n\to1$. For $n=1$ this is precisely the condition that the normal bundle $\nu\to Z$ admits an
$H$-twisted spin$^c$ structure~\cite{Wang}; in this case the Chan--Paton
bundle $E$ is a $\ZZ_2$-graded vector bundle on $Z$ with class $[E]\in
K(Z)$. For vanishing $H$-flux the anomaly cancellation condition for $n=1$ reduces to $W_3(\nu)=0$
and, when $M$ is spin, the worldvolume $Z$ is a spin$^c$ manifold.

Another way to deal with the anomaly is to maintain 
the requirement that the worldvolume $Z$ is a spin$^c$ manifold;
this ensures that the choice of boundary conditions represented by the
D-brane preserves a certain amount of supersymmetry in the string
sigma-model on the space of maps $\phi \colon \Sigma\to M$. In this case $[E]\in
K_{\bg}(Z,f^*[H])$, and combined with anomaly cancellation we then arrive at

\begin{definition}
A \emph{bundle gerbe D-brane} of a bundle gerbe
$(P,Y)$ over $M$ is a triple $(Z,E,f)$, where $f \colon Z\hookrightarrow M$
is a closed, embedded spin$^c$ submanifold and $E$ is a bundle gerbe
module of rank $n$ for the bundle gerbe $f^{-1}(P,Y)$ on $Z$.
\label{def:bgDbrane}\end{definition}

Note that this definition does not require the quantum $H$-flux on $M$ to be
torsion, but rather only that $n\,[H]\in\ker(f^*)\subseteq H^3(M,\ZZ)$; in particular, for a single
D-brane $n=1$ the Chan--Paton bundle $E\to f^{-1}(Y)$ gives a trivialization of the bundle gerbe $f^{-1}(P,Y)$. A similar notion of D-brane
was considered in~\cite{CareyWang}.

Deformation invariance, gauge symmetry enhancement and the possibility
of branes within branes imply that any bundle gerbe D-brane $(Z,E,f)$
should be subjected to the usual equivalence relations of geometric
$K$-homology~\cite{BaumDouglas}: bordism, direct sum and vector
bundle modification, respectively~\cite{ReisSzabo}. For the
topological classification of bundle gerbe D-branes, however,
we need to consider a larger class of triples wherein the spin$^c$ submanifold
$Z\subseteq M$ is generalised to an arbitrary continuous map $f \colon Z\to M$; non-embeddings
$f \colon Z\to M$ correspond to ``non-representable'' D-branes which are
physically significant in the correspondence between D-branes and
$K$-homology, see~\cite{ReisSzabo}. Geometric twisted $K$-homology is
defined in the present context by~\cite{Liu};
see~\cite{MathaiSinger,BCW,DeeleyGoffeng} for related approaches based on projective
bundles.

\subsection{Orientifold constructions}

Let us now apply the \emph{orientifold} construction to this setting,
which introduces an involution $\tau$ making $M$ into a
Real manifold. The connected components of the fixed point set $M^\tau$ of the orientifold involution
are called {orientifold planes}, or \emph{O-planes} for short.
In the worldsheet theory, the compact Riemann surface $\Sigma$ is not
oriented and need not even be orientable. The string fields $\phi$
should now be regarded as smooth maps from $\Sigma$ to the orbifold quotient of
$M$ by the involution $\tau$, which represents the
physical points of the orientifold spacetime. To make this precise,
following~\cite{SchSch,DFM} we introduce the orientation
double cover $\hat\pi \colon \widehat\Sigma\to\Sigma$ corresponding to the first
Stiefel--Whitney class $w_1(\Sigma)\in H^1(\Sigma,\ZZ_2)$; it is
canonically oriented with a canonical orientation-reversing involution
$\Omega \colon \widehat\Sigma\to\widehat\Sigma$, called \emph{worldsheet parity}, which permutes the sheets and
preserves the fibres. The string fields are
then smooth maps $\hat\phi \colon \widehat\Sigma\to M$ which are
equivariant in the sense that there is a commutative diagram
$$
\xymatrix{
\widehat\Sigma \ \ar[r]^{\hat\phi} \ar[d]_\Omega & \ M \ar[d]^\tau \\
\widehat\Sigma \ \ar[r]_{\hat\phi} & \ M
}
$$

Because of the orientation-reversing involution $\Omega$, the
geometric fields $F$ on $M$, which are background functions $F(\hat\phi(x))$ of the maps $\hat\phi \colon \widehat\Sigma\to M$, are required to satisfy equivariance conditions
under $\tau$ in order to survive to the orientifold quotient. In particular, the analog of the amplitude from \eqref{eq:Bfieldampl},
$$
\widehat{A}_{g,H}(\hat\phi,\Sigma) := \exp\big(-S_{\rm kin}(\hat\phi\,)\big) \ \big(\hol(\,\widehat{\Sigma},\hat\phi{}^*H)\big)^{1/2} \ ,
$$
involves $w_1(\Sigma)$-twisted forms, that is, forms on $\widehat{\Sigma}$ which are anti-invariant under pullback by $\Omega$ (cf.~\cite{SchSch,GSW,DFM1}). We require that $\widehat{A}_{g,H}(\hat\phi,\Sigma)$ be invariant under the combined actions of the involutions $\Omega$ and $\tau$; this forces the metric to be invariant, $\tau^*(g)=g$ (to ensure that the kinetic term $S_{\rm kin}(\hat\phi\,)$ is invariant) whereas the
three-form flux of the $B$-field is anti-invariant, $\tau^*(H)=-H$ (ensuring that the $B$-field amplitude $\hol(\,\widehat{\Sigma},\hat\phi{}^*H)$
is invariant). By Dirac charge quantization, the $H$-flux thus
determines a class 
$[H]\in\ker(1 \times \tau^*)\subseteq H^3(M,\ZZ)$. Recalling the
discussion from Section~\ref{sec:HR}, this is a necessary condition
for $[H]$ to lift to a class in $H^2(M; \ZZ_2, \overline{\cU(1)})$, but it is not
sufficient: In general the
vanishing condition must be imposed in equivariant cohomology as
dictated by the long exact sequence
\eqref{eq:beginlongsequence}. A Real bundle gerbe connection
whose $3$-curvature $H$ obeys $\tau^*(H)=-H$ renders the orientifold
$B$-field amplitude invariant, but to obtain a Real structure on a
given bundle gerbe with Dixmier--Douady class $[H]$ typically requires
assumptions on the topology of spacetime $M$; a situation where this
occurs is provided by the tautological bundle gerbe of
Example~\ref{ex:tautological}. A precise definition of connective structures on Real bundle gerbes and their holonomy will be provided in \cite{HMSV}. For the purposes of the present discussion, we offer
\begin{definition}
An \emph{orientifold $B$-field} is a $B$-field on $M$ whose quantum flux
$[H]\in \ker(1 \times \tau^*)\subseteq H^3(M,\ZZ)$ has a lift to the equivariant cohomology
$  H^2(M;\ZZ_2, \overline{\cU(1)})\simeq H^3_{\ZZ_2}(M, \ZZ(1))$.
\end{definition}
This definition agrees with those of~\cite[Section 6]{GSW}, \cite[Definition 2]{DFM} and \cite[Section 3.2]{DFM1}.   
 From the discussion above we have
\begin{proposition}
A $B$-field on $M$ is an orientifold $B$-field if and only if $(1\times\tau^*)[H]$ vanishes as a class in $H^2(M;\ZZ_2, \cU(1))\simeq H^3_{\ZZ_2}(M,\ZZ)$.
\label{prop:liftingH}\end{proposition}
In this case we can
take $[H]$ to be the Real Dixmier--Douady class $\DD_R(P)$ of a Real
bundle gerbe on $M$. Hence spacetime $M$ is
now endowed with a Real bundle gerbe~$(P,Y)$.
\begin{remark}
Arguing similarly to Section~\ref{sec:HR} via the Cartan--Leray spectral sequence, the free part of the Borel equivariant cohomology $H^3_{\ZZ_2}(M,\ZZ)$ is isomorphic to the invariants in ordinary cohomology $H^3(M,\ZZ)^{\ZZ_2}$. The class $(1\times\tau^*)[H]$ is automatically invariant, so the lifting condition of Proposition~\ref{prop:liftingH} on its free part is the necessary condition $\tau^*[H]=-[H]$ in $H^3(M,\ZZ)$ which usually appears in the string theory literature. However, the vanishing of $(1\times\tau^*)[H]$ in the torsion subgroup of $H^2(M;\ZZ_2, \cU(1))$  is required to obtain a sufficient condition.
\end{remark}

Similarly to the previous situation, we specify a submanifold
$Z\subseteq M$ such that the
string fields $\hat\phi\colon \widehat{\Sigma}\to M$ satisfy the boundary condition
$\hat\phi(\partial\widehat\Sigma)\subseteq Z$. Demanding that a D-brane be isomorphic to its orientifold image (in a suitable sense) defines the orientifold projection of open string states. We assume that $Z$ is preserved by $\tau$; for topological considerations a natural choice is to take $Z\subseteq M^\tau$ to coincide with an O-plane. Again the open string
field configurations on the D-brane include a bundle gerbe module $E$
for some Real bundle gerbe on $Z$ which we will specify momentarily;
the worldsheet parity involution $\Omega$ induces a map
$E\to\overline{E}$. Equivariance requires that there be an isomorphism
$\tau_E \colon \tau^{-1}(E)\to \overline{E}$ satisfying $(\tau_E\circ\tau^{-1})^2=1$,
hence $E$ is naturally a Real bundle gerbe module and defines an element in some twisted $\KR$-theory group.

At present there is no computation of the Freed--Witten anomaly
available for orientifold (or even orbifold) string theories. However,
we can glean it from the definition of Real twisted spin$^c$
structures given by Fok~\cite{Fok2015}---which we generalise and extend
in Section~\ref{sec:Geometric cycles}---and by demanding that the
induced Ramond-Ramond charges can be computed by suitable pushforward
to classes in the twisted $\KR$-theory $\KR(M,[H])$ under
the inclusion $f \colon Z\hookrightarrow M$; this pushforward will be constructed
explicitly in Section~\ref{sec:Geometric cycles}, see in particular Example~\ref{ex:RRcharge}. Then our Chan--Paton bundles 
will generically be bundle gerbe modules defining classes in
$\KR_{\bg}(Z,f^*[H]+ \WR_3(\nu))$, where $\WR_3(\nu)$ is the Real
Dixmier--Douady invariant of the Real lifting bundle gerbe corresponding to the normal bundle $\nu\to Z$ which is the obstruction to a Real spin$^c$
structure on $\nu$. (cf. also \cite[Remark (g)]{DFM}).

We shall elucidate these definitions and the precise meaning of this obstruction in some
detail in Section~\ref{sec:Geometric cycles}. Again the open string field configurations on the D-brane include a class in $H_{\ZZ_2}^3(Z,\ZZ(1))$ associated with \eqref{eq:Zncentral}, regarded now as a Real central extension, and by equating twisting classes as before we generalize Definition~\ref{def:bgDbrane} to
\begin{definition}\label{def:RealDbrane}
A \emph{Real bundle gerbe D-brane} of a Real bundle gerbe
$(P,Y)$ over the Real manifold $M$ is a triple $(Z,E,f)$, where $f \colon Z\hookrightarrow M$
is a closed, embedded Real spin$^c$ submanifold such that $\tau(Z)= Z$ and $E$ is a Real bundle gerbe
module of rank $n$ for $f^{-1}(P,Y)$.
\end{definition}
A similar definition of D-brane is given by~\cite{GSW} using Jandl gerbes. In Section~\ref{sec:Geometric cycles} we will generalize this definition in the category of Real spaces by considering arbitrary continuous equivariant maps $f \colon Z\to M$ between Real spaces, and defining 
geometric cycles for twisted $\KR$-homology by suitably combining them
with an equivariant construction. For vanishing quantum $H$-flux, equivariant geometric $K$-homology is
constructed in~\cite{BHSZ2,SzVal}; this definition is extended
to the twisted case by~\cite{Barcenas} in the
language of $PU(\cH)$-bundles. 

\begin{example}[Discrete torsion]\label{discretetorsion}
The difference between orientifold
group actions on a fixed $B$-field is known as \emph{discrete torsion}. In our setting, the
orientifold {discrete torsion} is parameterized by $  H^2(M;\ZZ_2, \overline{\cU(1)})$ via the map
$$
  H^2(M;\ZZ_2, \overline{\cU(1)}) \ \longrightarrow \ {\rm
  Tors}\big(H^3_{\ZZ_2}(M, \ZZ(1))\big) \ .
$$
In particular, the subgroup of discrete
$B$-fields is classified by the inclusion $  H^2({\rm pt}; \ZZ_2,
\overline{\cU(1)})\subset H^2(M;\ZZ_2, \overline{\cU(1)})$ under pullback by the projection $M\to{\rm pt}$; in Example~\ref{ex:gerbept} we gave an explicit construction of the non-trivial discrete
$B$-field in $  H^2({\rm pt};\ZZ_2, \overline{\cU(1)})\simeq \ZZ_2$, reflecting the fact even a point has over it a non-trivial Real bundle gerbe. Alternatively, it is classified by the equivariant cohomology $H^3_{\ZZ_2}(\pt, \ZZ(1))$, which is computed 
by~\cite{Fok2015} to be
$$
H^3_{\ZZ_2}(\pt, \ZZ(1))=\ZZ_2 \ .
$$
This coincides with the group
cohomology $H^2(B\ZZ_2,\overline{\cU(1)})\simeq H_{\rm gp}^3(\ZZ_2,{\ZZ}(1))$, which also classifies
non-central extensions of the orientifold group $\ZZ_2$ by $U(1)$, or equivalently projective Real representations of $\ZZ_2$~\cite{DFM,BraunStef};
this is used by~\cite{BraunStef} to provide projectivised group actions
on D-branes and a definition of twisted $\KR$-theory in terms of
projective Real vector bundles for torsion quantum $H$-flux in this subgroup. For the two inequivalent Real structures on the
trivialisable gerbe here, the corresponding twisted $\KR$-theory groups are $\KO$ and $\KO^4=\KSp$ (this is also a special case of~\cite[Proposition~3.29]{Fok2015}); more generally, the non-trivial projective Real representation of $\ZZ_2$ is a Real representation of the cyclic group $\ZZ_4$ and the $\KR$-theory twisted by the generator $\xi$ of $  H^2({\rm pt}; \ZZ_2, \overline{\cU(1)})\subset H^2(M;\ZZ_2, \overline{\cU(1)})$ can be computed from the equivariant $\KR$-theory $\KR_{\ZZ_4}(M)=\KR(M)\oplus\KR(M,\xi)$ for any Real manifold $M$~\cite{BraunStef}.
\end{example}

\begin{example}[Type I D-branes]
Consider the Real involution $\tau$ that acts trivially on
$M$; this is the receptacle for 
Type~I string theory. The $B$-field reduces to a discrete field with quantum flux $[H]\in H^2(M,\ZZ_2)\oplus \ZZ_2$ by Example~\ref{ex:trivialtau}, and
the Chan--Paton bundles $E$ now define classes $[E]$ in
$\KO_\bg(Z,f^*[H]+w_2(\nu))$ or $\KSp_\bg(Z,f^*[H]+w_2(\nu))$ corresponding to $\pm\,1\in \ZZ_2$, respectively, where $w_2(\nu)\in H^2(Z,\ZZ_2)$ is the
second Stiefel--Whitney class of the normal bundle $\nu\to Z$. Thus in this case we recover Type~I D-branes which support either an orthogonal or symplectic gauge theory. For vanishing quantum $H$-flux,
geometric $\KO$-homology is constructed
in~\cite{BaumHigsonSchick,ReisSzaboValentino}.
\end{example}

\begin{example}[D-branes in $S^{1,3}$]
Let $M=S^{1,3}$ be the unit sphere in $\RR^{1,3}$, or equivalently the Lie group $SU(2)\simeq S^3$ with group inversion $g\mapsto g^{-1}$ as Real structure. The orientifold fixed point set consists of two elements, the
identity and its negative which comprise the center of $SU(2)$, corresponding respectively to the north and south poles $(\pm\,1,0,0,0)\in\RR^{1,3}$. By Example~\ref{ex:RealH2G} we have
$$
H^2(S^{1,3};\ZZ_2,\overline{\cU(1)})=\ZZ_2\oplus\ZZ \ ,
$$
where the basic gerbe over $SU(2)$ is $(0,1)$ while the gerbe coming from the coboundary map on $H^1(S^{1,3};\ZZ_2, \cU(1))=\ZZ_2$  is $(-1,0)$. Generally, \emph{symmetric D-branes} in Lie groups correspond to (integral) conjugacy classes~\cite{Alekseev}. For $SU(2)$ the conjugacy class of an element corresponding to $(x,y)\in\RR^{1,3}$ is the intersection of $S^3$ with a hyperplane with fixed first coordinate $x\in\RR$. For any $x\neq\pm\,1$ these are two-spheres $S_x^2\subset\RR^{0,3}$ which are preserved by the involution $e_{1,3}$ and are Real spin$^c$, and by Example~\ref{ex:RealH2S2} we have $H^2(S_x^2;\ZZ_2,\overline{\cU(1)})=0$; hence these conjugacy classes correspond to single (rank~$1$) spherical Real bundle gerbe D2-branes. For $x=\pm\,1$ the conjugacy classes correspond to Real bundle gerbe D-particles sitting at the O0-planes which can support either real or symplectic bundles since $H^2({\rm pt}; \ZZ_2, \overline{\cU(1)})= \ZZ_2$. These results are in agreement with those of~\cite{Brunner,Huiszoon,Bachas}.
\end{example}

\section{Real bundle gerbe cycles and twisted $\KR$-homology}
\label{sec:Geometric cycles}

In this section we define cycles for a
geometric realisation of the homology theory dual to the bundle gerbe 
$\KR$-theory constructed in this paper. Amongst other things,
this will provide the topological classification of the Real bundle
gerbe D-branes discussed in Section~\ref{sec:orientifolds}.

\subsection{Real spin$^c$ structures}

Let $Z$ be a Real space, and let $V\to Z$ be an equivariant oriented 
real vector bundle of even rank $n$ equipped with a fibrewise inner product
with
respect to which the involution $\tau_V \colon V\to V$ is orthogonal. The bundle ${\sf F}(V)$ of oriented
orthonormal frames of $V$ is a principal $SO(n)$-bundle on $Z$. Following~\cite{Fok2015}, we can make its
structure group $SO(n)$ into a Real Lie
group $SO(\RR^{p,q})$ by assigning the involutive automorphism $g\mapsto \sigma_{p,q}(g)= e_{p,q}\, g \, e_{p,q}$ of $SO(n)$ which corresponds to
the involution $e_{p,q} \colon \RR^n\to\RR^n$ introduced in Section~\ref{sec:twistedKRtheory}, where $p+q=n$. Note that $e_{p,q}\in SO(n)$ if $q$ is even and
$e_{p,q}\in O(n)$ if $q$ is odd; moreover $e_{0,n}$ acts trivially and
$e_{0,n}\, e_{p,q}=e_{q,p}$, so that $\sigma_{p,q}=\sigma_{q,p}$.
\begin{definition}
The vector bundle $V$ is \emph{Real $(p,q)$-oriented} if its frame bundle ${\sf F}(V)$ is a Real $SO(\RR^{p,q})$-bundle.
\end{definition}
Let us examine some necessary and sufficient conditions under which $V$ admits a Real
$(p,q)$-orientation in this sense.
For this, let ${\sf F}(V)^{\sigma_{p,q}}\to Z$ be the $SO(n)$-bundle which as a
manifold is equal to ${\sf F}(V)$ but with twisted group action $u
\cdot_{\sigma_{p,q}}g=u \, \sigma_{p,q}(g)$ for $u\in {\sf F}(V)$ and $g\in SO(n)$. Then a Real
structure $\tau_{{\sf F}(V)}^{p,q} \colon {\sf F}(V)\to {\sf F}(V)$ commuting
with the involution $\tau \colon Z\to Z$ and satisfying
$\tau_{{\sf F}(V)}^{p,q}(u\, g)= \tau_{{\sf F}(V)}^{p,q} (u) \, \sigma_{p,q}(g)$ is the same
thing as an $SO(n)$-bundle morphism $\tau_{{\sf F}(V)}^{p,q} \colon {\sf F}(V)\to {\sf F}(V)^{\sigma_{p,q}}$ covering $\tau$, since ${\sf F}(V)= {\sf F}(V)^{\sigma_{p,q}}$ as manifolds, it makes sense to
demand that $\tau_{{\sf F}(V)}^{p,q}$ be an involution. Such an involution is easy to construct; with the involutive bundle morphism $\tau_{{\sf F}(V)}$ on ${\sf F}(V)$ induced fibrewise by $\tau_V$ that
satisfies $\tau_{{\sf F}(V)}(u\, g)= \tau_{{\sf F}(V)}(u)\, g$, the involution $e_{p,q} \colon \RR^n\to\RR^n$ induces a fibrewise involutive map which composed with $\tau_{{\sf F}(V)}$ yields the desired isomorphism $\tau_{{\sf F}(V)}^{p,q}$ when either $\tau_V$ is orientation-preserving and $q$ is even or $\tau_V$ is orientation-reversing and $q$ is odd. Then a
necessary condition is that ${\sf F}(V)$ and $\tau^{-1}\big({\sf F}(V)^{\sigma_{p,q}}\big)$ are isomorphic as $SO(n)$-bundles. Now if
$f \colon Z\to BSO(n)$ is a classifying map for ${\sf F}(V)$ then
$B(\sigma_{p,q})\circ f\circ \tau$ is a classifying map for $\tau^{-1}\big({\sf F}(V)^{\sigma_{p,q}}\big)$, where $B(\sigma_{p,q}) \colon BSO(n)\to BSO(n)$
is the involution induced by $\sigma_{p,q}$. It can be checked that
$\sigma_{p,q}=\sigma_{q,p}$ is an inner automorphism of $SO(n)$ if and
only if $q$ is even, in which case it can be deformed via
automorphisms to the identity map. Then $B(\sigma_{p,q})$ can be
deformed to the identity so that $B(\sigma_{p,q})\circ f\circ \tau$
and $f\circ \tau$ are homotopic maps, and hence $\tau^{-1}\big({\sf F}(V)^{\sigma_{p,q}}\big)\simeq \tau^{-1}\big({\sf F}(V)\big) \simeq {\sf F}(V)$ since $\tau_{{\sf F}(V)}$ commutes with $\tau$. We conclude that if $\tau_V \colon V\to V$ is an orientation-preserving involution and $q$ is even, then $V$ is Real $(p,q)$-oriented.
  
Henceforth we assume that the bundle $V\to Z$ is Real
$(p,q)$-oriented. Its $\ZZ_2$-invariant fibrewise inner product
defines a Real bundle of Clifford algebras 
$$
C\ell(V):={\sf F}(V)\times_{SO(\RR^{p,q})}C\ell(\RR^{p,q}) 
$$
on $Z$. The Lie group $Spin^c(n)\subseteq\CC\ell(n)$ is a central extension
$Spin^c(n):=Spin(n)\times_{\ZZ_2}U(1)$ of $SO(n)$, which is a Real
Lie group $Spin^c(\RR^{p,q})$ under the involutive automorphism which descends to the Real
structure on
$SO(\RR^{p,q})$ and restricts to complex conjugation on $U(1)$. In particular, there is a Real central extension
\begin{equation}
1 \ \longrightarrow \ U(1) \ \longrightarrow \ Spin^c(\RR^{p,q}) \
\longrightarrow \ SO(\RR^{p,q}) \ \longrightarrow \ 1 \ .
\label{eq:RealcentralSpin}
\end{equation}
Following again~\cite{Fok2015} we have
\begin{definition}
Let $V$ be an equivariant oriented real vector bundle of even Real
rank
$n=p+q$ over
a Real space $Z$ which is Real $(p,q)$-oriented. A \emph{Real $(p,q)$-spin$^c$ structure} or
\emph{$\KR$-orientation of type $(p,q)$} on $V$ is an
extension of the frame bundle ${\sf F}(V)$ to a Real
$Spin^c(\RR^{p,q})$-bundle $ \widehat{\sf F}(V)$ over $Z$ whose structure group
lifts that of ${\sf F}(V)$ as the Real central extension \eqref{eq:RealcentralSpin}. The bundle $V$ with a given Real spin$^c$ structure is
called a \emph{Real spin$^c$} or \emph{$\KR$-oriented} vector bundle.
\end{definition}

\begin{remark}
If $V\to Z$ has odd rank $n$, we apply the above considerations to $C\ell(V\oplus \RR_Z)$ instead, where $\RR_Z:=Z\times\RR$ is the trivial real line bundle with the trivial involution on its fibre.
\end{remark}

For a $\KR$-oriented bundle $V$ of type $(p,q)$, the extension $ \widehat{\sf F}(V)$
may be regarded as a Real $U(1)$-bundle over ${\sf F}(V)$ which fits in a
diagram of fibrations
$$
\xymatrix@C=5mm@R=3mm{
 & Spin^c(\RR^{p,q})~\ar[rr]\ar[dd] & & ~SO(\RR^{p,q}) \ar[dd] \\
U(1) ~ \ar[ru]\ar[rd] & & & \\
 & ~ \widehat{\sf F}(V)~\ar[rr]\ar[rd]& & ~ {\sf F}(V) \ar[ld] \\
 & & ~ Z ~ & 
}
$$
The topological obstruction to the existence
of a Real spin$^c$ structure on $V$ is the Real Dixmier--Douady class of the Real lifting bundle
gerbe associated to the Real central extension~\eqref{eq:RealcentralSpin}. It is easy to see that when the involution on $Z$ is trivial then a Real spin$^c$ structure is the same thing as a spin structure on $V$.

\begin{lemma}\label{Realsum}
If $V$ and $W$ are Real spin$^c$ vector bundles, then their Whitney sum
$V\oplus W$ carries a natural Real spin$^c$ structure.
\label{lem:Whitneyspin}\end{lemma}
\begin{proof}
Let $n=p+q$ and $m=r+s$ be the respective  ranks of $V$ and $W$. The maps
$e_i\mapsto e_i$, $i=1,\dots,p$ and $e_i\mapsto e_{i+r}$,
$i=p+1,\dots,n$, and $e_j\mapsto e_{j+p}$,
$j=1,\dots,r$ and
$e_j\mapsto e_{j+n}$,
$j=r+1,\dots,m$ give respective equivariant inclusions of $C\ell(\RR^{p,q})$ and $C\ell(\RR^{r,s})$ in $C\ell(\RR^{p+r,q+s})$. These inclusions induce a
diagram
$$
\xymatrix{
Spin^c(\RR^{p,q})\times Spin^c(\RR^{r,s}) \ \ar[d] \ar[r] & \ Spin^c(\RR^{p+q,r+s}) \ar[d] \\
SO(\RR^{p,q})\times SO(\RR^{r,s}) \ \ar[r] & \ SO(\RR^{p+r,q+s})
}
$$
which gives the desired Real $(p+r,q+s)$-spin$^c$ structure on $V\oplus W$.
\end{proof}

Let $V\to Z$ be any Real spin$^c$ vector bundle with Real spin$^c$
structure $ \widehat{\sf F}(V)\to {\sf F}(V)$ of type $(p,q)$. Any fixed equivariant
orientation-reversing isometry $\eta$ of $\RR^{p,q}$ induces an
equivariant automorphism of $C\ell(\RR^{p,q})$,
and hence of $Spin^c(\RR^{p,q})$, which is
also denoted $\eta$. Define a Real $U(1)$-bundle
$ \widehat{\sf F}^\eta(V)\to {\sf F}(V)$ with the same Real total space as
$ \widehat{\sf F}(V)$, but with the action of the Real group $Spin^c(\RR^{p,q})$
twisted by $\eta$; it defines the \emph{opposite} Real spin$^c$ vector
bundle $-V$.

If $Z$ is a Real manifold, a Real orientation of its tangent bundle $TZ$ can be specified by choosing a complete Riemannian metric on $Z$ and taking $\tau \colon Z\to Z$ to be an isometric involution. A \emph{Real spin$^c$ structure} on $Z$ is a Real
spin$^c$ structure on $TZ$. A Real manifold
together with a given Real spin$^c$ structure is called a \emph{Real
  spin$^c$ manifold}. 

\begin{lemma}
If $Z$ is a Real spin$^c$ manifold, then its boundary $\partial Z$
carries a natural Real spin$^c$ structure.
\end{lemma}
\begin{proof}
The frame bundle $ {\sf F}(T\partial Z)$ can be mapped to
$\partial\, {\sf F}(TZ)$. If $Z$ has Real dimension $n=p+q$, then the Real
$(p,q)$-spin$^c$ structure on $Z$ can be pulled back to a Real $(p-1,q)$-spin$^c$
structure on $\partial Z$ via the pullback diagram
$$
\xymatrix{
Spin^c(\RR^{p-1,q}) \ \ar[r] \ar[d] & \ Spin^c(\RR^{p,q}) \ar[d] \\
SO(\RR^{p-1,q}) \ \ar[r] & \ SO(\RR^{p,q})
}
$$
induced by the equivariant inclusion
$C\ell(\RR^{p-1,q})\hookrightarrow C\ell(\RR^{p,q})$
which sends $e_i\mapsto e_{i+1}$ for $i=1,\dots,n-1$.
\end{proof}

\subsection{Bundle gerbe $\KR$-homology}

Let $M$ be a Real space and let $(P,Y)$ be a Real bundle gerbe over
$M$ with Dixmier--Douady class $[H]=\DD_R(P)\in   H^2(M;\ZZ_2,\overline{\cU(1)})$.

\begin{definition}
\label{def:bg-cycle}
A \emph{bundle gerbe $\KR$-cycle} is a triple $(Z,E,f)$ where
$Z$ is a compact Real spin$^c$ manifold without boundary, $f \colon  Z\to M$
is a continuous equivariant map, and $E$ is a Real bundle gerbe module for $f^{-1}(P^*,Y)$.
\end{definition}

Notice that the definition of a Real bundle gerbe $D$-brane  (Definition \ref{def:RealDbrane}) is a special case of this 
definition.

We note that since $E$ is of finite rank, the pullback to $Z$  of the Real Dixmier--Douady class $\DD_R(f^{-1}(P^*,Y))=-f^*(\DD_R(P,Y))$ must be torsion. Moreover, the manifold $Z$ need not be connected, hence the
disjoint union
$$
(Z_1,E_1,f_1)\amalg (Z_2,E_2,f_2):= (Z_1\amalg Z_2,E_1\amalg
E_2,f_1\amalg f_2)
$$
is a well-defined operation on the set of all bundle gerbe
$\KR$-cycles. We say that two bundle gerbe $\KR$-cycles $(Z_1,E_1,f_1)$ and $(Z_2,E_2,f_2)$  are
\emph{isomorphic} if there exists an equivariant diffeomorphism $h \colon Z_1\to Z_2$
preserving the Real spin$^c$ structures such that $f_1=f_2\circ h$ and $h^{-1}(E_2)\simeq E_1$
as Real bundle gerbe modules for $f_1^{-1}(P^*,Y)$.
We denote the set of isomorphism classes of bundle gerbe $\KR$-cycles 
by $\RCyc(P,Y)$; it
is a commutative semi-group with addition $+$
induced by disjoint union of bundle gerbe $\KR$-cycles. Henceforth when we
refer to a bundle gerbe $\KR$-cycle we shall  mean an isomorphism class
of bundle gerbe $\KR$-cycles.

\begin{definition}\label{def:bordism}
Two bundle gerbe $\KR$-cycles $(Z_1,E_1,f_1)$ and $(Z_2,E_2,f_2)$  are 
\emph{{Real} spin$^c$ bordant} if there exists a compact Real
spin$^c$ manifold $\underline{Z}$ with a $\ZZ_2$-invariant boundary, a continuous equivariant map
$\underline{f}  \colon \underline{Z} \to M$ and a Real bundle gerbe module
$\underline{E}$ for $\underline{f}^{-1}(P^*,Y)$ such that the two bundle gerbe
$KR$-cycles $\partial(\,\underline{Z}\,,\,\underline{E}\,,\,
\underline{f}\, ):= (\partial
\underline{Z}\,,\,\underline{E}|_{\partial
  \underline{Z}}\,,\,\underline{f}|_{\partial \underline{Z}}\, )$ and
$(Z_1,E_1,f_1)\amalg (-Z_2,E_2,f_2)$ are isomorphic, where $-Z_2$
denotes the Real manifold $Z_2$ with the opposite Real spin$^c$
structure on its tangent bundle $TZ_2$. The triple
$(\,\underline{Z}\,,\,\underline{E}\,,\, \underline{f}\, )$ is called
a \emph{Real spin$^c$ bordism} of bundle gerbe $\KR$-cycles.
\end{definition}

The most intricate equivalence relation on the semi-group
$\RCyc(P,Y)$ is a twisted Real version of vector bundle
modification. For this, let $S^{p,q}$ be the unit sphere of dimension $p+q-1$ in
$\RR^{p,q}$ with respect to the standard flat Euclidean metric on
$\RR^p\times \RR^q$; then $S^{n,0}=S^{n-1}$ is the standard
$n-1$-sphere with the trivial Real involution. The frame bundle of
$T\RR^{p,q}$ can be identified with $\RR^{p,q}\times SO(\RR^{p,q})$, and
we can equip $\RR^{p,q}$ with the trivial Real spin$^c$ structure
$\RR^{p,q}\times Spin^c(\RR^{p,q})$. Then the associated Real spin$^c$ structure
on $S^{p,q}$ is the Real $Spin^c(\RR^{p-1,q})$-bundle $\widehat{\sf F}(TS^{p,q})$ with fibre at
$x\in S^{p,q}$ given by the space of all elements of $Spin^c(\RR^{p,q})$ whose
image in $SO(\RR^{p,q})$ is a matrix with first column equal to $x$. 

Let $V_{p,q}$ be a Real spin$^c$ vector bundle of type $(p,q)$ with even-dimensional fibres over a
compact Real
spin$^c$ manifold $Z$. Then the Whitney sum $V_{p,q}\oplus\RR_Z$ is a
Real spin$^c$
vector bundle over $Z$ of type $(p+1,q)$, with the trivial involution on
the trivial real line bundle $\RR_Z$ and bundle projection $\lambda$. Fixing a representative
within the $\ZZ_2$-homotopy class of Real
$Spin^c(\RR^{p+1,q})$-bundles $ \widehat{\sf F}(V_{p,q}\oplus\RR_Z)$ over $Z$, we define a
$\ZZ_2$-invariant metric on $V_{p,q}\oplus \RR_Z$. Let ${Z}_{{p,q}}$ be the unit sphere bundle of
$V_{p,q}\oplus\RR_Z$; it is Real spin$^c$ bordant to any other sphere bundle
defined by choosing a different representative of the $\ZZ_2$-homotopy
class. The Real manifold ${Z}_{{p,q}}$ may be described explicitly as the
fibre bundle
$$
{Z}_{{p,q}}= \widehat{\sf F}(V_{p,q}\oplus\RR_Z)\times_{Spin^c(\RR^{p+1,q})}S^{p+1,q}
$$
over $Z$ with a Real structure commuting with $\tau$ and projection $\rho_{p,q}$; here $Spin^c(\RR^{p+1,q})$ acts on the
Real sphere $S^{p+1,q}$ by projection to its isometry group $SO(\RR^{p+1,q})$. The tangent bundle of
$V_{p,q}\oplus\RR_Z$ sits in a split exact sequence
$$
0 \ \longrightarrow \ \lambda^{-1}(V_{p,q}\oplus\RR_Z) \ \longrightarrow \
T(V_{p,q}\oplus\RR_Z) \ \longrightarrow \ \lambda^{-1}(TZ) \ \longrightarrow \ 0
$$
and upon choosing a splitting we can identify the tangent bundle
$$
TZ_{p,q} \simeq \rho_{p,q}^{-1}(TZ) \ \oplus \ \big(\, \widehat{\sf
  F}(V_{p,q}\oplus\RR_Z) \times_{Spin^c(\RR^{p+1,q})} TS^{p+1,q} \, \big) \ .
$$
It follows that the Real spin$^c$ structures on $TZ$ and $V_{p,q}$ naturally
induce a Real spin$^c$ structure on $TZ_{p,q}$, so $Z_{p,q}$ is a compact Real
spin$^c$ manifold.  There are two special instances
of this construction that we are interested in, which will
respectively implement the periodicities of complex and real $K$-theory. 

Firstly, consider the case $(p,q)=(k,k)$ for $k\geq 1$. As a Real space
$\RR^{k,k}\simeq\CC^k$ with the involution given by complex
conjugation, and
$C\ell(\RR^{k,k})\simeq\CC\ell(2k)=\CC\ell^+(2k)\oplus\CC\ell^-(2k)$
is the complex Clifford algebra with its natural $\ZZ_2$-grading. The group
$Spin^c(2k)$ has two irreducible half-spin representations
$\Delta_{k,k}^{\pm}$ of equal
dimension $2^{k-1}$, and the
associated bundles of half-spinors
$S_{k,k}^\pm:=\widehat{\sf F}(TS^{k+1,k})\times_{Spin^c(\RR^{k,k})}
\Delta_{k,k}^\pm$ on $S^{k+1,k}$ are Real vector
bundles. By
the Atiyah--Bott--Shapiro construction, the dual of the positive
spinor bundle $(S_{k,k}^+)^*$ together with the trivial line bundle
generate $\KR(S^{k+1,k})$~\cite{Ati,Lawsonbook}; for $k=1$ this is
essentially the Bott generator constructed from the Hopf 
bundle $H\to S^2=\CC P^1$ with its natural Real structure induced
by complex conjugation, see Example~\ref{ex:S2}~(c).

Secondly, let $(p,q)=(8k,0)$ for $k\geq 1$. Then $\RR^{8k,0}\simeq\RR^{8k}$ is endowed
with the trivial involution and $C\ell(\RR^{8k,0})\simeq C\ell(8k)$ is
a real Clifford algebra. The group
$Spin(8k)$ has two irreducible real half-spin representations
$\Delta_{8k,0}^{\pm}$ of equal
dimension $2^{4k-1}$, and the
associated bundles of half-spinors $S_{8k,0}^\pm:=\widehat{\sf F}(TS^{8k})\times_{Spin(8k)} \Delta_{8k,0}^\pm$ on $S^{8k}$ are real
vector bundles. Again by
the Atiyah--Bott--Shapiro construction, the dual of the positive
spinor bundle $(S_{8k,0}^+)^*$ together with the trivial line bundle
generate $\KR(S^{8k})\simeq \KO(S^{8k})$. 

In both of these instances, the bundle
$$
{\sf S}_{p,q}:= \widehat{\sf
  F}(V_{p,q}\oplus\RR_Z)\times_{Spin^c(\RR^{p+1,q})} \big( S_{p,q}^+\big)^*
$$
is a Real vector bundle over $Z_{p,q}$.

\begin{definition} \label{Realvb}
Let $(P,Y)$ be a Real bundle gerbe. Let
$(Z,E,f)$ be a bundle gerbe $\KR$-cycle and let $V_{p,q}\to Z$ be a Real
spin$^c$ vector bundle of type $(p,q)$. Let
$\tilde\pi_{p,q} \colon (f\circ\rho_{p,q})^{-1}(Y)\to Z_{p,q}$ be the pullback of the
surjective submersion $Y\to M$ to $Z_{p,q}$, and let $\tilde\rho_{p,q} \colon (f\circ
\rho_{p,q})^{-1}(Y)\to f^{-1}(Y)$ be the induced projection. Then the
\emph{Real vector bundle modification} of $(Z,E,f)$ by $V_{p,q}$ is the bundle gerbe $\KR$-cycle
$$
(Z,E,f)_{p,q}:= \big(Z_{p,q}\,,\,\tilde\rho_{p,q}^{-1}(E)\otimes
\tilde\pi_{p,q}^{-1}({\sf S}_{p,q}) \,,\,f\circ\rho_{p,q}\big)
$$
for $(p,q)=(k,k)$ and $(p,q)=(8k,0)$ with $k\geq 1$, which we respectively call the
\emph{complex} and \emph{real} modifications.
\end{definition}

The \emph{$\KR$-homology group} $\KR^{\bg}_\ast(M,P)$ of the Real
bundle gerbe $(P,Y)$ is defined to be the abelian group obtained by quotienting
$\RCyc(P,Y)$ by the equivalence relation $\sim$ generated by the
disjoint union/direct sum relation, that is $(Z,E_1,f)\amalg
(Z,E_2,f)\sim (Z,E_1\oplus E_2,f)$, Real spin$^c$ bordism, and Real
vector bundle modification. The homology class of a
bundle gerbe $\KR$-cycle $(Z,E,f)\in\RCyc(P,Y)$ is denoted $[Z,E,f]\in
\KR^{\bg}_\ast(M,P)$. The group operation is induced by disjoint union of
bundle gerbe $\KR$-cycles.
The identity element of the group $\KR^{\bg}_\ast(M,P)$ is represented by
$[\emptyset,\emptyset,\emptyset]$, or more generally by any null bordant $\KR$-cycle
$\partial [W, E, f]$, see Definition~\ref{def:bordism}. 
Inverses are induced by taking opposite Real spin$^c$
structures, that is $-[Z,E,f]:=[-Z,E,f]$; this follows from the Real spin$^c$ bordism
$(Z,E,f)\amalg (-Z,E,f)=\partial(Z\times[0,1],\pi_Z^{-1}(E),f\circ\pi_Z)$
with the trivial involution on $[0,1]$ and $\pi_Z \colon Z\times[0,1]\to Z$
the projection.

By construction, the equivalence relation on $\RCyc(P,Y)$
preserves the type $(p,q)$ of the Real spin$^c$ structure on $Z$
$\textrm{mod}~(1,1)$ and the dimension of $Z$ $\textrm{mod}~8$ in bundle gerbe $\KR$-cycles
$(Z,E,f)$, so one can define the subgroups $\KR^{\bg}_{p,q}(M,P)$
 consisting of classes of bundle gerbe $\KR$-cycles
$(Z,E,f)$ for which all connected components of $Z$ carry Real
spin$^c$ structures of type~$(p,q)~\textrm{mod}~(1,1)$ and are
of dimension $ n=p+q~\textrm{mod}~8$.
Then the abelian group
$$
\KR^{\bg}_\ast(M,P) = \bigoplus _{n=0}^7 \,\KR^{\bg}_{n}(M,P) 
$$
has a natural $\ZZ_8$-grading, where $\KR^{\bg}_n(M,P):= \KR^{\bg}_{0,n}(M,P)$.

\begin{lemma}\label{Kclass}
The homology class of a bundle gerbe $\KR$-cycle $(Z,E,f)$ depends only
  on the class of $E$ in $\KR_{\bg}(Z,f^{-1}(P^*))$.
\label{lem:KRclassdep}\end{lemma}
\begin{proof}
Let $[E]$ denote the class of $E$ in $\KR_{\bg}(Z,f^{-1}(P^*))$ and suppose
that $[E]=[F]$ for another Real bundle gerbe module $F$. Then there exists
a Real bundle gerbe module $G$ such that $E\oplus G\simeq F\oplus
G$. Passing to equivalence classes in
$\KR^{\bg}_\ast(M,P)$ using the disjoint union/direct sum relation gives
$[Z,E,f]+[Z,G,f]=[Z,F,f]+[Z,G,f]$, and so
$[Z,E,f]=[Z,F,f]$ in $\KR^{\bg}_\ast(M,P)$.
\end{proof}

\begin{remark}\label{rem:isoKR}
Lemma~\ref{lem:KRclassdep} implies that any Real stable isomorphism
$(P,Y)\to (Q,X)$ induces a canonical isomorphism $\KR^{\bg}_\ast(M,P)\simeq \KR^{\bg}_\ast(M,Q)$, and in particular
the isomorphism class of the abelian group
$\KR^{\bg}_\ast(M,P)$ depends only on the Real Dixmier--Douady class of $P$. As in Remark \ref{rem:K notation}, when the bundle gerbe $P$ with class $[H]$ is understood, we write $\KR_\ast^{\bg}(M, [H])$. 
Since $\KR_{\bg}(Z,f^{-1}(P^*))\simeq \KR(Z)$ for any trivialisable Real bundle gerbe $(P,Y)$, our formalism includes also a definition of geometric $\KR$-homology in the untwisted case in terms of Real vector bundles. Moreover, we may use it to define an isomorphic version of bundle gerbe $\KR$-homology wherein
the Real bundle gerbe module $E$ is replaced by a class
$\xi\in \KR_{\bg}(Z,f^{-1}(P^*))$. Representing $\xi=[E]-[F]$ by two Real
bundle gerbe modules, we get a well-defined element
$[Z,\xi,f]\in \KR^{\bg}_\ast(M,P)$ by setting
$$
[Z,\xi,f]:=[Z,E,f]-[Z,F,f] \ .
$$
Conversely, there is a map $[Z,E,f]\mapsto[Z,[E],f]$.
\end{remark}

The functor $\KR^{\bg}_\ast$ is defined to be the $\ZZ_8$-graded covariant functor from
the category of pairs of Real manifolds with Real bundle gerbes to the category of abelian groups defined
on objects by $(M,P) \mapsto \KR^{\bg}_\ast(M,P)$ and on equivariant continuous maps
$\phi \colon (M,\phi^{-1}(Q))\to (N,Q)$ by the induced homomorphism of $\ZZ_8$-graded abelian groups
$$
\KR^{\bg}_\ast(\phi):=\phi_*\,\colon\, \KR^{\bg}_\ast(M,\phi^{-1}(Q)) \ \longrightarrow \
\KR^{\bg}_\ast(N, Q)
$$
with
$$
 \phi_*[Z,E,f]:= [Z,E,\phi\circ f] \ .
$$
Note that this transformation is well-defined and functorial; one has
$({\rm id}_{(M,P)})_*={\rm id}_{\KR^{\bg}_\ast(M,P)}$ and
$(\phi\circ\psi)_*=\phi_*\circ\psi_*$, and since Real bundle gerbe
modules over $Z$ extend to Real bundle gerbe modules over
$Z\times[0,1]$, it follows by Real spin$^c$ bordism that
induced homomorphisms depend only on their $\ZZ_2$-homotopy classes. By restricting our definitions to the category of manifolds with real bundle gerbes, our formalism also includes a definition of bundle gerbe $\KO$-homology.

Using the $\KR(Z)$-module structure of the bundle gerbe $\KR$-theory groups
$\KR_{\bg}(Z,f^{-1}(P^*))$, we can endow the bundle gerbe $\KR$-homology group $\KR^{\bg}_\ast(M,P)$ with the structure of a module over the 
$\KR$-theory ring
$\KR(M)$. We define the
$\ZZ_8$-graded left action
$$
\KR(M)\otimes \KR^{\bg}_\ast(M,P) \ \longrightarrow
\ \KR^{\bg}_\ast(M,P)
$$
which is given for any Real vector bundle $F\to M$ and any bundle gerbe $\KR$-cycle
class $[Z,E,f]\in \KR^{\bg}_{p,q}(M,P)$ by
$$
[F] \cdot [Z,E,f]:= [Z,(f\circ \tilde \pi)^{-1}(F)\otimes E,f]
$$
and extended linearly, where $\tilde \pi \colon f^{-1}(Y) \to Z$ is the pullback of the surjective submersion. 

If $(M_1,(P_1,Y_1))$ and $(M_2,(P_2,Y_2))$ are Real spaces endowed with Real
bundle gerbes, then the \emph{exterior product}
\begin{eqnarray*}
\KR^{\bg}_{p_1,q_1}(M_1,P_1)\otimes
\KR^{\bg}_{p_2,q_2}(M_2,P_2) \longrightarrow
\KR^{\bg}_{p_1+p_2,q_1+q_2}(M_1\times M_2,P_1\otimes P_2)
\end{eqnarray*}
is defined on $[Z_1,E_1,f_1]\in
\KR^{\bg}_{p_1,q_1}(M_1,P_1)$ and $[Z_2,E_2,f_2]\in
\KR^{\bg}_{p_2,q_2}(M_2,P_2)$ by
$$
[Z_1,E_1,f_1]\otimes [Z_2,E_2,f_2]:= [Z_1\times Z_2,E_1\otimes E_2, (f_1,f_2)] \ ,
$$
where $Z_1\times Z_2$ has the product Real $(p_1+p_2,q_1+q_2)$-spin$^c$
structure uniquely induced by the Real $(p_1,q_1)$ and $(p_2,q_2)$
spin$^c$ structures on $Z_1$ and $Z_2$, respectively
(cf. Lemma~\ref{lem:Whitneyspin}), and here $E_1\otimes E_2$
is the Real $f_1^{-1}(P_1^*)\otimes f_2^{-1}(P_2^*)$-bundle gerbe module with fibres $(E_1\otimes E_2)_{(y_1,y_2)}=(E_1)_{y_1}\otimes (E_2)_{y_2}$ for
$(y_1,y_2)\in f_1^{-1}(Y_1)\times f_2^{-1}(Y_2)$. This product is natural with
respect to continuous equivariant maps.

\subsection{Twisted $\KR$-homology}

We shall now review the constructions of twisted 
$\KR$-homology groups, which were defined in~\cite{Fok2015,Mou} using a Real version of Kasparov's
 $\KK$-theory.

\begin{definition} Let $A$ be a separable $\mathbb{Z}_2$-graded
  $C^*$-algebra.  A \emph{Real structure} on $A$ is an anti-linear,
  degree~$0$ involutive $*$-automorphism $\sigma$; the pair
  $(A,\sigma)$ is called a \emph{Real $\ZZ_2$-graded
    $C^*$-algebra}. An \emph{equivariant graded homomorphism} $A \to
  B$  is a grading preserving $*$-homomorphism that intertwines the Real structures.
\end{definition}		 

 If $A$ is a Real ungraded $C^*$-algebra, then we assign the trivial $\ZZ_2$-grading with $A$ as its even part and $0$ as its odd part.
  
\begin{example} Let $\cH$ be a separable $\ZZ_2$-graded Hilbert space
  equipped with an anti-linear, degree~$0$ involution $\tau_\cH$. The $\mathbb{Z}_2$-graded $C^*$-algebra $B(\cH)$  of bounded linear operators on $\cH$ inherits a Real structure $\sigma$  defined by
$$\sigma(T) = \tau_\cH\circ T \circ \tau_\cH \ ,$$ 
for all $T\in B(\cH)$. This further induces a Real structure on the
two-sided $*$-ideal of compact operators $\cK(\cH)$. Let
$B(\cH)^\sigma$ denote the fixed point set of the involution $\sigma$,
that is the set of operators which commute with $\tau_\cH$.
\end{example}	

\begin{example} Let $(M,\tau)$ be a Real manifold. Then the separable $C^*$-algebra $\cC(M)$ of continuous complex-valued functions $f\colon M\to \CC$ vanishing at infinity has an induced Real structure given by $\sigma(f)(m) = \overline{f(\tau(m))}$.
\end{example}		

\begin{definition}\label{FredMod}
	Let $A$ be a Real separable $\mathbb{Z}_2$-graded $C^*$-algebra. A \emph{$(p,q)$-graded Real Fredholm module} over $A$ is a triple $(\rho, \cH, F)$ where
	\begin{enumerate}
		\item $\cH$ is a Real $\ZZ_2$-graded separable Hilbert
                  space which is a
$*$-module  over the Real Clifford algebra $C\ell(\RR^{p,q})$ whose
generators are skew-adjoint operators of odd degree in $B(\cH)^\sigma$;
		\item $\rho\colon A\to B(\cH)$ is a Real graded
                  representation that commutes with the
                  $C\ell(\RR^{p,q})$-action; and
		\item $F\in B(\cH)^\sigma$ is a bounded operator of odd degree which commutes with the  $C\ell(\RR^{p,q})$-action and satisfies
		\[(F^2-1)\rho(a) \ , \ (F-F^*)\rho(a) \ ,\  [F,
                \rho(a)] \ \in \ \cK(\cH)\]
		for all $a\in A$.
	\end{enumerate}
\end{definition}

Let ${\rm RFMod}_{p,q}(A)$ denote the set of all $(p,q)$-graded Real
Fredholm modules over $A$. The direct sum of two Real Fredholm
modules  $(\rho, \cH, F)$ and  $(\rho', \cH', F'\, )$ is the Real
Fredholm module $(\rho\oplus \rho', \cH\oplus \cH', F\oplus F'\, )$
and $(0,0,0)$ is the zero module. We introduce an equivalence relation
$\sim$ on the semi-group $({\rm RFMod}_{p,q}(A), \oplus)$ generated by the relations:
\begin{enumerate}
		\item[$(i)$] \emph{Real unitary equivalence:} $(\rho,
                  \mathcal{H}, F)\sim (\rho', \mathcal{H}', F'\, )$
                  if and only if there is a degree preserving unitary
                  isomorphism $U\colon \mathcal{H}'\to \mathcal{H}$
                  that intertwines with the $C\ell(\RR^{p,q})$ generators and the Real structures, and satisfies
		\[(\rho', \mathcal{H}', F'\, )=(U^*\, \rho\, U,
                \mathcal{H}', U^*\, F\, U) \ .\]
		\item[$(ii)$] \emph{Real homotopy equivalence:}
                  $(\rho, \mathcal{H}, F)\sim (\rho', \mathcal{H}',
                  F'\, )$ if and only if there exists a norm
                  continuous function $t\mapsto F_t$ such that
                  $(\rho_t, \mathcal{H}_t, F_t)$ is a Real Fredholm
                  module for all $t\in[0, 1]$ with  $F_0=F$ and $F_1=F'$.
	\end{enumerate}

The \emph{$\KR$-homology group} of a Real separable
$\mathbb{Z}_2$-graded $C^*$-algebra $A$ is the free abelian group
$\KR^{p,q}(A)$ generated by ${\rm RFMod}_{p,q}(A)/\sim$ modulo
the relation $[x_0\oplus x_1]=[x_0]+[x_1]$ where $[x_0], [x_1] \in
{\rm RFMod}_{p,q}(A)/\sim$.
Equivalently, we could have defined $\KR^{p,q}(A) := \KR(A \hat \otimes C\ell(\RR^{p,q}))$ where the $(1, 1)$-periodicity is more discernible. 
The inverse of a class in $\KR^{p,q}(A)$ represented by the module
$(\rho, \mathcal{H}, F)$ is given by  $(\rho, \mathcal{H}^{\op}, -F)$,
where $\cH^{\op}$ is the Hilbert space $\cH$ with the opposite
$\ZZ_2$-grading, opposite Real structure and where the Clifford
algebra generators reverse their signs. The zero element in
$\KR^{p,q}(A)$ is represented by \emph{degenerate} Real Fredholm
modules, that is those for which the three operators listed in item (3) of Definition \ref{FredMod} are identically zero in $\cK(\cH)$.
For a Real manifold $M$ we define its {$\KR$-homology groups} by
$$\KR_{p,q}(M) := \KR^{p,q}(\cC(M)) \ . $$
As usual this is $(1,1)$-periodic in $(p,q)$, so that
$\KR_{p,q}(M)\simeq \KR_{q-p}(M)$, and $8$-periodic in $(0,q)$.

Recall that a Real $PU(\cH)$-bundle over $M$ is a principal
$PU(\cH)$-bundle $\mathcal{P}$ with a Real structure
$\tau_{\mathcal{P}}$ that commutes with the involution $\tau$ on $M$ and is
compatible with the right $PU(\cH)$-action, that is $\tau_{\mathcal{P}}(p\, g) = \tau_{\mathcal{P}}(p)\, \sigma(g)$, where $\sigma$ is the anti-linear involution on $PU(\cH)$ induced by complex conjugation on $\cH$. From Proposition \ref{prop:PU(H)}, we know that Real $PU(\cH)$-bundles are classified up to isomorphism by their Real Dixmier--Douady class $\DD_R(\mathcal{P})\in H^2(M; \ZZ_2; \overline{\cU(1)})$.  The Real projective unitary group $PU(\cH)$ acts by automorphisms on the Real elementary $C^*$-algebra $\cK(\cH)$ and the associated bundle 
$$
\cA = \mathcal P\times_{PU(\cH)} \cK(\cH)
$$
is called a \emph{Real Dixmier--Douady bundle}. It is a locally trivial $\cK(\cH)$-bundle with an induced involution that maps fibre to fibre anti-linearly. 
The opposite Real Dixmier--Douady bundle $\cA^{\op}$ is obtained by replacing each fiber $\cA_m$ by the opposite Real $C^*$-algebra $\cA^{\op}_m$, so in particular $\DD_R(\cA^{\op}) = - \DD_R(\cA)$. 

A \emph{Real spinor bundle} for $\cA$ is a Real bundle of Hilbert spaces $\mathcal S$ on $M$ such that $\mathcal{A}$ is isomorphic to $\mathcal{K}(\mathcal{S})$.  Two Real Dixmier--Douady bundles $\cA_1$ and $\cA_2$ are \emph{Morita isomorphic} if $\mathcal{A}_1\hat \otimes \mathcal{A}_2^{\op}$ admits a Real spinor bundle. 
Morita isomorphism is the appropriate notion of stable isomorphism for Real Dixmier--Douady bundles and we have
\begin{proposition}[\cite{Fok2015}]
	Real Dixmier--Douady bundles over $M$ are classified up to Morita isomorphisms by their  Real Dixmier--Douady class $\DD_R(\cA) \in H^2(M; \ZZ_2, \overline{\cU(1)})$.
\end{proposition}

	Let $M$ be a Real manifold with a Real Dixmier--Douady bundle
        $\cA$ and let $\Gamma_M(\cA)$ denote the Real separable
        $C^*$-algebra of sections of $\cA$ vanishing at infinity. The
        \emph{twisted $\KR$-homology group} of the pair $(M,\cA)$ is defined by 
	\[\KR_{p,q}(M, \mathcal{A}) :=\KR^{p,q}(\Gamma_M(\cA)) \ .\]

A \emph{Morita morphism} between two Real Dixmier--Douady bundles $(M_1,\cA_1), (M_2,\cA_2)$ locally modelled on $\cK(\cH_1), \cK(\cH_2)$ consists of a pair
$$(f, \mathcal{E})\colon (M_1, \cA_1) \ \longrightarrow \ (M_2, \cA_2)$$
where $f\colon M_1 \to M_2$ is an equivariant proper smooth map and
$\mathcal{E}$ is a Real $(f^{-1}(\cA_2), \cA_1)$-bimodule, that is a
Real bundle of Hilbert spaces on $M_1$ which is a Hilbert
$f^{-1}(\cA_2)^{\op}\hat\otimes \cA_1$-module locally modelled on the
$(\cK(\cH_1),\cK(\cH_2))$-bimodule $\cK(\cH_1, \cH_2)$. A Morita
morphism exists if and only if $\DD_R(\cA_1)  = f^* \DD_R(\cA_2)$. Any
two Morita morphisms are related by a Real line bundle via $(f,
\mathcal{E}) \mapsto (f, \mathcal{E}\otimes L)$ where $L$ is
classified by its Real Chern class in $H^1(M_1;\ZZ_2,\overline{\cU(1)})$. A
trivialisation of $L$ is called a \emph{$2$-isomorphism} between the
Morita morphisms. Twisted $\KR$-homology is then a covariant 2-functor relative to Morita morphisms $(f, \mathcal{E})\colon (M_1, \cA_1)\to (M_2, \cA_2)$,
$$f_*\colon  \KR_*(M_1, \cA_1) \ \longrightarrow \ \KR_*(M_2, \cA_2) \ , $$
where the induced pushforward map $f_*$ depends only on the 2-isomorphism class of $(f, \mathcal{E})$, and the Real Picard group $H^1(M_1;\ZZ_2,\overline{\cU(1)})$ acts on $\KR$-homology by Morita automorphisms. 

The notion of Real Fredholm modules generalises straightforwardly to
Real Kasparov $(\cA,\cB)$-modules, by substituting $\cH$ in Definition
\ref{FredMod} with Real Hilbert $(A,B)$-bimodules, leading to
bivariant $\KKR$-theory; see \cite[Chapter~9]{Mou} for more details on
the construction of the $\KKR$-bifunctor via correspondences and the
Real Kasparov product. Twisted $\KR$-theory groups of a pair $(M,\cA)$ can thus be defined as 
	$$
	\KR^{p,q}(M, \mathcal{A}):=\KR_{p,q}(\Gamma_M(\cA)) = \KKR_{p,q}(\CC, \Gamma_M(\cA))
	$$
	where the Real structure on the $C^*$-algebra $\CC$ is given by complex conjugation. 
	We have

	\begin{proposition}
	\label{prop:K-isom2}
Let $\cP \to M$ be a Real $PU(\cH)$-bundle with torsion Real Dixmier-Douady class, $L_\cP$ the associated lifting bundle gerbe  and $\cA$ the associated Real Dixmier-Douady bundle. Then there is a natural isomorphism
		\begin{equation*}  \KR_{\bg}^{p,q}(M, L_\cP) \simeq
          \KR^{p,q}(M, \mathcal{A}) \ ,	\end{equation*}
sending Real bundle gerbe modules to $\Gamma_M(\cA)$-modules.
	\end{proposition}
	\begin{proof}
This follows from Proposition~\ref{prop:PU(H)}, Theorem~\ref{K-isom} and the subsequent discussion on extension to  bigraded groups.
\end{proof}

Let $V$ be a Real $(p,q)$-oriented vector bundle on $M$ and recall
that a $\KR$-orientation of type $(p,q)$ on $V$ corresponds to a lift of the frame bundle ${\sf F}(V)$ to a Real
$Spin^c(\RR^{p,q})$-bundle $ \widehat{\sf F}(V)$. The obstruction to
$\KR$-orientability can be equivalently characterised by the Clifford
bundle $C\ell(V)$: this is a Real Dixmier--Douady bundle and $V$ is
$\KR$-oriented if and only if $C\ell(V)$ admits a Real spinor
bundle, that is it is Morita trivial. In analogy with the complex
case, if the tangent bundle $TM$ is Real $(p, q)$-oriented, then
$(\cC(M), \Gamma_M(C\ell(TM)))$ is a Poincar\'e duality pair; that
is there exists a $\KR$-homology fundamental class $[M]\in \KR_{p,q}(M, C\ell(TM))$ which implements the Poincar\'e duality isomorphism 
$$\mathtt{PD}_M\colon \KR^{r, s}(M, \cA) \ \longrightarrow \ \KR_{p-r,
  q-s}(M, \cA^{\op}\hat \otimes C\ell(TM)) \ , \ \ [E] \longmapsto [E]\cap [M]
\ .$$
Poincar\'e duality in twisted $\KR$-theory can be proven along the
same lines as in \cite{Tu, EEK}, but using instead the framework of
$\KKR$-theory and Real Dixmier--Douady bundles. In particular, the cap
product $\cap$ corresponds to Kasparov product with $[M]$. 
\begin{remark}
In the case that
$M$ is a Real spin$^c$ manifold of type $(p,q)$, its fundamental class
$[M]\in\KR_{p,q}(M)$ can be represented by the (unbounded) $(p,q)$-graded
Real Fredholm module $(\rho,\cH,T)$, where $\cH$ is the Hilbert space
of $L^2$-sections of the Real spinor bundle $S_{p,q}=\widehat{\sf
  F}(TM)\times_{Spin^c(\RR^{p,q})} C\ell(\RR^{p,q})$,
$\rho$ is the natural module action of $\cC(M)$ on $\cH$ by
multiplication, and $T$ is the corresponding Dirac operator; in the
untwisted case Poincar\'e duality maps the class of a Real vector
bundle $E\to M$ to the Fredholm module obtained as above with the
spinor bundle replaced by $S_{p,q}\otimes E$ and $T$ the corresponding
twisted Dirac operator.
\end{remark}

Let $(M_1,\cA_1), (M_2,\cA_2)$ be pairs of Real manifolds with Real Dixmier--Douady bundles, and assume that $TM_1$ is Real $(p_1, q_1)$-oriented and $TM_2$ is Real $(p_2, q_2)$-oriented. 

\begin{definition}\label{gysinmap}
For any Morita morphism $(f, \mathcal{E})\colon (M_1, \cA_1)\to (M_2, \cA_2)$, the \emph{Gysin homomorphism} in twisted $\KR$-theory is the unique group homomorphism 
$$f_!\colon  \KR^{r, s}(M_1, \cA_1\hat\otimes C\ell(TM_1)) \
\longrightarrow \ \KR^{r+p_2-p_1,s+q_2-q_1}(M_2, \cA_2\hat\otimes  C\ell(TM_2) ) $$ defined by declaring the diagram 
\begin{equation}\label{gysin}
\xymatrix{
\KR^{r, s}(M_1, \cA_1\hat\otimes C\ell(TM_1)) \ \ar[d]_{\mathtt{PD}_{M_1}} \ar[r]^{\hspace{-0.9cm}f_!}  & \  \KR^{r+p_2-p_1,s+q_2-q_1}(M_2, \cA_2\hat\otimes   C\ell(TM_2) ) \ar[d]^{\mathtt{PD}_{M_2}} \\
\KR_{p_1-r, q_1-s}(M_1, \cA_1^{\op}) \ \ar[r]_{f_*} & \ \KR_{p_1-r, q_1-s}(M_2, \cA_2^{\op})
}
\end{equation}
to be commutative.
\end{definition}
\begin{remark} By construction the Gysin homomorphism is functorial, and in particular it depends only on the homotopy class of the map $f$.
\end{remark} 
We further have a corresponding Thom isomorphism in twisted $\KR$-theory.
\begin{proposition}[\cite{Mou}] Let $M$ be a Real manifold with Real
  Dixmier--Douady bundle $\cA$. If  $\pi\colon V\to M$ and $TM$ are Real $(p, q)$-oriented vector bundles, then there is an isomorphism of abelian groups 
$$\KR^{r+p, s+q}(M, \cA \hat\otimes C\ell(V))\simeq \KR^{r, s}(V,
\pi^{-1}(\cA)) \ .$$
\label{prop:Thomiso}\end{proposition}

\subsection{The Real twisted assembly map}

We will finally define a natural isomorphism from bundle gerbe 
$\KR$-homology to twisted $\KR$-homology.   Throughout this section   $M$ is a Real manifold and $\cP \to M$  a Real $PU(\cH)$  
bundle with Real lifting bundle gerbe $L_\cP$ and associated Real Dixmier-Douady bundle $\cA$.  

 If $(Z,E,f)$ represents a
bundle gerbe $\KR$-cycle in $\KR^{\bg}_{p,q}(M,L_\cP)$,  then by  Proposition \ref{prop:K-isom2} the bundle gerbe module $E$ defines a class $[E]$ in $\KR(Z, f^{-1}(\cA^{\op}))$. Since all connected components of $Z$ carry Real
spin$^c$ structures of type~$(p,q)~\textrm{mod}~(1,1)$,  Poincar\'e
duality gives a homology class $\mathtt{PD}_Z[E] \in \KR_{p,q}(Z,
f^{-1}(\cA) ).$ 

To proceed we first need an alternative description of
Real vector bundle modification in terms of the Gysin homomorphism. Let $\rho_{p,q}\colon Z_{p,q} \to Z$ denote the
unit sphere bundle of $V_{p,q}\oplus\RR_Z$ as in
Definition~\ref{Realvb}. It admits a canonical north pole section
$s\colon Z \to Z_{p,q}$ defined by $x\mapsto (s_0(x),1)$, where $s_0$
is the zero section of $V_{p,q}$. By Definition \ref{gysinmap} and the
isomorphism in Proposition \ref{prop:K-isom2}, we obtain homomorphisms $s_!\colon
\KR_{\bg}^{r,s}(Z,f^{-1}(L^*_\cP)) \to \KR_{\bg}^{r+p,s+q}(Z_{p,q},(f\circ
\rho_{p,q})^{-1}(L^*_\cP) ) $.   

\begin{lemma} \label{Realvb2} 
Let $(Z,E,f)$ be a bundle gerbe $\KR$-cycle. Then its Real vector bundle modification $(Z,E,f)_{p,q}$ is Real spin$^c$ bordant to $ [Z_{p,q}, s_! [E], \,f\circ\rho_{p,q}]$. 
\end{lemma}
Here $s_! [E] \in \KR_{\bg} ^{p,q}(Z_{p,q},(f\circ
\rho_{p,q})^{-1}(L^*_\cP)) \simeq \KR_\bg(Z_{p,q},(f\circ
\rho_{p,q})^{-1}(L^*_\cP) )$ where the isomorphism is due to Clifford
periodicity since $(p,q)$ is either $(k,k)$ or $(8k,0)$. The proof of
Lemma~\ref{Realvb2} is a Real twisted analogue of the proof of
\cite[Lemma 3.5]{BOOSW} and amounts to showing that the bundle gerbe
$\KR$-theory classes $[\tilde\rho_{p,q}^{-1}(E)\otimes
\tilde\pi_{p,q}^{-1}({\sf S}_{p,q})]$ and $s_! [E]$ agree in
$\KR_\bg(Z_{p,q},(f\circ \rho_{p,q})^{-1}(L^*_\cP) )$, using Proposition~\ref{prop:Thomiso}. The details are left for the reader.

We define the \emph{assembly map} $\eta\colon \KR^{\bg}_{p,q}(M, L_\cP ) \to \KR_{p,q}(M,\cA)$ by
$$\eta[Z,E,f] = f_*(\mathtt{PD}_Z[E])$$
where $f_*\colon \KR^{p,q}\big(\Gamma_Z(f^{-1}(\cA)) \big) \to
\KR^{p,q}\big(\Gamma_M(\cA) \big)$ is the induced pushforward map.

\begin{proposition} The assembly map $\eta$ is well-defined and functorial. 
\end{proposition} 
\begin{proof} 
Functoriality of $\eta$ follows by the naturality property of the pushforward map in twisted $\KR$-homology. To show that $\eta$ is well-defined, we verify that it respects the three equivalence relations on bundle gerbe $\KR$-cycles.
For the disjoint union/direct sum relation, we have
\begin{align*}\eta([Z,E_1,f]\amalg
[Z,E_2,f])  &= \eta[Z\amalg Z,E_1\amalg E_2,f\amalg f] \\[4pt] &= (f\amalg f)_*(\mathtt{PD}_Z[E_1]\oplus \mathtt{PD}_Z[E_2]) \\[4pt]
&= \eta[Z,E_1,f] + \eta[Z,E_2,f] \\[4pt] &= \eta[Z,E_1\oplus
                                               E_2,f] \ .
\end{align*}
If $[Z_{p,q}, s_![E], f\circ \rho_{p,q}]$  is the Real vector bundle
modification of a bundle gerbe $\KR$-cycle $(Z,E,f)$ using the description in Lemma \ref{Realvb2}, then 
$$ \eta[Z_{p,q}, s_![E], f\circ \rho_{p,q}] = f_* \rho_{p,q*}\mathtt{PD}_{Z_{p,q}}(s_![E]) = f_* \rho_{p,q*}s_*\mathtt{PD}_Z[E] = \eta[Z,E,f]$$ 
where the second equality follows by the commutative diagram
\eqref{gysin} while the last equality is due to $\rho_{p,q}\circ s =
{\rm id}_Z$. Finally, if $(\,\underline{Z}\,,\,\underline{E}\,,\,
\underline{f}\, )$ is any Real spin$^c$ bordism, then we need to show
that $\eta [\partial \underline{Z}\,,\,\underline{E}|_{\partial
  \underline{Z}}\,,\,\underline{f}|_{\partial \underline{Z}}\,
]=0$. By adapting the proof of \cite[Lemma 3.8]{BOOSW} to the Real
twisted setting, it follows that $\partial [\, \accentset{\circ}{\underline{Z}}\, ]
= [\partial \underline{Z}\, ]$ where $\partial \colon \KR_{p,q}(\,
\accentset{\circ}{\underline{Z}} \, ) \to  \KR_{p-1,q}( \partial \underline{Z}\, )
$ is the connecting boundary homomorphism. If $i\colon \partial  \underline{Z} \hookrightarrow  \underline{Z}$ denotes the inclusion, then
$$\eta[\partial \underline{Z}\,,\,\underline{E}|_{\partial
  \underline{Z}}\,,\,\underline{f}|_{\partial \underline{Z}}\, ] =
(f\circ i)_*\mathtt{PD}_{\partial\underline
  Z}[\, \underline{E}|_{\partial \underline{Z}}\, ] = f_*\circ
i_*\circ \partial \big(\mathtt{PD}_{\accentset{\circ}{\underline{Z}}}[\,
\underline{E}\, ] \big) = 0
$$
because $ i_*\circ \partial=0$. 
\end{proof}

We will prove that the assembly map $\eta$ is an isomorphism by
adapting the arguments in \cite{BOOSW} to the Real twisted setting and
constructing an explicit inverse  to $\eta$. For this, we first need a
few preliminary technical results.

\begin{lemma} \label{retract}
Let $M$ be Real compact manifold. Then there exists an equivariant retraction $M \xrightarrow{j} W  \xrightarrow{f}M$ into a Real compact spin$^c$ manifold $W$ of type $(p,q)$. 
\end{lemma} 
\begin{proof} 
Let $V=\RR^{r,s}$ be an $n$-dimensional Real vector space equipped
with the involution given by $e_{r,s}\colon (x,y) \mapsto (x,-y)$ such
that $n=r+s=p+q$ and $(p-q)-(r-s)=0 \ \text{mod}\ 8$. Then $V$ has a
Real spin$^c$ structure of type $(p,q)$ by~\cite[Proposition 3.15]{Fok2015}. 

By the Mostow embedding theorem \cite{Mos}, every Real compact
manifold $M$ has a $\ZZ_2$-equivariant closed embedding into a finite-dimensional real linear $\ZZ_2$-space $V$. By \cite{Jar} there exists
further a $\ZZ_2$-invariant open neighbourhood $U$ with a
$\ZZ_2$-equivariant retraction $f_U\colon U \to M$ onto $M$, that is a
Real compact manifold is a $\ZZ_2$-Euclidean neighbourhood retract. As
shown by~\cite{CdG}, the dimension of $V$ can be chosen to be $3d+2$
or $3d+3$ where $d=\dim(M)$. We may then take $V$ with the $\ZZ_2$-module structure $e_{r,s}$ and a Real spin$^c$ structure of type $(p,q)$ as above. 

Next we proceed as in the proof of~\cite[Lemma 2.1]{BOOSW}. We choose
a $\ZZ_2$-invariant metric $\varrho$ on $U$, define $\phi \colon U \to
\RR_{\geq0}$ by $\phi(m) = \inf_{m'\in M} \, \varrho(m,m'\, )$ to be the distance
to $M$, and fix an approximation to $\phi$ in the chosen metric by a
smooth $\ZZ_2$-invariant function $\psi$. Since $M$ is compact,
$\phi^{-1}[0,a] \subset U$ is compact if $a$ is chosen to be smaller than
the distance from $M$ to the complement $V\setminus U$. For a regular value
$a'\in(0,a)$, the level set $\psi^{-1}(-\infty,a'\, ]\subset U$ is then a
compact Real manifold with boundary and a neighbourhood of $M$. The
double of this space is a Real closed manifold $W$ with a Real
spin$^c$ structure of type $(p,q)$ induced by  $V$, an equivariant
inclusion $j\colon M \to W$ into one of the two copies and an equivariant retraction $f\colon W\to M$ given by the fold map composed with $f_U\colon U \to M$.
\end{proof}

\begin{proposition} \label{equiv-cycles}  Let $M$ be a Real manifold with a Real
bundle gerbe $(P,Y)$ and let $[Z,E,f] \in  \KR^{\bg}_{p,q}(M,P)$. If
the equivariant map $f$ factorises as $Z\xhookrightarrow{h} \tilde Z
\xrightarrow{\tilde f} M$ where  $h$ is an inclusion of Real spin$^c$
manifolds of type $(p,q)$ and $\tilde f$ is an equivariant smooth map, then  
$$[Z,E, f] = [\tilde Z, h_![E], \tilde f\, ] \ . $$
\end{proposition}
\begin{proof} 
The statement is a Real analogue of~\cite[Theorem 4.1]{BOOSW} and the
proof proceeds along similar lines. Let  $\nu = h^{-1} (T\tilde
Z)/TZ$ denote the Real normal bundle of $h$ with the induced Real
spin$^c$ structure of type $(1,1)$. The idea is to construct an
explicit Real spin$^c$ bordism between the Real vector bundle modifications of
$[Z,E, f]$ along $\nu\oplus \RR^{1,1}_Z$ with its canonical Real
spin$^c$-structure  as defined in Lemma \ref{Realsum} and $[\tilde Z,
h_![E], \tilde f\, ]$ along the Real trivial bundle $\RR^{1,1}_{\tilde Z}$. 

The unit sphere bundle of $\RR^{1,1}_{\tilde Z}\oplus\RR_{\tilde Z}$
is simply $\tilde{Z}_{1,1}=\tilde Z\times S^{1,1}$ and its north pole
section is the inclusion $\tilde s\colon \tilde Z \to \tilde Z\times
S^{1,1}$, so the Real vector bundle modification of $[\tilde Z,
h_![E], \tilde f\, ]$ is given by $[\tilde{Z}_{1,1}, \tilde s_!
h_![E], \tilde f\circ \pi_{\tilde Z}]$ where $\pi_{\tilde
  Z}\colon \tilde Z_{1,1}\to\tilde Z$ is the projection. Note that
$\tilde{Z}_{1,1}$ is the boundary of the Real unit disc bundle $\tilde
Z \times D^{2,1}$. By the equivariant tubular neighbourhood theorem,
the normal bundle $\nu$ is $\ZZ_2$-equivariantly diffeomorphic to a
tubular neighbourhood of $Z$. Thus for any $\epsilon \in (0,1)$, the
Real $\epsilon$-disc bundle $D_\epsilon(\nu\oplus \RR^{1,1}_Z\oplus
\RR_Z)$, defined with respect to a $\ZZ_2$-invariant metric on $\nu$,
is contained in $\tilde Z \times D^{2,1}$ and its boundary is the Real
$\epsilon$-sphere bundle $Z_{1,1}^\epsilon$. Recall that if $s\colon Z
\to Z^\epsilon_{1,1}$ is the canonical north pole section and
$\rho_{1,1}$ is the projection, then the Real vector bundle
modification of $[ Z, E,  f]$ is given by  $[Z^\epsilon_{1,1}, s_!
[E], f\circ \rho_{1,1}]$, as the $\epsilon$-scaling of the sphere
bundle $Z_{1,1}$ does not affect the Real vector bundle modification.

Now the space $W = (\tilde Z \times D^{2,1}) \backslash
D_\epsilon(\nu\oplus \RR^{1,1}_Z\oplus \RR_Z)$ obtained by removing
the  Real $\epsilon$-disc bundle is  a Real compact  spin$^c$ manifold
with boundary $\tilde{Z}_{1,1} \amalg( -Z_{1,1}^\epsilon)$. Unlike the
case of~\cite{BOOSW}, the manifold $W$ does not have corners because we are only dealing with closed manifolds $Z$ and $\tilde Z$. The canonical embedding of the cylinder $e\colon Z \times  [\epsilon,1] \to W$, which sends $ [\epsilon,1]$ to the north pole direction $\RR_Z$,  gives rise to the diagram
$$
    \xymatrix{
       &  Z  \ar[dl]_{\tilde s \circ h}
       \ar[d]^{\!\id_Z{\times}\epsilon}_{\id_Z{\times}1\!\!}  \ar[dr]^{s}  &   &   \\
\tilde Z_{1,1}    \ar@{^{(}->}[dr]_{\tilde j}     &         Z   \times  [\epsilon, 1]     \ar[d]^{e}       & Z^\epsilon_{1,1} \ar@{_{(}->}[dl]^{j}  \\
       &  W & & }
$$
of embeddings of Real spin$^c$ compact manifolds, where both the left
and right triangles are pullback diagrams. The Real spin$^c$ bordism
between $[\tilde{Z}_{1,1}, \tilde s_! h_![E], \tilde f\circ
\pi_{\tilde Z}]$ and $[Z^\epsilon_{1,1}, s_! [E], f\circ \rho_{1,1}]$
is then given by $[W, e_! \pi_Z^{-1} (E) , \tilde f\circ \pi_{\tilde
  Z}]$  where $\tilde f\circ \pi_{\tilde Z}$  is the canonical
extension to $W$: Then $W$ has the correct boundary and by
functoriality of the Gysin homomorphism we have
$$ e_! \pi_Z^{-1} [E]\big|_{\tilde{Z}_{1,1} } = \tilde{j}^{-1}e_!
\pi_Z^{-1}  [E] = (\tilde s \circ h)_!(\id_Z\times 1)^{-1}  \pi_Z^{-1}
[E] =  \tilde s_!  h_! [E] \ ,$$
and
$$ e_! \pi_Z^{-1} [E]\big|_{-Z_{1,1}^\epsilon} = j^{-1} e_! \pi_Z^{-1}
[E]  = s_! (\id_Z\times \epsilon)^{-1}  \pi_Z^{-1}  [E] = s_! [E]  \ . $$
As remarked in \cite{BOOSW}, the restriction of $\tilde f\circ \pi_{\tilde Z}$ to $Z^\epsilon_{1,1}$ is only homotopic to $f\circ \rho_{1,1}$, but it is possible to modify the map $\tilde f\circ \pi_{\tilde Z}$ by scaling the radius of $S^{1,1}$ in order to achieve a true bordism. 
\end{proof}

\begin{corollary} \label{cycles} 
Let $M$ be a Real compact spin$^c$ manifold of type $(p,q)$ with a
Real bundle gerbe $(P,Y)$ and let  $[Z,E,f] \in
\KR^{\bg}_{p,q}(M,P)$. Then
$$[Z,E,f] = [M, f_! [E], \id_M] \ . $$
\end{corollary}
\begin{proof} 
Choose a Real equivariant embedding $j\colon Z\to V$ into a finite-dimensional Real vector space $V$ with a Real spin$^c$ structure of type $(p,q)$. The map $j$ is $\ZZ_2$-homotopic to the Real constant  map $c\colon Z\to V$ with value $0$ via a linear homotopy, and this extends to the one-point compactification $V^+ \simeq S^{p, q}$ by composition with the inclusion map $V\hookrightarrow V^+$. Thus we obtain  embeddings of Real compact spin$^c$ manifolds
$$
    \xymatrix{
        Z \ar@{^{(}->}[dr]^{(f,j)}  &   &   \\
                               &\ \ \ \  M \times V^+ \ar[r]^{\hspace{0.7cm}\pi_M} & M \\
        M\ar@{^{(}->}[ur]_{(\id_M,c)} & & }
$$
with $\pi_M \circ (f,j) = f$, $\pi_M \circ (\id_M,c) = \id_M$ and
where $(f,j)$ is equivariantly homotopic to $(f,c) = (\id_M, c) \circ
f$. The result then follows by
\begin{align*}
[Z,E,f] &= [Z,E, \pi_M \circ (f,j)] \\[4pt] 
&= [M\times V^+, (f,j)_! [E], \pi_M] \\[4pt]  
&= [M\times V^+, (f,c)_! [E], \pi_M] \\[4pt]
&= [M\times V^+, (\id_M, c)_!  f_! [E], \pi_M] \\[4pt] 
&= [M,  f_! [E], \pi_M \circ (\id_M, c)] \\[4pt] 
&= [M, f_! [E], \id_M] \ ,
\end{align*}
where we have applied Proposition \ref{equiv-cycles} at the second and fifth equality, and the functoriality of the Gysin homomorphism at the third and fourth equality.
\end{proof}

We can now use Lemma~\ref{retract} to define a group homomorphism $\beta_W\colon \KR^{p,q}(\Gamma_M(\cA)) \to  \KR^{\bg}_{p,q}(M,L_\cP)$ by 
$$
\beta_W(x) = [W, \mathtt{PD}^{-1}_W j_*(x), f] \ , 
$$
where $\mathtt{PD}^{-1}_W j_*(x) \in \KR_{\bg}(W, f^{-1}(L^*_\cP))$ using
$j^{-1}(f^{-1} (\cA) ) = \cA$ and the natural isomorphism in Proposition \ref{prop:K-isom2}. It follows by Remark \ref{RealSS} and Lemma
\ref{Kclass} that this is a well-defined bundle gerbe $\KR$-cycle for
any fixed retract $W$.

\begin{theorem}\label{assembly}
Let $M$ be a Real compact manifold with a Real bundle gerbe $(P,Y)$. Then the assembly map $\eta\colon \KR^{\bg}_{p,q}(M,L_\cP) \to \KR_{p,q}(M,\cA)$ is an isomorphism of abelian groups. 
\end{theorem}
\begin{proof} First let us assume that $M$ has a Real spin$^c$ structure of type $(p,q)$ and let $\beta_M$ be the homomorphism corresponding to $W=M$ and $j=f=\id_M$ in Lemma \ref{retract}. Then we have 
$$(\beta_M \circ \eta )[Z,E,f] = \beta_M ( f_*(\mathtt{PD}_Z[E])) =
\beta_M (\mathtt{PD}_M(f_![E])) =  [M, f_![E], \id_M ] = [Z, E, f]$$
where the second equality follows by the commutative diagram \eqref{gysin} and Proposition \ref{prop:K-isom2} understood, and the last equality follows by Corollary~\ref{cycles}. This implies that $\eta$ is injective with right inverse $\beta_M$. 

On the other hand, for any choice of retract $W$ we have
$$ (\eta \circ  \beta_W) (x) = \eta[W, \mathtt{PD}^{-1}_W j_*(x), f] =  f_*(\mathtt{PD}_W \mathtt{PD}^{-1}_W j_*(x))= f_*  j_*(x) = x$$
which implies that $\eta$ is surjective with left inverse $\beta_W$. Consequently $\eta$ is a bijection and $\beta_M = \beta_W = \eta^{-1}$ by uniqueness of inverses. In particular, it follows that $\beta_W$ is independent of the choice of retract. 

For an arbitrary Real compact manifold $M$, the surjectivity argument
still applies and it suffices to show that $\beta_W \circ \eta =
\id_{\KR^{\bg}_{p,q}(M,P)}$ for any  choice of retract $W$. For any $
[Z,E, \tilde f\, ]\in \KR^{\bg}_{p,q}(M,P)$ we have 
$$(\eta \circ j_*  \circ\beta_W  \circ \eta) [Z,E, \tilde f\, ] =  j_*
f_*  \mathtt{PD}_W  \mathtt{PD}_W^{-1} j_*\tilde f_*\mathtt{PD}_Z[E] =
(j_* \circ \eta) [Z,E, \tilde f\, ] = (\eta\circ  j_*) [Z,E, \tilde f\,
] $$
where the last equality follows by functoriality of $\eta$. Since
$j_*$ is (split) injective by naturality and $\eta$ is an isomorphism
on $\KR^{\bg}_{p,q}(W,f^{-1}(P) )$ where $j_*$ and $ j_*  \circ
\beta_W \circ \eta$ take values, we conclude that $\beta_W \circ \eta
= \id_{\KR^{\bg}_{p,q}(M,P)} $.
\end{proof}

\begin{remark} 
The $\ZZ_2$-Euclidean neighbourhood retraction property, and hence Lemma~\ref{retract}, holds more generally for any Real finite CW-complex $M$~\cite{Jar}. The proof of Theorem \ref{assembly}  therefore applies verbatum to Real finite CW-complexes, if we realise twistings on $M$ by Real Dixmier--Douady bundles $\cA$ and twisted $\KR$-homology by cycles $[Z,\sigma, f]$ with $\sigma \in \KR(Z,f^{-1}(\cA^{\op}))$. The equivalence relation on these $\KR$-cycles is generated by Real spin$^c$ bordism and Real vector bundle modification formulated in terms of the Gysin homomorphism as in Lemma~\ref{Realvb2}.
\end{remark}

\begin{example}[$K$-theoretic Ramond-Ramond charge]\label{ex:RRcharge}
Let us work in the setting of Corollary~\ref{cycles}. In this case the fundamental $\KR$-homology class $[M]$ of the manifold $M$ can be taken to lie in the untwisted group $\KR_{p,q}(M)$, and under the assembly map $\eta$ it arises from
$$
\eta[M,\CC_M,\id_M] = (\id_M)_*(\mathtt{PD}[\CC_M]) =[\CC_M] \cap [M]=[M] \ .
$$
Since $\eta$ defines a natural equivalence between the functors $\KR_*^{\bg}$ and $\KR_*$, it follows that the bundle gerbe $\KR$-cycle $[M,\CC_M,\id_M]$ can be identified as the fundamental class of $M$ in $\KR^\bg_{p,q}(M,[H])$. Moreover, by Poincar\'e duality and Proposition \ref{prop:K-isom2}, every class in the twisted $\KR$-theory $\KR(M,-[H])$ is represented by a bundle gerbe $\KR$-cycle.

Now let $(Z,E,f)$ be a Real bundle gerbe D-brane in the sense of Definition~\ref{def:RealDbrane}. Then these considerations together with Corollary~\ref{cycles} give a twisted $\KR$-theory definition of the Ramond-Ramond charge of such a D-brane as the canonical element
$$
f_![E] \ \in \ \KR(M,-[H]) \ .
$$
This formula generalises the special case where $Z\subseteq M^\tau$ coincides with an O$^+$-plane and $f\colon Z\hookrightarrow M$ is the inclusion, with $[E]\in \KO_\bg(Z,-f^*[H])$. Moreover, the charges of (generalised) Real bundle gerbe D-branes are classified by the twisted $\KR$-theory $\KR(M,-[H])$ of spacetime; the model \eqref{eq:tachyondef} for $\KR(M,-[H])$ then nicely makes contact with the tachyon field picture of $K$-theory charges~\cite{Witten} (cf.\ also~\cite{GaoHori}).
\end{example}

\end{document}